\documentclass{article}
\usepackage{arxiv}
\usepackage{algorithm}
\usepackage{caption}
\usepackage{subcaption}
\usepackage[noend]{algpseudocode}
\usepackage[utf8]{inputenc}
\usepackage{hyperref}
\usepackage{url}
\usepackage{booktabs}
\usepackage{amsfonts,amsthm,amsmath,amssymb,mathtools}
\usepackage{nicefrac}
\usepackage{microtype}
\usepackage{graphicx}
\usepackage{float}
\usepackage{natbib}
\newtheorem{theorem}{\bf Theorem}
\newtheorem{proposition}[theorem]{\bf Proposition}
\newtheorem{lemma}[theorem]{\bf Lemma}
\newtheorem{corollary}[theorem]{\bf Corollary}
\newtheorem{definition}{Definition}
\theoremstyle{remark}

\newtheorem{assumption}{Assumption}

\usepackage{algorithm}
\usepackage[noend]{algpseudocode}
\usepackage{bm}
\usepackage{bbm}
\usepackage{caption}
\usepackage{subcaption}
\usepackage{tikz}
\usetikzlibrary{arrows,backgrounds,matrix,positioning,shapes.geometric,
shapes.misc,calc,patterns}
\usepackage{placeins} \usepackage[many]{tcolorbox}
\usepackage{capt-of} \definecolor{boxbg}{cmyk}{.04,.04,.12,.08} \newtcolorbox[auto counter]{examplebox}[2][]{
  title=Example~\thetcbcounter:~#2,#1,
  boxrule = 0pt,   titlerule = 0.5pt,
  colback = boxbg, colbacktitle = boxbg,
  coltitle = black,
  fonttitle=\bfseries,
  lower separated=false,
}
\usepackage[]{xcolor}
\definecolor{darkgreen}{RGB}{97, 151, 152}
\definecolor{changecolor}{RGB}{192,64,0}
\definecolor{lavender}{rgb}{0.75,0.58,0.89}

\newcommand{\com}[1]{} \usepackage[normalem]{ulem}

\DeclareMathOperator{\rk}{rank}
\DeclareMathOperator{\E}{E}
\DeclareMathOperator{\var}{var}
\DeclareMathOperator{\cov}{cov}

\DeclareMathOperator{\tw}{tw}

\DeclareMathOperator{\pa}{pa}

\DeclareMathOperator{\loglik}{LL}
\DeclareMathOperator{\elbo}{ELBO}
\newcommand{\set}[1]{\mathopen{\{}#1\mathclose{\}}}
\newcommand{\gradient}[1]{\nabla_{#1}}
\newcommand{\partialDer}[2]{\frac{\partial #2}{\partial #1}}

\newcommand{\cond}[2]{\left.#1\mathrel{}\middle|\mathrel{}#2\right.}
\newcommand{\transpose}[1]{#1^{\top}}
\newcommand{\inverse}[1]{#1^{-1}}
\DeclareMathOperator{\mtovh}{vech} \newcommand{\branchLength}[1]{\ell(#1)}
\newcommand{\vect}[1]{#1}

\renewcommand{\headeright}{}
\renewcommand{\undertitle}{}
\title{
Leveraging graphical model techniques
to study evolution on phylogenetic networks
}

\author{
Benjamin Teo\\
Department of Statistics\\University of Wisconsin-Madison
\And
Paul Bastide\\
IMAG, Universit\'e de Montpellier,\\CNRS
\And
C\'ecile An\'e\\
Departments of Statistics and of Botany\\University of Wisconsin-Madison
}
\date{}

\begin{document}
\maketitle

\begin{abstract}
The evolution of molecular and phenotypic traits is commonly
modelled using Markov processes along a phylogeny.
This phylogeny can be
a tree, or a network if it includes reticulations,
representing events such as hybridization or admixture.
Computing the likelihood of data observed at the leaves
is costly as the size and complexity of the phylogeny grows.
Efficient algorithms exist for trees, but cannot be applied to
networks.
We show that a vast array of models for trait evolution along
phylogenetic networks can be reformulated as graphical models,
for which efficient belief propagation algorithms exist.
We provide a brief review of belief propagation on general graphical models,
then focus on linear Gaussian models for continuous traits.
We show how belief propagation techniques can be applied
for exact or approximate (but more scalable) likelihood
and gradient calculations,
and prove novel results for efficient parameter inference of some models.
We highlight the
possible fruitful interactions between graphical models and phylogenetic
methods. For example, approximate likelihood approaches
have the potential to greatly reduce
computational costs for phylogenies with reticulations.
\end{abstract}


\setcounter{tocdepth}{2}
\tableofcontents
\newpage

\section{Introduction}
Stochastic processes are used to model the
evolution of traits over time along a phylogeny,
a graph representing the historical relationships between species,
populations or individuals of interest, in which
internal nodes represent divergence (e.g.\ speciation) or
merging (e.g.\ introgression) events.
In this work, we consider traits that may be multivariate,
discrete and/or continuous, with a focus on continuous traits.
Trait evolution models are used to infer evolutionary dynamics
\cite{Pybus2012,Dellicour2020}
and historical correlation between traits \cite{Cybis2015,Zhang2021,Zhang2023},
predict unobserved traits at ancestral nodes or extant leaves \cite{Revell2014,Lartillot2014},
or estimate phylogenies from rich data sets \cite{2015yu,2016SolisLemus-SNaQ,
2022Neureiter,2023Maier,2023Nielsen-admixturebayes}.

Calculating the likelihood is no easy task
because the traits at ancestral nodes are unobserved and need to be integrated out.
This problem is very well studied for phylogenetic trees, with efficient
solutions for both discrete and continuous traits
\cite{felsenstein81,fitzjohn2012diversitree,freckleton2012fast,2014HoAne-algo,
2017Goolsby-Rphylopars,mitov2020fast}.
Admixture graphs and phylogenetic networks with reticulations are now gaining
traction due to growing empirical evidence for gene flow, hybridization and
admixture \cite{moran2021genomic,kong2022classes,2023Maier}.
Yet many methods for these networks could be improved with
more efficient likelihood calculations.

The vast majority of phylogenetic models make a
Markov assumption, in that the trait distribution at all nodes
can be expressed by a set of local models.
At the root, this model describes the prior distribution of the ancestral trait.
For each node in the phylogeny, a local transition model describes
the trait distribution at this node conditional on the trait(s) at its
parent node(s).
As each local model can be specified individually with its own set of parameters,
the overall evolutionary model can be very flexible,
including possible shifts in rates, constraints, and mode of evolution across
different clades \cite{Clavel2015,2024Bartoszek}.
Other models do not make a Markov assumption, such as 
threshold models \cite{2005Felsenstein,Goldberg2020}, or
models that
combine a backwards-in-time coalescent process for gene trees
and forward-in-time mutation process along gene trees \cite{2012Bryant-snapp}.
We show here
that some of these models can still be expressed as a product of
local conditional distributions, over a graph that is more complex than
the initial phylogeny.

These evolutionary models are special cases of graphical models,
also known as Bayesian networks, which have been heavily studied \cite{koller2009probabilistic}.
The likelihood calculation task
has received a lot
of attention, including algorithms for efficient approximations when the network
is too complex to calculate the likelihood exactly \cite{koller2009probabilistic}.
Another well-studied task
is that of predicting the state of unobserved variables
(ancestral states in phylogenetics) conditional on the observed data.
We argue here that the field of phylogenetics could greatly benefit
from applying and expanding knowledge from graphical models for the
study and use of phylogenetic networks.

In section~\ref{sec:likelihoodcomplexity}
we review the challenge brought by phylogenetic models in which
only tip data are observed, and current techniques for
efficient likelihood calculations.
In section~\ref{sec:cts trait evolution}
we focus on general Gaussian models for the evolution of a
continuous trait, possibly multivariate to capture evolutionary correlations
between traits.
On reticulate phylogenies, these models need to describe the trait of admixed populations conditional on their parental populations.
Turning to graphical models in section~\ref{sec:review-GM-BP},
we describe their general formulation
and show that many phylogenetic models can be expressed as special cases,
from known examples to less obvious examples
(using the coalescent process on species trees, or species networks).
We then provide a short review of belief propagation, a core technique
to perform inference on graphical models, first in its general form
and then specialized for continuous traits in linear Gaussian models.
In section~\ref{sec:loopybp}
we describe loopy belief propagation, a technique to perform
approximate inference in graphical models, when exact inference does not
scale. As far as we know, loopy belief propagation has never been used in phylogenetics.
Section~\ref{sec:inferencewithBP} describes leveraging BP for parameter inference:
fast calculations of the likelihood and its gradient
can be used in any likelihood-based framework, frequentist or Bayesian.
Finally, section~\ref{sec:discussion} discusses future challenges
for the application and extension of graphical model techniques
in phylogenetics.
These techniques offer a range of avenues to expand the
phylogeneticists' toolbox for fitting evolutionary models on phylogenetic
networks, from approximate inference methods that are more scalable, to
algorithms for fast gradient computation for better parameter inference.

\section{Complexity of the phylogenetic likelihood calculation}
\label{sec:likelihoodcomplexity}

\subsection{The pruning algorithm}

Felsenstein's pruning algorithm \cite{felsenstein1973maximum,felsenstein81}
launched the era of model-based phylogenetic inference, now rich with
complex models to account for a large array of biological processes:
including DNA and protein substitution models, variation of their
substitution rates across genomic loci, lineages and time,
and evolutionary models for continuous traits and geographic distributions.
The pruning algorithm gave the key to calculate the likelihood
of these models along a phylogenetic tree, in a practically feasible way.
The basis of this algorithm, which extends to tasks beyond likelihood
calculation, was discovered in other areas and given other names, such as
the sum-product algorithm, forward-backward algorithm,
message passing, and belief propagation (BP).
In particular, the field of statistical human genetics saw the 
early development of such algorithms for models on pedigrees \cite{1970Hilden,1971ElstonStewart,1972HeuchLi}
including loops \cite{1978CanningsThompsonSkolnick,2002Lange}.
See \cite[section 9.8]{koller2009probabilistic} for a review
of the literature on variable elimination algorithms.

The pruning algorithm, which is a form of BP,
computes the the full likelihood of all the observed taxa
by traversing the phylogenetic tree once,
taking advantage of the Markov property: where the evolution
of the trait of interest along a daughter lineage is independent of
its past evolution, given knowledge of the parent's state.
The idea is to traverse the tree and calculate the likelihood of the descendant
leaves of an ancestral species conditional on its state, from similar
likelihoods calculated for each of its children.
If the trait is discrete with 4 states
for example (as for DNA), then this entails keeping track of 4 likelihood values at each ancestral species.
If the trait is continuous with a Gaussian distribution, e.g.\ from a
Brownian motion (BM) or an Ornstein-Uhlenbeck (OU) process \cite{1997Hansen},
then the likelihood at an ancestral species is a nice function of its state that
can be concisely parametrized by quantities akin to the
posterior mean and variance conditional on descendant leaves.
Felsenstein's independent contrasts (IC) \cite{felsenstein85}
also captures these partial posterior quantities
and can be viewed as a special implementation of BP for likelihood calculation.

BP is used ubiquitously for the analysis of discrete traits, such as
for DNA substitution models
(e.g.\ in \texttt{RAxML} \cite{2014Stamatakis-raxmlv8},
\texttt{IQ-TREE} \cite{2015Nguyen-iqtree},
\texttt{MrBayes} \cite{2003Ronquist-mrbayes3})
or for discrete morphological traits in comparative methods
(e.g.\ in \texttt{phytools} \cite{2012Revell-phytools},
\texttt{BayesTraits} \cite{2004PagelMeadeBarker},
\texttt{corHMM} \cite{2021Boyko-corHMM,2023Boyko-hOUwie},
\texttt{RevBayes} \cite{hohna2016revbayes}).
For discrete traits, there is simply no feasible alternative.
On a tree with 20 taxa and 19 ancestral species,
the naive calculation of the likelihood at a given DNA site
would require the calculation and summation of $4^{19}$ or 274 billion
likelihoods, one for each nucleotide assignment at the 19 ancestral species.
This calculation would need to be repeated for each site in the alignment,
then repeated all over during the search for a well-fitting phylogenetic tree.

\subsection{Continuous traits on trees: the lazy way}

For continuous traits under a Gaussian model (including the Brownian motion),
BP is not used as ubiquitously because a multivariate Gaussian distribution
can be nicely captured by its mean and covariance matrix:
the multivariate Gaussian formula can serve as an alternative.
For example, for one trait $Y$ with ancestral state $\mu$ at the root of the
phylogeny, the phylogenetic covariance $\boldsymbol{\Sigma}$
between
the taxa at the leaves can be obtained from the branch lengths in the tree.
Under a BM, the covariance $\cov(Y_i,Y_j)$ between taxa $i$ and $j$ is
$\Sigma_{ij}=\sigma^2 t_{ij}$ where $t_{ij}$
is the length between the root and their most recent common ancestor.
The likelihood of the observed traits at the $n$ leaves can then
be calculated using matrix and vector multiplication techniques as
\begin{equation}\label{eq:gaussianlik}
    (2\pi)^{-n/2}\det|\boldsymbol{\Sigma}|^{-1/2}
  \exp\left(-\frac{1}{2}(\vect{Y}-\mu)^{\top}\boldsymbol{\Sigma}^{-1}(\vect{Y}-\mu)\right).
\end{equation}
This alternative to BP has the disadvantage of requiring the inversion
of the covariance matrix $\boldsymbol{\Sigma}$, a task whose computing time
typically grows as $m^3$ for a matrix of size $m\times m$.
It also has the disadvantage that $\boldsymbol{\Sigma}$ needs to be calculated
and stored in memory in the first place.
For multivariate observations of $p$ traits on each of $n$ taxa, the covariance
matrix has size $m=pn$ so the typical calculation cost of \eqref{eq:gaussianlik}
is then $O(p^3n^3)$, which can quickly become very large.
For example, with only 30 taxa and 10 traits,
$\boldsymbol{\Sigma}$ is a $300\times 300$-matrix.
Studies with large $p$ and/or large $n$ are now frequent,
especially from geometric morphometric data with $p$ over 100 typically
(e.g.\ \cite{2023Hedrick})
or with expression data on $p>1,000$ genes easily, that also require
more complex models to account for variation (e.g.\ within species,
between organs, between batches)
\cite{2013Dunn,2019Shafer}. Studies with a large number $n$ of taxa are now frequent
(e.g.\ $n>5,000$ in birds and mammals \cite{2012Jetz,2019Upham}) and
virus phylogenies can be massive
(e.g.\ $n>500,000$ SARS-CoV-2 strains \cite{2023DeMaio}).
Viral continuous traits previously studied include
virulence traits
(e.g.\ $n>1,000$ and $p=3$ traits in HIV \cite{Blanquart2017,2022Hassler-missingdata}) and geographic data for phylodynamics
(e.g.\ $n=801$ and $p=2$ continuous coordinates to describe the spread of the West Nile Virus \cite{Dellicour2020}). 

In these cases with large data size $np$,
the matrix-based alternative to BP is prone to numerical
inaccuracy and numerical instability
in addition to the increased computational time,
because it is hard to accurately invert a large matrix.
Even when the matrix is of moderate size, numerical inaccuracy
can arise when the matrix is ``ill-conditioned''.
These problems were identified under OU models on phylogenetic trees
that have closely-related sister taxa,
or under early-burst (EB) models with strong morphological diversification
early on during the group radiation, and much slowed-down evolution later on
\cite{2017AdamsCollyer,2020JhwuengOMeara,2023Bartoszek-modelselection-mvOU}.

For some simple models, the large $np \times np$ covariance matrix
can be decomposed as a Kronecker product of a $p\times p$ trait covariance
and a $n \times n$ phylogenetic covariance.
This decomposition can
simplify the complexity of calculating the likelihood.
However, this decomposition is not available under many models, such as the
multivariate Brownian motion with shifts in the evolutionary rates
(e.g.\ \cite{2018CaetanoHarmon}) or the
multivariate Ornstein-Uhlenbeck model with non-scalar rate or selection matrices
\cite{Bartoszek2012,Clavel2015}.

\subsection{BP for continuous traits on trees} \label{sec:tree_BP_continuous}

To bypass the complexity of matrix inversion, Felsenstein pioneered IC to test
for phylogenetic correlation between traits, assuming a BM
model on a tree \cite{felsenstein85}.
Many authors then used BP approaches to
handle Gaussian models beyond the BM
\cite{fitzjohn2012diversitree,freckleton2012fast,2017Goolsby-Rphylopars}.
Notably, \cite{2014HoAne-algo} describe a fast algorithm that can be used
for non-Gaussian models as well.
Most recently, \cite{mitov2020fast} highlighted that BP can be applied to
a large class of Gaussian models: including the BM and the OU process
with shifts and variation of rates and selection regimes across branches.
Software packages that use these fast BP algorithms include
\texttt{phylolm} \cite{2014HoAne-algo},
\texttt{Rphylopars} \cite{2017Goolsby-Rphylopars},
\texttt{BEAST} \cite{2023Hassler-review}
or the most recent versions of
\texttt{hOUwie} \cite{2023Boyko-hOUwie} and
\texttt{mvSLOUCH} \cite{2023Bartoszek-modelselection-mvOU}.

All the methods cited above only use the first post-order tree traversal of BP to compute the
likelihood. A second preorder traversal allows, in the Gaussian case, for the computation of
the distribution of all internal nodes conditionally on the model and on the traits values at
the tips.
These distributions can then be used for, e.g., ancestral state reconstruction \cite{Lartillot2014},
expectation-maximization algorithms for shift detection in the optimal values of an OU
\cite{Bastide2017},
or the computation of the gradient of the likelihood in the BM
\cite{Zhang2021,Fisher2021} or general Gaussian model \cite{2021Bastide-HMC}.
Such BP techniques have also been used for taking gradients of the likelihood with respect to branch
lengths in sequence evolution models \cite{1996FelsensteinChurchill,Ji2019}
or for phylogenetic factor analysis \cite{Tolkoff2018,2022Hassler-PFA}.

\subsection{From trees to networks}

So far, Felsenstein's pruning algorithm and related BP approaches
have been restricted to phylogenetic \emph{trees}, mostly.
There is now ample evidence that reticulation is ubiquitous in
all domains of life from biological processes such as lateral gene transfer,
hybridization, introgression and gene flow between populations.
Networks are recognized to be better than trees for representing
the phylogenetic history of species and populations in many groups.
Although current studies using networks have few taxa,
typically between 10-20 (e.g.\ \cite{2023Nielsen-admixturebayes}),
they tend to have increasingly more tips as network inference methods
become more scalable
(e.g.\ $n=39$ languages in \cite{2022Neureiter}).
As viruses are known to be affected by recombination,
we also expect future virus studies to use large network phylogenies
\cite{2022Ignatieva-sarscov2}, so that BP will become essential for network studies too.
In this work, we describe approaches currently used for
trait evolution on phylogenetic networks.
We argue that the field of evolutionary biology would benefit from
applying BP approaches on networks more systematically.
Transferring knowledge from the mature and rich literature on BP
would advance evolutionary biology research when phylogenetic networks are used.

\subsection{Current network approaches for discrete traits}

For discrete traits on general networks, very few approaches use BP techniques
as far as we know.
For DNA data for example,
\texttt{PhyLiNC} \cite{allen2020estimating} and
\texttt{NetRAX} \cite{lutteropp2022netrax} extend the typical tree-based model
to general networks, assuming no incomplete lineage sorting. That is, each
site is assumed to evolve along one of the trees displayed in the network,
chosen according to inheritance probabilities at reticulate edges.
\texttt{PhyLiNC} assumes independent (unlinked) sites.
\texttt{NetRAX} assumes independent loci, which may have a single site each.
Each locus may have its own set of branch lengths and substitution model
parameters. Both methods calculate the likelihood of a network $N$ via extracting
its displayed trees and then applying BP on each tree.
Similarly, comparative methods for binary and multi-state traits  implemented in
\texttt{PhyloNetworks} also extract displayed trees and then apply BP on each
displayed tree \cite{2020karimi}.
While these approaches use BP on each displayed tree, a network with
$h$ reticulations can have up to $2^h$ displayed trees. This leads to a
computational bottleneck when the number of reticulations increases.

BP approaches have also been used for models with incomplete lineage sorting,
modelled by the coalescent \cite{1982kingman}.
Notably, \texttt{SNAPP} models the evolution of unlinked biallelic markers
along a species tree, accounting for incomplete lineage sorting \cite{2012Bryant-snapp}.
This method was recently made faster with \texttt{SNAPPER} \cite{2020Stoltz-snapper}
and extended to phylogenetic networks with \texttt{SnappNet} \cite{2021Rabier-snappnet}.
The coalescent process introduces the challenge that each site
may evolve along \emph{any} tree, depending on past coalescent events.
\texttt{SNAPP} introduced a way to bypass the difficulties of handling
coalescent histories and hence decrease computation time.
After we describe BP for general graphical models, we recast this
innovation as BP on a graphical model formulation of the problem.

BP was also used to calculate the likelihood of the
joint sample frequency spectrum (SFS). To account for incomplete lineage sorting
on a tree, \cite{kamm2017momi} use the continuous-time Moran model to
reduce computational complexity, and assume that each site undergoes at most
one mutation.
In \texttt{momi2},
\cite{kamm2020momi2} extend the approach to phylogenetic networks by assuming
a pulse of admixture at reticulations. The associated graphical model
is much simpler than that required by \texttt{SNAPP} or \texttt{SnappNet}
thanks to the assumption of no recurrent mutation.

\subsection{Current network approaches for continuous traits}

Compared to the rich toolkit available for the analysis of continuous traits
on trees, the toolkit for phylogenetic networks is still limited.
\texttt{PhyloNetworks} includes comparative methods on networks
\cite{2017phylonetworks}, implemented in \texttt{Julia} \cite{2017julia}.
These methods extend phylogenetic ANOVA to networks, for a
continuous response trait predicted by any number of
continuous or categorical traits, with residual variation being
phylogenetically correlated. So far, the models available in
\texttt{PhyloNetworks} include the BM, Pagel's $\lambda$, possible
within-species variation, and shifts at reticulations to model transgressive
evolution \cite{2018Bastide-pcm-net,2023Teo-wsv}.
However, all calculations are based on working with the full covariance
matrix, without BP.
\texttt{TreeMix} \cite{2012PickrellPritchard},
\texttt{ADMIXTOOLS} \cite{2012Patterson,2023Maier},
\texttt{poolfstat} \cite{Gautier2022} and
\texttt{AdmixtureBayes} \cite{2023Nielsen-admixturebayes} use allele frequency as a
continuous trait. They model its evolution along a network,
or admixture graph, using a Gaussian model in which the evolutionary
rate variance is affected by the ancestral allele frequency
\cite{2019SoraggiWiuf,2020Lipson}.
Again, these methods work with
the phylogenetic covariance matrix, rather than BP approaches.
They also consider subsets of up to 4 taxa at a time
via $f_2$, $f_3$ and $f_4$ statistics, which simplifies
the likelihood calculation.
To identify selection and adaptation on a network,
\texttt{PolyGraph} \cite{2018Racimo-polygraph} and
\texttt{GRoSS} \cite{2019RefoyoMartinez-GRoSS}
assume a similar model and use the full covariance matrix.
In summary, BP has yet to be used for continuous trait evolution
on networks.

\section{Continuous trait evolution on a phylogenetic network}
\label{sec:cts trait evolution}

We now present phylogenetic models
for the evolution of continuous traits, to which we apply BP later.
We generalize the framework in \cite{mitov2020fast} and
\cite{2021Bastide-HMC} from trees to networks,
and we extend the network model in \cite{2018Bastide-pcm-net} from the BM
to more general evolutionary models.
We consider a multivariate $X$ consisting of $p$ continuous traits,
and model their correlation over time.
Our model ignores the potential effects of incomplete lineage sorting on $X$,
a reasonable assumption for highly polygenic traits.

\subsection{Linear Gaussian models}

Most random processes used to model continuous trait evolution on a
phylogenetic tree are extensions of the BM to capture processes such as evolutionary trends,
adaptation, and variation in rates across lineages for example.
In its most general form, the linear Gaussian evolutionary model
on a tree
(referred to as the GLInv family in \cite{mitov2020fast})
assumes that the
trait $X_v$ at node $v$ has the following distribution conditional
on its parent $\pa(v)$
\begin{equation}\label{eq:lineargaussian}
  X_v\mid X_{\pa(v)} \sim \mathcal{N}(\bm{q}_v X_{\pa(v)} + {\omega}_v, \bm{V}_v)
\end{equation}
where the actualization matrix $\bm{q}_v$,
the trend vector ${\omega}_v$ and the covariance matrix $\bm{V}_v$
are appropriately sized and do not depend on trait values $X_{\pa(v)}$.
When the tree is replaced by a network, 
a node $v$ can have multiple parents $\pa(v)$. 
In this case, we can write $X_{\pa(v)}$ as
the vector formed by stacking the elements of
$\{X_u\mid u\in\pa(v)\}$ vertically, with length equal to
the number of traits times the number of parents of $v$.
In the following, we show that \eqref{eq:lineargaussian},
already used on trees, can easily be extended to networks,
to describe both evolutionary models along one lineage
and a merging rule at reticulation events.

\subsection{Evolutionary models along one lineage}

For a tree node $v$ with parent node $u$, we need to describe
the evolutionary process along one lineage,
graphically modelled by the tree edge $e=(u,v)$.
It is well known that a wide range of evolutionary models
can fit in the general form \eqref{eq:lineargaussian}
\cite{mitov2020fast,2021Bastide-HMC}.
For instance, 
the BM with variance rate $\bm{\Sigma}$
(a variance-covariance matrix for a multivariate trait)
is described by \eqref{eq:lineargaussian} where
$\bm{q}_v$ is the $p\times p$ identity matrix $\bm{I}_p$,
there is no trend ${\omega}_v = \bm{0}$,
and the variance is proportional to the edge length $\ell(e)$:
$\bm{V}_v = \ell(e)\bm{\Sigma}$.

Allowing for rate variation amounts to letting the variance rate vary
across edges $\bm{\Sigma}=\bm{\Sigma}(e)$.
For example, the Early Burst (EB) model assumes that the variance rate at any
given point in the phylogeny depends on the time $t$ from the root to that point, as:
\[\bm{\Sigma}(t) = \bm{\Sigma}_0 e^{bt}\,.\]
For this $t$ to be well-defined on a reticulate network,
the network needs to be time-consistent
(distinct paths from the root to a node all share the same length).
The rate $b$ is a rate of variance decay if it is negative,
to expected during adaptive radiations, with a burst of variation
near the root (hence Early Burst) before a slow-down of trait evolution
\cite{2010Harmon}.
When $b>0$, this model is called ``accelerating rate'' (AC)
\cite{2003BlombergGarlandIves}.
\cite{2017ClavelMorlon} used a flexible extension of this model (on a tree),
replacing $t$ by one or more covariates that are known functions of time,
such as the average global temperature and other environmental variables:
\[\bm{\Sigma}(t) = \tilde{\sigma}(t, T_1(t),\cdots,T_k(t))\,.\]
Then, the variance accumulated along edge $e=(u,v)$ is given by
\[\bm{V}_v = \int_{t(u)}^{t(v)}\bm{\Sigma}(t) dt\;.\]
In the particular case of the EB model, we get
\[\bm{V}_v = \bm{\Sigma}_0 \, e^{bt(u)}(e^{b \ell(e)}-1)/b\;.\]

Allowing for shifts in the trait value, perhaps due to jumps or
cladogenesis, amounts to including ${\omega}_v\neq 0$ for some $v$.

Adaptive evolution is typically modelled by the OU process, which
includes a parameter $\bm{A}_e$ for the strength of selection along edge $e$.
This selection strength is often assumed constant across edges, and is
typically denoted as $\alpha$ for a univariate trait.
The OU process also includes a primary optimum value $\bm{\theta}_e$,
which may vary across edges when we are interested in detecting shifts in
the adaptive regime across the phylogeny.
Under the OU model, the trait evolves along edge $e$ with random drift
and a tendency towards $\bm{\theta}_e$:
\[
 dX^{(e)}(t) = \bm{A}_e(\bm{\theta}_e -X^{(e)}(t)) dt + \bm{R}_e dB(t)
\]
where $B$ is a standard BM and the drift variance is
$\bm{\Sigma}_e = \bm{R}_e\bm{R}_e^{\top}$.
Then, conditional on the starting value at the start of $e$,
the end value $X_v$ is linear Gaussian as in \eqref{eq:lineargaussian}
with actualization $\bm{q}_v = e^{-\ell(e) \bm{A}_e}$,
trend ${\omega}_v = (\bm{I} - e^{-\ell(e) \bm{A}_e}) \bm{\theta}_e$
and variance
\[
  \bm{V}_v 
  = \int_0^{\ell(e)} e^{-s \bm{A}_e} \bm{\Sigma}_e e^{-s \bm{A}_e^{\top}}\, ds 
  = \bm{S}_e - e^{- \ell(e) \bm{A}_e} \bm{S}_e e^{- \ell(e) \bm{A}_e^{\top}}
\]
where $\bm{S}_e$ is the stationary variance matrix. These equations simplify greatly if $\bm{A}_e$ and $\bm{\Sigma}_e$
commute, such as if $\bm{A}_e$ is scalar of the form $\alpha_e\bm{I}_p$,
including when the process is univariate.
In this case,
\[
\bm{V}_v = (1-e^{-2\alpha \ell(e)})\bm{\Sigma}_e/(2\alpha) \,.
\]
Shifts in adaptive regimes can be modelled by shifts in any of the parameters
$\bm{\theta}_e$, $\bm{A}_e$ or $\bm{\Sigma}_e$ across edges.

Finally, variation within species, including measurement error, can be easily
modelled by grafting one or more edges at each
species node, to model the fact that the measurement taken from
an individual may differ from the true species mean. The model for
within-species variation, then, should also follow \eqref{eq:lineargaussian}
by which an individual value is assumed to be normally distributed
with a mean that depends linearly on the species mean, and a variance independent
of the species mean -- although this variance can vary across species.
Most typically, observations from species $v$ are modelled using
$\bm{q}=\bm{I}_p$, ${\omega} = \bm{0}$ and some phenotypic variance
to be estimated, that may or may not be tied to the evolutionary
variance parameter from the phylogenetic model across species.
This additional observation layer can also be used
for factor analysis, where the unobserved latent trait evolving on the
network has smaller dimension than the observed trait.
In that case, $\bm{q}$ is a rectangular, representing
the loading matrix \cite{Tolkoff2018,2022Hassler-PFA}.

\subsection{Evolutionary models at reticulations}

For a continuous trait and a hybrid node $h$, \cite{2018Bastide-pcm-net}
and \cite{2012PickrellPritchard} assumed that $X_h$ is a
weighted average of its \emph{immediate} parents, using their state immediately
before the reticulation event.
Specifically, if $h$ has parent edges $e_1,\ldots,e_m$,
and if we denote by $X_{\underline{e_k}}$ the state
at the end of edge $e_k$ right before the reticulation event ($1\leq k \leq m$),
then the weighted-average model assumes that
\begin{equation}\label{eq:weighedaverage-merging}
  X_h = \sum_{e_k\text{ parent of }h} \gamma(e_k) X_{\underline{e_k}} \;.
\end{equation}
This model is a reasonable null model for polygenic traits,
reflecting the typical observation that hybrid species show intermediate
phenotypes. In this model, the biological process underlying the reticulation
event (such as gene flow versus hybrid speciation) does not need to be known.
Only the proportion of the genome inherited by each parent, $\gamma(e_k)$, needs
to be known. Compared to the evolutionary time scale of the phylogeny,
the reticulation event is assumed to be instantaneous.

To describe this process as a graphical model, we may add a degree-2 node
at the end of each hybrid edge $e$ to store the value $X_{\underline{e}}$,
so as to separate the description of the evolutionary process along each
edge from the description of the process at a reticulation event.
With these extra degree-2 nodes, the weighted-average model
\eqref{eq:weighedaverage-merging} corresponds to the linear Gaussian model
\eqref{eq:lineargaussian} with no trend ${\omega}_h=\bm{0}$,
no variance $\bm{V}_h=\bm{0}$, and with actualization
$\bm{q}_h=[\gamma(e_1)\bm{I}_p \,\ldots\, \gamma(e_m)\bm{I}_p]$
made of scalar diagonal blocks.

Several extensions of this hybrid model can be considered.
\cite{2018Bastide-pcm-net} modelled transgressive evolution with a
shift ${\omega}_h\neq\bm{0}$, for the hybrid population to differ
from the weighted average of its immediate parents, even possibly taking
a value outside their range.
\cite{2015JhwuengOMeara} considered transgressive shifts at each hybrid node
as random variables with a common variance, corresponding to a model with
${\omega}_h=\bm{0}$ but non-zero variance $\bm{V}_h$.

More generally, we may consider models in which the hybrid value is any linear
combination of its immediate parents
$\bm{q}_v X_{\pa(v)}$ as in \eqref{eq:lineargaussian}.
A biologically relevant model could consider $\bm{q}_v$ to be diagonal,
with, on the diagonal, parental weights $\gamma(e,j)$ that may depend
on the trait $j$ instead of being shared across all $p$ traits.

We may also consider both a fixed transgressive shift
${\omega}_h\neq\bm{0}$ and an additional hybrid variance $\bm{V}_h$.
For both of these components to be identifiable in the typical case
when we observe a unique realization of the trait evolution,
the model would need extra assumptions to induce sparsity.
For example, we may assume that $\bm{V}_h$ is shared across all
reticulations and is given an informative prior,
to capture small variations around the parental weighted average.
We may also need a sparse model on the set of ${\omega}_h$ parameters, e.g.\
letting ${\omega}_h\neq\bm{0}$ only at a few candidate reticulations $h$,
chosen based on external domain knowledge.

For a continuous trait known to be controlled by a single gene, we may prefer
a model similar to the discrete trait model presented later in Example~\ref{ex:net1},
by which $X_h$ takes the value of one of its immediate parent $X_{\underline{e}}$
with probability $\gamma(e)$.
This model would no longer be linear Gaussian, unless we condition on
which parent is being inherited at each reticulation. Such conditioning
would reduce the phylogeny to one of its displayed tree. But it would require
other techniques to integrate over all parental assignments to each hybrid
population, such as Markov Chain Monte Carlo or Expectation Maximization.

\subsection{Evolutionary models with interacting populations}\label{sec:model_interaction}

Models have been proposed in which the evolution
of $X^{(e)}(t)$ along one edge $e$ depends on the state on other edges
existing at the same time $t$ \cite{2016Drury-competition,2017ManceauLambertMorlon-interactinglineages,Bartoszek2017-interactingpop,2020Duchen-OUmigration}.
These models can describe ``phenotype matching''
that may arise from ecological interactions (mutualism, competition) or
demographic interactions (migration),
in which traits across species or populations converge to or diverge from one another.
To express this coevolution, we consider
the set $E(t)$ of edges contemporary to one another at time $t$
and divide the phylogeny into \emph{epochs}: time intervals $[\tau_{i},\tau_{i+1}]$
during which the set $E(t)$ of interacting lineages is constant, denoted as $E_i$.
Within each epoch $i$ (i.e.\ $t\in [\tau_{i},\tau_{i+1}]$), the vector of all traits
$(X^{(e)}(t))_{e\in E_i}$ is modelled by a linear stochastic differential equation.
Since its mean is linear in and its variance independent of the starting value
$(X^{(e)}(\tau_i))_{e\in E_i}$,
these models are linear Gaussian
\cite{2017ManceauLambertMorlon-interactinglineages,Bartoszek2017-interactingpop}.
In fact, they can be expressed by~\eqref{eq:lineargaussian}
on a supergraph of the original phylogeny, in which an edge $(u,v)$ is added
if $u$ is at the start $\tau_{i}$ of some epoch $i$,
$v$ is at the end $\tau_{i+1}$, and
if the mean of $X_v$ conditional on all traits at time $\tau_{i}$
has a non-zero coefficient for $X_u$.
The specific form of $\bm{q}_v$, ${\omega}_v$ and $\bm{V}_v$
in~\eqref{eq:lineargaussian} depend on the specific interaction model,
and may be more complex than the merging rule~\eqref{eq:weighedaverage-merging}.

\section{A short review of graphical models and belief propagation}
\label{sec:review-GM-BP}

Implementing BP techniques on general networks is more complex than on trees
and involves the construction of an auxiliary graph known
as a clique tree or cluster graph (section~\ref{sec:defCGCT}).
To explain why, we review here the main ideas of graphical models and
belief propagation for likelihood calculation.

\subsection{Graphical models}\label{subsec:def-GM}

A probabilistic graphical model is a graph representation of a
probability distribution. Each node in the graph represents a random variable,
typically univariate but possibly multivariate.
We focus here on graphical models with directed edges
on a directed acyclic graph (DAG).
Edges represent dependencies between variables, where the direction is typically
used to represent causation.
The graph expresses conditional independencies satisfied by the joint
distribution of all the variables at all nodes in the graph.

Given the directional nature of evolution and inheritance,
models for trait evolution on a phylogeny are often readily formulated as
directed graphical models.
\cite{hohna2014probabilistic} demonstrate the utility of
representing phylogenetic models as graphical models for exposing assumptions,
and for interpretation and implementation.
They present a range of examples common in evolutionary biology,
with a focus on how graphical models facilitate greater modularity and transparency.
Directed graphical models have also been used to parametrize distributions
on tree topologies, for accurate approximations of posterior distributions
and for variational inference
\cite{2018ZhangMatsen,2024ZhangMatsen,2023JunMatsen-compositelikelihood-subsplitDAG}.
In this issue, similar DAGs are used to store phylogenetic trees efficiently,
for parsimony-based inference \cite{2024DummMatsen}.
Here we focus instead on the computational gains that BP allows on graphical models.

A directed graphical model consists of a DAG $G$
and a set of conditional distributions, one for each node in $G$.
At a node $v$ with parent nodes $\pa(v)$, the distribution of variable
$X_v$ conditional on its set of parent variables
$X_{\pa(v)} = \{X_u; u\in\pa(v)\}$ is given by a
\emph{factor} $\phi_v$, which is a function whose \emph{scope} is
the set of variables from $v$ and $\pa(v)$.
For each node $v$, the set formed by this node and its parents $\{v\}\cup\pa(v)$
is called a \emph{node family}.
If $V$ denotes the vertex set of $G$, then the set of factors
$\{\phi_v, v\in V\}$ defines the joint density of the graphical model as
\begin{equation}\label{eq:factor_decomposition}
  p_\theta(X_v; v\in V) =
  \prod_{v\in V}\phi_v(X_v | X_u, \theta; u\in\pa(v))
\end{equation}
where we add the possible dependence of factors on model parameters $\theta$.
This factor formulation implies that, conditional on its parents,
$X_v$ is independent of any non-descendant node
(e.g.\ ``grandparents'') \cite{koller2009probabilistic}.

\begin{examplebox}[float=h,label={ex:tree1}]{Brownian motion (BM) on a tree}
Consider the phylogenetic tree $T$ in Fig.~\ref{fig:tree1}(a).
The graphical model for the node states of $T$ under a BM, whose parameters
$\theta$ are the trait evolutionary variance rate $\sigma^2$, the ancestral
state at the root $x_\rho$ and edge lengths $\ell_i$, has the same topology as
$T$.
On a tree, each node family consists of a node $v$ and its single parent, or the
root $\rho$ by itself.
The distribution $\phi_\rho$ may be deterministic
as when $x_\rho$ is a fixed parameter of the model,
or may be given a prior distribution $\phi_\rho$.

\tcblower
\begin{center}
  \includegraphics{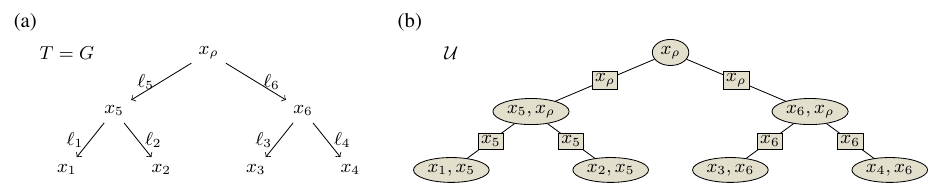}
\end{center}
\captionof{figure}{
  Example graphical model on a phylogenetic tree with
  factors defined by the BM.
The joint distribution of all variables at all nodes is given by the
  product of factors $\prod_v \phi_v$ as in \eqref{eq:factor_decomposition},
  where $\phi_v$ is the distribution of $x_v$ conditional on its parent variable
  $x_{\pa(v)}$: $\mathcal{N}(x_{\pa(v)}, \sigma^2\ell_v)$ under the BM.
  (a) Phylogenetic tree $T$.
  The graphical model uses the same graph $G=T$.
  (b) Clique tree $\mathcal U$ for the graphical model
  (see Def.~\ref{def:clustergraph} in section~\ref{sec:defCGCT}).
  Its nodes are clusters
  of variables in $T$ (ellipses). Each edge is labelled by a sepset
  (squares, see Def.~\ref{def:clustergraph}\ref{def:cl:item2}):
  a subset of variables shared by adjacent clusters.
}\label{fig:tree1}
\end{examplebox}

\begin{examplebox}[float=h,label={ex:net1}]{Discrete trait on a network}
A \emph{rooted phylogenetic network} is a DAG with a single root, and
taxon-labelled leaves (or tips).
A node with at most one parent is called a \emph{tree node} and its incoming
edge is a tree edge.
A node with multiple parents is called a \emph{hybrid} node, and represents a
population (or species more generally) with mixed ancestry.
An edge $e = (u,h)$ going into a hybrid node $h$ is called a hybrid edge.
It is assigned an \emph{inheritance} probability $\gamma(e)>0$ that represents the
proportion of the genome in $h$ that was inherited from the parent population $u$
(via edge $e$).
Obviously, at each hybrid node $h$ we must have
$\sum_{u\in\pa(h)}\gamma((u,h))=1$.
The phylogenetic network $N$ in Fig.~\ref{fig:net1}(a) has one hybrid node $x_5$
whose genetic makeup comes from $x_4$ with proportion $0.4$
and from $x_6$ with proportion $0.6$.

\begin{center}
  \includegraphics{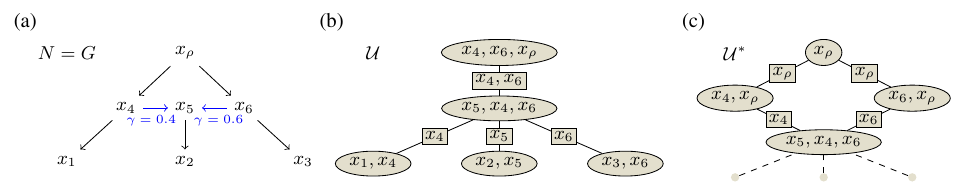}
\end{center}
\captionof{figure}{
(a) Phylogenetic network $N$ with hybrid edges shown in blue.
The graphical model uses the same graph $G=N$.
$N$ displays two trees, depending on which hybrid edge is retained.
One tree, with sister taxa 1 and 2, has probability $0.4$.
The other tree, with sister taxa 2 and 3, is displayed with probability $0.6$.
The distribution of the hybrid node $x_5$ depends on both its parents, and
induces a factor cluster $\{x_4,x_5,x_6\}$ of size 3 in $\mathcal{U}$ and
$\mathcal{U}^*$
(see the family-preserving property in
Def.~\ref{def:clustergraph}\ref{def:cl:item1}). (b) Clique tree $\mathcal{U}$ for the graphical model.
(c) Cluster graph
$\mathcal{U}^*$ (see Def.~\ref{def:clustergraph}) for the same graphical model in which
$\{x_4,x_6,x_\rho\}$ in $\mathcal{U}$ is replaced by smaller clusters
$\{x_4,x_\rho\}$, $\{x_6,x_\rho\}$ and $\{x_\rho\}$ that induce a cycle.
Leaf clusters are not shown.
}\label{fig:net1}

For a discrete trait $X$, the traditional model of evolution on a tree
can be extended to a network $N$ as follows.
Along each edge $e$, $X$ evolves according to a Markov process with some
transition rate matrix $Q$ for an amount of time\ $\ell(e)$ that depends on the edge.
At a tree node, the state of $X$ at the end of its parent edge is passed as the
starting value to each daughter lineage, as in the traditional tree model.
At reticulations, we follow previous authors to model the value $x_h$
at a hybrid node $h$
\cite{2000StrimmerMoulton,2020karimi,allen2020estimating,lutteropp2022netrax}.
Let $x_{\underline{e}}$ denote the state at the end of edge $e$,
going forward in time.
If $h$ has $m$ parent edges $e_1,\cdots,e_m$, then $x_h$ is assumed to take
value $x_{{\underline{e_k}}}$ with probability $\gamma(e_k)$.
This model reflects the idea that the trait is controlled by unknown genes,
but the proportion of genes inherited from each parent is known.
Incomplete lineage sorting, which can lead to hemiplasy for a trait
\cite{2008Avise-hemiplasy}, is unaccounted for. Similar to
Example~\ref{ex:tree1}, the graphical model uses the topology of the network $N$.

\medskip
To describe the factors of this graphical model and simplify notations,
consider the case when $X$ is binary with states 0 and 1.
For a tree node $v$, the factor $\phi_v$ can be represented by the
$2\times 2$ matrix $\exp(\ell(e)Q)$, where $e$ is the parent edge of $v$.
For a hybrid node $h$ with $m$ parents $p_1,\cdots,p_m$ and edges $e_k=(p_k,h)$
with $\gamma(e_k)=\gamma_k$,
the factor $\phi_h$ has scope
$(X_h,X_{p_1},\cdots,X_{p_m})$, and can be described by a $2\times 2^m$ matrix
to store the conditional probabilities
$\mathbb{P}(X_h=j \mid X_{p_1}=i_1,\cdots,X_{p_m}=i_m)$.
This is a $2\times 4$ matrix in the typical case when $h$ is admixed from
$m=2$ parental populations.
With $m=2$ and with parental values $(X_{p_1},X_{p_2})$
arranged ordered $((0,0),(0,1),(1,0),(1,1))$, then
\[\phi_h = \left(\!\begin{array}{cccc}
    1&\gamma_1&\gamma_2&0\\0&\gamma_2&\gamma_1&1
  \end{array}\!\right)\;.
\]
\end{examplebox}

\FloatBarrier

\begin{examplebox}[label={ex:tree2}]{Binary trait with incomplete lineage sorting (ILS)}
More complex evolutionary processes, such as
ILS, can result in a graph $G$ for the graphical model that
is constructed from but not identical to the phylogeny.
Such is the case for the evolution of a genetic marker whose gene tree
is generated according to the coalescent model along the species
phylogeny \cite{1982kingman,2003RannalaYang,2023fogg_phylocoalsimulations}.
For a marker with 2 alleles, say ``green'' and ``red'',
the data consist of the number of red alleles in a sample of individuals from
each species.
In SM section A, we formulate the likelihood calculations
by \cite{2012Bryant-snapp} on a species tree, and
by \cite{2021Rabier-snappnet} on a species network,
as belief propagation.
For this evolutionary model, the graph $G$ of the associated graphical model
differs significantly from the original phylogeny
(illustrated in Figs.~S1 and~S2).
\end{examplebox}

\begin{examplebox}[label={ex:threshold}]{Discrete trait determined by an unobserved continuous trait}
The threshold model uses a latent (unobserved) continuous trait, or ``liability'',
evolving as a Brownian Motion, to determine an observed discrete trait.
The discrete trait changes state when the liability crosses a threshold
\cite{Wright1934,2005Felsenstein}.
As the liability is unobserved, the discrete trait  has ``memory'' and
is not Markovian \cite{Goldberg2020}:
the probability to transition from one state to another 
depends on the amount of time spent in the current state.
However, we can express the model as a graphical model suitable for BP
by modelling the liability at all nodes in the phylogeny, and by adding
an observation layer to the graph $G$ for the value of the discrete trait at each tip.
This layer adds a pendant edge to connect an observed trait to its corresponding
liability.
Thresholding adds significant complexity to likelihood calculations.
Existing algorithms, on trees,
use approximations \cite{Hiscott2016,Goldberg2020}
or resort to sampling the latent liability in a Bayesian context
\cite{Felsenstein2012,Revell2014,Cybis2015}, including Hamiltonian techniques that exploit the gradient
\cite{Zhang2021,Zhang2023}. 

\smallskip
Representing discrete traits as thresholded liabilities
makes it easy to model correlations between continuous and discrete traits,
in multivariate datasets with both types of traits \cite{Cybis2015}.
Using graphical models, such models can leverage BP techniques
and hence be extended to phylogenetic networks and general Gaussian processes.
\end{examplebox}

\subsection{Belief Propagation}\label{subsec:BP}
BP is a framework for efficiently computing various integrals of the factored
density $p_\theta$ by grouping nodes and their associated variables into clusters
and integrating them out according to rules along a clique tree
(also known by junction tree, join tree, or tree decomposition)
or along a cluster graph, more generally.

\label{sec:defCGCT}
\subsubsection{Cluster graphs and Clique trees}

\begin{definition}[cluster graph and clique tree]\label{def:clustergraph}
  Let $\Phi=\{\phi_v, v\in V\}$ be the factors of a graphical model on graph $G$ and let
  $\mathcal{U}=(\mathcal{V},\mathcal{E})$ be an undirected graph
  whose nodes $\mathcal{C}_i\in\mathcal{V}$, called \emph{clusters},
  are sets of variables in the scope of $\Phi$.
$\mathcal{U}$ is a cluster graph for $\Phi$ if it satisfies the following properties:
  \begin{enumerate}
\item (\emph{family-preserving}) \label{def:cl:item1}
    There exists a map $\alpha:\Phi\rightarrow\mathcal{V}$
    such that for each factor $\phi_v$, its scope
    (node family for node $v$ in the graphical model)
    is a subset of the cluster $\alpha(\phi_v)$.
    \item (\emph{edge-labeled}) \label{def:cl:item2}
    Each edge $\{\mathcal{C}_i,\mathcal{C}_j\}$ in $\mathcal{E}$
    is labelled with a non-empty \emph{sepset} $\mathcal{S}_{i,j}$
    (``separating set'')
    such that $\mathcal{S}_{i,j}\subseteq\mathcal{C}_i\cap\mathcal{C}_j$.
\item (\emph{running intersection}) \label{def:cl:item3}
    For each variable $x$ in the scope of $\Phi$, $\mathcal{E}_x\subseteq
    \mathcal{E}$, the set of edges with $x$ in their sepsets forms a tree that
    spans $\mathcal{V}_x\subseteq\mathcal{V}$, the set of clusters that contain
    $x$.
  \end{enumerate}
  If $\mathcal{U}$ is acyclic, then $\mathcal{U}$ is called a \emph{clique tree}
  and we refer to its nodes as \emph{cliques}.
  In this case, properties \ref{def:cl:item2} and \ref{def:cl:item3} imply that
  $\mathcal{S}_{i,j} = \mathcal{C}_i\cap\mathcal{C}_j$.
\end{definition}

A clique tree $\mathcal{U}$ is shown in Fig.~\ref{fig:tree1}(b)
for the BM model from Example~\ref{ex:tree1},
on the tree $T$ in Fig.~\ref{fig:tree1}(a).
To check the running intersection property for
$x_5$, for example, we extract the graph defined by edges with $x_5$ in their
sepsets (squares). There are 2 such edges. They induce a subtree of $\mathcal{U}$
that connects all 3 clusters (ellipses) containing $x_5$, as desired.
More generally, when the graphical model is defined on a tree $T$,
a corresponding clique tree $\mathcal{U}$ is easily constructed, where cliques
in $\mathcal{U}$ correspond to edges in $T$, and edges in $\mathcal{U}$ correspond to nodes in $T$.
Multiple clique trees can be constructed for a given graphical model.
In this example, the clique $\{x_\rho\}$ (shown at the top) could be
suppressed, because it is a subset of adjacent cliques.

For the network $N$ in Fig.~\ref{fig:net1}(a) and the evolution of a
discrete trait in Example~\ref{ex:net1}, one possible clique tree $\mathcal{U}$
is shown in Fig.~\ref{fig:net1}(b). Note that $x_5,x_4$ and $x_6$ have to
appear together in at least one of the clusters for the clique tree
to be family-preserving (property~\ref{def:cl:item1}),
because $x_4$ and $x_6$ are partners with a common child $x_5$
whose distribution depends on both of their states.

We first focus on clique trees, which provide a structure for the exact
likelihood calculation. In section~\ref{sec:loopybp} we discuss
the advantages of cluster graphs, to approximate the likelihood
at a lower computational cost.

\subsubsection{Evidence}

To calculate the likelihood of the data, or the marginal distribution of
the traits at some node conditional on the data, we inject \emph{evidence}
into the model, in one of two equivalent ways.
For each observed value $\mathrm{x}_{v,t}$ of the $t^\mathrm{th}$ trait
$x_{v,t}$ at node $v$, we add to the model the indicator function
$\mathbbm{1}_{\{\mathrm{x}_{v,t}\}}(x_{v,t})$ as an additional factor.
Equivalently, we can plug in the observed value
$\mathrm{x}_{v,t}$ in place of the variable $x_{v,t}$ in all factors where
$x_{v,t}$ appears, and then drop $x_{v,t}$ from the scope of all these factors.
This second approach is more tractable than the first to avoid the degenerate
zero-variance Dirac distribution. But it requires careful bookkeeping of the
scope and of re-parametrization of each factor with missing data,
when some traits but not all are observed at some nodes.
Below, we assume that the factors and their scopes have been modified
to absorb evidence from the data.

\subsubsection{Belief update message passing}\label{subsubsec: BP rules}

There are multiple equivalent algorithms to perform BP.
We focus here on the \emph{belief update} algorithm.
It assigns a \emph{belief} to each cluster and to each sepset in the cluster graph.
After running the algorithm, each belief should provide the marginal probability
of the variables in its scope and of the observed data, with all other variables
integrated out as desired to calculate the likelihood.
The belief of cluster $\mathcal{C}_i$, denoted as $\beta_i$,
is initialized as the product of all
factors assigned to that cluster:
\begin{equation}\label{eq:clusterfactor}
  \beta_i^{\mathrm{(initial)}} =
  \psi_i = \prod_{\phi; \,\alpha(\phi)=\mathcal{C}_i}\phi\quad\mbox{ for cluster }\mathcal{C}_i
\end{equation}
The belief of an edge between cluster $i$ and $j$, denoted as $\mu_{i,j}$,
is initialized to the constant function 1.
These beliefs are then updated iteratively by passing messages.
Passing a message from $\mathcal{C}_i$ to $\mathcal{C}_j$ along an edge with
sepset $\mathcal{S}_{i,j}$ corresponds to passing information about the marginal
distribution of the variables in $\mathcal{S}_{i,j}$ as shown in Algorithm~\ref{alg:bp}.
\begin{algorithm}[!h]
  \caption{Belief propagation: message passing along an edge from
  $\mathcal{C}_i$ to $\mathcal{C}_j$ with sepset $\mathcal{S}_{i,j}$.}
  \label{alg:bp}
  \begin{algorithmic}[1]
\State compute the message
  $\tilde\mu_{i\rightarrow j} = \int_{\mathcal{C}_i\setminus \mathcal{S}_{i,j}}
    \beta_i d(\mathcal{C}_i\!\setminus\!\mathcal{S}_{i,j})$,
  that is, the marginal probability of $\mathcal{S}_{i,j}$
  based on belief $\beta_i$,
  by integrating all other variables in $\mathcal{C}_i$,
\State update the cluster belief about $\mathcal{C}_j$:
  $\beta_j \leftarrow \beta_j \tilde\mu_{i\rightarrow j} / \mu_{i,j}$,
\State update the edge belief about $\mathcal{S}_{i,j}$:
  $\mu_{i,j} \leftarrow \tilde\mu_{i\rightarrow j}$.
\end{algorithmic}
\end{algorithm}
If $\mathcal{U}$ is a clique tree, then all beliefs converge to the true
marginal probability of their variables and of the observed data,
after traversing $\mathcal{U}$ only twice: once to pass messages
from leaf cliques towards some root clique, and then back from the root clique
to the leaf cliques.
If our goal is to calculate the likelihood, then
one traversal is sufficient. Once the root clique has received messages
from all its neighboring cliques, we can marginalize over all its variables
(similar to step 1) to obtain the probability of the observed data only,
which is the likelihood.
The second traversal is necessary to obtain the marginal probability of
all variables, such as if one is interested in the posterior distribution
of ancestral states conditional on the observed data.

Some equivalent formulations of BP only store sepset messages, and
avoid storing cluster beliefs. This strategy requires less memory
but more computing time if $\mathcal{U}$ is traversed multiple times.

In Example~\ref{ex:tree1} on a tree (Fig.~\ref{fig:tree1}(a)),
the conditional distribution of $x_v$ at a non-root node $v$ corresponds to a
factor $\phi_v$ for the BM model along edge $(\pa(v),v)$ in $T$.
This factor is assigned to clique $\mathcal{C}_v=\{\pa(v),v\}$ in $\mathcal U$
to initialize the belief $\beta_v$ of $\mathcal{C}_v$.
If $v$ if a leaf in $T$, then $\beta_v$ is further multiplied by
the indicator function at the value $\mathrm{x}_v$ observed at $v$,
such that the belief of clique $\mathcal{C}_v$ can be expressed as a function
of the leaf's parent state only:
$\phi_v(x_{\pa(v)}) = \mathbb{P}(\mathrm{x}_v \mid x_{\pa(v)})$.
The prior distribution $\phi(x_\rho)$ at the root $\rho$ of $T$
(which can be an indicator function if the root value is fixed as a model
parameter) can be assigned to any clique containing $\rho$.
In Fig.~\ref{fig:tree1}, $\mathcal U$ includes a clique
$\mathcal{C}_\rho=\{x_\rho\}$ drawn at the top, to which we assign the root
prior $\phi_\rho(x_\rho)$ and which we will use as the root of $\mathcal U$.
Since $\mathcal U$ is a clique tree, BP converges after traversing $\mathcal U$
twice: from the tips to $\mathcal{C}_\rho$ and then back to the tips.
IC \cite{felsenstein1973maximum,felsenstein85}
implements the first ``rootwards'' traversal of BP.
For example, the belief of clique
$\{x_5,x_\rho\}$ after receiving messages (steps 1-3) from both of its daughter
cliques is the function
\[\beta_5(x_5,x_\rho) =
  \exp\left(-\frac{(x_\rho - x_5)^2}{2\ell_5}
            -\frac{(x_5 - \mathrm{x}^*_5)^2}{2v^*_5} + g^*_5\right)
\]
where
\[\mathrm{x}^*_5=\frac{\ell_2 \mathrm{x}_1+\ell_1 \mathrm{x}_2}{\ell_1+\ell_2},
  \quad
  v^*_5 = \frac{\ell_1\ell_2}{\ell_1+\ell_2},\quad\mbox{and}\quad
  g^*_5 = -\frac{(\mathrm{x}_2-\mathrm{x}_1)^2}{2(\ell_1+\ell_2)}
          - \log((2\pi)^{3/2}\ell_1\ell_2\ell_5)
\]
are quantities calculated for IC:
$\mathrm{x}^*_5$ corresponds to the estimated ancestral state at node 5,
$v^*_5$ corresponds to the extra length added to $\ell_5$ when pruning the
daughters of node 5, and $g^*_5$ captures the
contrast $(\mathrm{x}_2-\mathrm{x}_1)/\sqrt{\ell_1+\ell_2}$ below node 5.
At this stage of BP, $\beta_5(x_5,x_\rho)$ can be interpreted as
$\mathbb{P}(\mathrm{x}_1,\mathrm{x}_2,x_5 \mid x_\rho)$
such that the message $\tilde\mu_{5\rightarrow\rho}(x_\rho)$
sent from $\{x_5,x_\rho\}$ to the root clique $\mathcal{C}_\rho$ is
the partial likelihood
$\mathbb{P}(\mathrm{x}_1,\mathrm{x}_2 \mid x_\rho)$ after $x_5$ is integrated out.
The first pass is complete when $\mathcal{C}_\rho$
has received messages from all its neighbors. Its final belief is then
$\beta_\rho(x_\rho) = \mathbb{P}(\mathrm{x}_1,\cdots,\mathrm{x}_4 \mid x_\rho)\phi_\rho(x_\rho)$.
If $x_\rho$ is a fixed model parameter, then this is the likelihood.
Otherwise, we get the likelihood by integrating out $x_\rho$ in $\beta_\rho(x_\rho)$.

\medskip
In Example~\ref{ex:net1} on a network (Fig.~\ref{fig:net1}), we label the
cliques in $\mathcal U$ as follows:
$\mathcal{C}_v=\{x_v,x_{\pa(v)}\}$ for leaves $v=1,2,3$,
$\mathcal{C}_5=\{x_5,x_4,x_6\}$ for hybrid node $v=5$ and its parents,
and $\mathcal{C}_\rho=\{x_4,x_6,x_\rho\}$.
To initialize beliefs, we assign $\phi_v$ to $\mathcal{C}_v$ for $v=1,2,3,5$,
and $\phi_4$, $\phi_6$ are both assigned to $\mathcal{C}_\rho$.
Unlike in Example~\ref{ex:tree1},
a clique may correspond to more than a single edge in $N$.
This is expected at a hybrid node $h$, because the factor describing its
conditional distribution needs to contain $h$ and both of its parents.
But for $\mathcal U$ to be a clique tree, the root clique $\mathcal{C}_\rho$
also has to contain the factors from 2 edges in $N$.
Also, unlike for trees, sepsets may contain more than a
single node. Here, the two large cliques are separated by
$\{x_4,x_6\}$ so they will send messages $\tilde\mu(x_4,x_6)$ about the
joint distribution of these two variables.
In this binary trait setting, these messages and sepset
belief can be stored as $2\times 2$ arrays,
and the 3-node cliques beliefs can be stored as arrays of $2^3$ values.
As they involve more variables than when $G$ is a tree
(in which case BP would store only 2 values at each sepset),
storing and updating them requires more computating time and memory.

More generally, we see that the computational complexity of BP
scales with the size of the cliques and sepsets. This complexity may
become prohibitive on a more complex phylogenetic network, even for
a simple binary trait without ILS, if the size of the largest cluster
in $\mathcal U$ is too large ---a topic that we explore later.

\medskip
Example~\ref{ex:tree2} illustrates the fact that beliefs cannot always
be interpreted as partial (or full) likelihoods at every step of BP,
unlike in Examples~\ref{ex:tree1} and \ref{ex:net1}.
For example, consider the tip clique ${\mathcal C}_1$ containing
the total number of alleles and the number of red alleles in species 1,
and the number of their ancestral alleles ($n$ total, of which $r$ are red)
just before the speciation event that led to species~1 (Fig. S1(c)).
At the first iteration of BP, the first message sent by ${\mathcal C}_1$
is the quantity denoted by $F^{T}(n,r)$ in \cite{2012Bryant-snapp}.
It is \emph{not} a partial likelihood, because it is not the
likelihood of some partial subset of the data conditional on
some ancestral values ($n$ and $r$).
Intuitively, this is because nodes with data below
variables in ${\mathcal C}_1$ in $G$ are not all below
${\mathcal C}_1$ in the clique tree. Information from these data will flow
towards ${\mathcal C}_1$ at later steps.
The beauty of BP is that after a second traversal of the clique tree,
${\mathcal C}_1$'s belief is
guaranteed to converge to
the likelihood of the \emph{full} data, conditional on the state of the clique
variables.
See SM section~A.1 for details.

\subsubsection{Clique tree construction}\label{subsubsec: CT construction}

For a given graphical model on $G$, there are many possible clique trees and
cluster graphs. For running BP, it is advantageous to have small clusters and
small sepsets. Indeed, clusters and sepsets with fewer variables require
less memory to store beliefs, and less computing time to run
steps~1 (integration) and 2 (belief update).
Ideally, we would like to find the best clique tree: whose largest clique
is of the smallest size possible.
For a general graph $G$, finding this best clique tree is hard but good heuristics exist \cite{koller2009probabilistic}.

The first step is to create the \emph{moralized} graph $G^m$ from $G$.
This is done by connecting all nodes that share a common child, and then
undirecting all edges.
We can then \emph{triangulate} $G^m$, that is, build a new graph $H$
by adding edges to $G^m$ such that $H$ is chordal (any cycle includes a chord).
The \emph{width} of $H$ is the size of its largest clique minus 1.
The \emph{treewidth} of $G^m$ is the smallest width of all its possible
triangulations $H$.
Finding $H$ of minimum width
is hard, though efficient heuristics exist
(e.g.\ greedy minimum-fill \cite{rose1972graph, fishelson2003optimizing}).
The nodes of $\mathcal U$ are then defined as
the maximal cliques of $H$
\cite{blair1993introduction}.
Finally, the edges of $\mathcal U$ are formed such that $\mathcal U$
becomes a tree and such that the sum of the sepset sizes is maximum,
by finding a maximum spanning tree using Kruskal's algorithm or
Prim's algorithm \cite{cormen2009introduction}.
All these steps have polynomial complexity.

\subsection{BP for Gaussian models}
\label{subsec: BP for GBN}

Before discussing BP on cluster graphs that are not clique trees,
we focus on BP updates for the evolutionary models presented in
section~\ref{sec:cts trait evolution}. On a phylogenetic network $N$,
the joint distribution of all present and ancestral species
$(X_v)_{v\in N}$ is multivariate Gaussian precisely when it comes from a graphical model on $N$
whose factors $\phi_v$ are \emph{linear Gaussian} \cite{koller2009probabilistic}.
The factor at node $v$ is linear Gaussian if, conditional on its parents,
$X_v$ is Gaussian with a mean that is \emph{linear} in the parental values
and a variance \emph{independent} of parental values,
hence the term $\mathcal{G}_{\text{LInv}}$ used by \cite{mitov2020fast}.
In other words, for the joint process to be Gaussian, each factor $\phi_v(x_v\mid x_{\pa(v)})$
should be of the form~\eqref{eq:lineargaussian}.

Such models have been called Gaussian Bayesian networks
or graphical Gaussian networks, and are special cases of
Gaussian processes (on a graph).
These Gaussian models are convenient for BP because linear
Gaussian factors have a convenient parametrization that allows for a compact
representation of beliefs and belief update operations. 
Namely, the factor giving the conditional distribution
$\phi_v(x_v\mid x_{\pa(v)})$ from \eqref{eq:lineargaussian}
can be expressed in a \emph{canonical form}
as the exponential of a quadratic form:
\begin{equation}\label{eq:canonicalform}
    \text{C}(x;\bm{K}, h, g)=\exp\left(-\frac{1}{2}x^{\top}\bm{K}x + h^{\top}x + g\right)\;.
\end{equation}
For example, if we think of $\phi_v(x_v\mid x_{\pa(v)})$ as a function of
$x_v$ primarily, we may use the parametrization
$\text{C}(x_v; \bm{K}, h, g)$ with
\[\bm{K} = \bm{V}^{-1}_v, \quad
  h=\bm{V}_v^{-1}\left(\bm{q}_v x_{\pa(v)} + {\omega}_v\right),\quad
  \mbox{ and }\quad
  g = -\frac{1}{2}\left(\log|2\pi\bm{V}_v| +
        \|\bm{q}_v x_{\pa(v)}+{\omega}_v\|^2_{\bm{V}_v^{-1}}\right)
\]
where $\|y\|^2_{\bm{M}}$ denotes $y^{\top}\bm{M}y$.
We can also express $\phi_v$
as a canonical form over its full scope
\[\phi_v(x_v\mid x_{\pa(v)}) = \text{C}\left(
  \begin{bmatrix}x_v \\ x_{\pa(v)}\end{bmatrix};\, \bm{K}_v, h_v, g_v\right)
\]
with
\begin{equation}\label{eq:canonicalformGGM}
  \bm{K}_v =
  \begin{bmatrix}\bm{V}_v^{-1} & -\bm{V}_v^{-1}\bm{q}_v \\
    -\bm{q}_v^{\top}\bm{V}_v^{-1} & \bm{q}_v^{\top}\bm{V}_v^{-1}\bm{q}_v
  \end{bmatrix} =
  \begin{bmatrix}\bm{I} \\ -\bm{q}_v^{\top}\end{bmatrix}\bm{V}_v^{-1}
  \begin{bmatrix}\bm{I} & -\bm{q}_v\end{bmatrix}, \quad
  h_v =
  \begin{bmatrix}\bm{V}_v^{-1}{\omega}_v \\
    -\bm{q}_v^{\top}\bm{V}^{-1}_v{\omega}_v
  \end{bmatrix}, \quad
  g_v = -\frac{1}{2}(\log|2\pi\bm{V}_v| +
        \|{\omega}_v\|_{\bm{V}_v^{-1}})\;.
\end{equation}
If $v$ is a leaf with fully observed data, then we need to plug-in
the data $\mathrm{x}_v$ into $\phi_v$ and consider this factor
as a function of $x_{\pa(v)}$ only.
We can express $\phi_v(\mathrm{x}_v\mid x_{\pa(v)})$ as the canonical form
$\text{C}(x_{\pa(v)}; \bm{K}, h, g)$ with
\[\bm{K} = \bm{q}_v^{\top}\bm{V}^{-1}_v\bm{q}_v, \quad
  h=\bm{q}_v^{\top}\bm{V}_v^{-1} (\mathrm{x}_v - {\omega}_v),\quad
  \mbox{ and }\quad
  g = -\frac{1}{2}\left(\log|2\pi\bm{V}_v| +
        \|\mathrm{x}_v - {\omega}_v\|^2_{\bm{V}_v^{-1}}\right)\;.
\]
If data are partially observed at leaf $v$, the same principle applies.
We can plug-in the observed traits into $\phi_v$ and express $\phi_v$ as a
canonical form over its reduced scope: $x_{\pa(v)}$ and any unobserved $x_{v,t}$.
Some quadratic terms captured by $\bm{K}_v$ on the full scope
become linear or constant terms after plugging-in the data, and some linear
terms captured by $h_v$ on the full scope become constant terms in the
canonical form on the reduced scope.

\medskip
An important property of this canonical form is its closure under the
belief update operations: marginalization (step 1) and factor product (step 2).
Indeed, the product of two canonical forms with the same scope satisfies:
\[\text{C}(x;\bm{K}_1,h_1,g_1) \; \text{C}(x;\bm{K}_2,h_2,g_2) =
  \text{C}(x;\bm{K}_1+\bm{K}_2,h_1+h_2,g_1+g_2)\;.
\]
Now consider marginalizing a factor $\text{C}(x;\bm{K}, h, g)$
to a subvector $x^*$ of $x$, by integrating out the elements
$x\!\setminus\!x^*$ of $x$ .
let $\bm{K}_\mathrm{S}$ and $\bm{K}_\mathrm{I}$ be the submatrices of $\bm{K}$
that correspond to $x^*$ (\textbf{S}cope of marginal or \textbf{S}epset) and
$x\!\setminus\!x^*$ (variables to be \textbf{I}ntegrated out), and let
$\bm{K}_{\mathrm{S},\mathrm{I}}=\bm{K}_{\mathrm{I},\mathrm{S}}^{\top}$ be the
cross-terms. If $\bm{K}_\mathrm{I}$ is invertible, then:
\begin{equation*}
  \int\text{C}_{x\setminus x^*}(x;\bm{K},h,g) \;
  d(x\!\setminus\!x^*) = \text{C}(x^*;\bm{K}^*,h^*,g^*)
\end{equation*}
where
$\bm{K}^* = \bm{K}_\mathrm{S}-\bm{K}_{\mathrm{S},\mathrm{I}} \bm{K}_\mathrm{I}^{-1}\bm{K}_{\mathrm{I},\mathrm{S}}$,
$h^* = h_\mathrm{S}-\bm{K}_{\mathrm{S},\mathrm{I}}\bm{K}_\mathrm{I}^{-1} h_\mathrm{I}$
with $h_\mathrm{S}$ and $h_\mathrm{I}$ defined as the subvector of $h$
corresponding to $x^*$ and $x\!\setminus\!x^*$ respectively,
and
$g^* = g+(\log|2\pi\bm{K}_\mathrm{I}^{-1}|+ \|h_\mathrm{I}\|_{\bm{K}_\mathrm{I}^{-1}})/2$.

If the factors of a Gaussian network are non-deterministic, then each belief can be
parametrized by its canonical form, and the above equations can be applied to
update the cluster and sepset beliefs for BP (Algorithm~\ref{alg:bp}).
For cluster $\mathcal{C}_i$, let $(\bm{K}_i,h_i,g_i)$ parametrize its belief
$\beta_i$. For sepset $\mathcal{S}_{i,j}$, let
$(\bm{K}_{i,j},h_{i,j},g_{i,j})$ parametrize its belief $\mu_{i,j}$.
Also, for step 1 of BP, let
$(\bm{K}_{i\rightarrow j},h_{i\rightarrow j},g_{i\rightarrow j})$
parametrize the message $\tilde{\mu}_{i\rightarrow j}$ sent from
$\mathcal{C}_i$ to $\mathcal{C}_j$.
Then BP updates can be expressed as shown below.

\begin{algorithm}[!h]
  \caption{Gaussian belief propagation: from
  $\mathcal{C}_i$ to $\mathcal{C}_j$ with sepset $\mathcal{S}_{i,j}$.}
  \label{alg:gbp}
  \begin{algorithmic}[1]
\State compute message $\tilde{\mu}_{i\rightarrow j}$:
  $\begin{cases}
    \bm{K}_{i\rightarrow j} = \bm{K}_\mathrm{S}-\bm{K}_{\mathrm{S},
    \mathrm{I}}\bm{K}_\mathrm{I}^{-1}\bm{K}_{\mathrm{I},\mathrm{S}} \\
    h_{i\rightarrow j} = h_{\mathrm{S}}-\bm{K}_{\mathrm{S},
    \mathrm{I}}\bm{K}_\mathrm{I}^{-1}h_\mathrm{I} \\
    g_{i\rightarrow j} = g_i+(\log |2\pi \bm{K}_{\mathrm{I}}^{-1}|+
    \|h_\mathrm{I}\|_{\bm{K}_\mathrm{I}^{-1}})/2
  \end{cases}$
  \State update the cluster belief $\beta_j$ about $\mathcal{C}_j$:
  $\begin{cases}
    \bm{K}_j\leftarrow\bm{K}_j+
    \text{ext}(\bm{K}_{i\rightarrow j}-\bm{K}_{i,j}) \\
    h_j\leftarrow h_j+
    \text{ext}(h_{i\rightarrow j}-h_{i,j}) \\
    g_j\leftarrow g_j+g_{i\rightarrow j}-g_{i,j}
  \end{cases}$
  \State update the edge belief $\mu_{i,j}$ about $\mathcal{S}_{i,j}$:
  $\begin{cases}
    \bm{K}_{i,j}\leftarrow\bm{K}_{i\rightarrow j} \\
    h_{i,j}\leftarrow h_{i\rightarrow j} \\
    g_{i,j}\leftarrow g_{i\rightarrow j}
  \end{cases}$
\end{algorithmic}
\end{algorithm}
In step 1, $\bm{K}_\mathrm{S}$ and $\bm{K}_\mathrm{I}$ are the submatrices of
$\bm{K}_i$ that correspond to $\mathcal{S}_{i,j}$ and
$\mathcal{C}_i\setminus\mathcal{S}_{i,j}$.
Similarly, $h_\mathrm{S}$ and $h_\mathrm{I}$ are subvectors of $h_i$.
In step 2, $\text{ext}(\bm{K}_{\tilde{\mu}}-\bm{K}_{i,j})$ extends
$\bm{K}_{\tilde{\mu}}-\bm{K}_{i,j}$ to the same scope as $\bm{K}_j$ by
padding it with zero-rows and zero-columns for
$\mathcal{C}_j\setminus \mathcal{S}_{i,j}$.
Similarly, $\text{ext}(h_{i\rightarrow j}-h_{i,j})$ extends
$h_{i\rightarrow j}-h_{i,j}$ to scope $\mathcal{C}_j$ with $0$ entries
on rows for $\mathcal{C}_j\setminus \mathcal{S}_{i,j}$.

If the phylogeny is a tree, performing these updates
from the tips to the root corresponds to
the recursive equations~(9), (10) and (11) of \cite{mitov2020fast},
and to the propagation formulas (A.3)-(A.8) of \cite{2021Bastide-HMC},
who both considered the general linear Gaussian model \eqref{eq:lineargaussian}.

At any point, a belief $\text{C}(x;\bm{K}, h, g)$
gives a local estimate of the conditional mean ($\bm{K}^{-1}h$)
and conditional variance ($\bm{K}^{-1}$) of trait $X$ given data $\mathrm{Y}$,
for $\bm{K}\succ 0$.
An exact belief, such that
$\text{C}(x;\bm{K}, h, g)\propto p_\theta(x\mid\mathrm{Y})$,
gives exact conditional estimates, that is:
$\E(X\mid\mathrm{Y})=\bm{K}^{-1}h$ and
$\var(X\mid\mathrm{Y})=\bm{K}^{-1}$.

\section{Scalable approximate inference with loopy BP}
\label{sec:loopybp}

The previous examples focused on clique trees and the
exact calculation of the likelihood. We now turn to the use of cluster graphs with cycles, or
\emph{loopy} cluster graphs, such as in Fig.~\ref{fig:net1}(c) or Fig.~\ref{fig:net2}(c-d).
BP on a loopy cluster graph, abbreviated as \emph{loopy BP},
can approximate the likelihood and posterior distributions of
ancestral values, and can be more computationally efficient than BP on a clique tree.

\begin{figure}[h]
  \centering
  \includegraphics[scale=1.1]{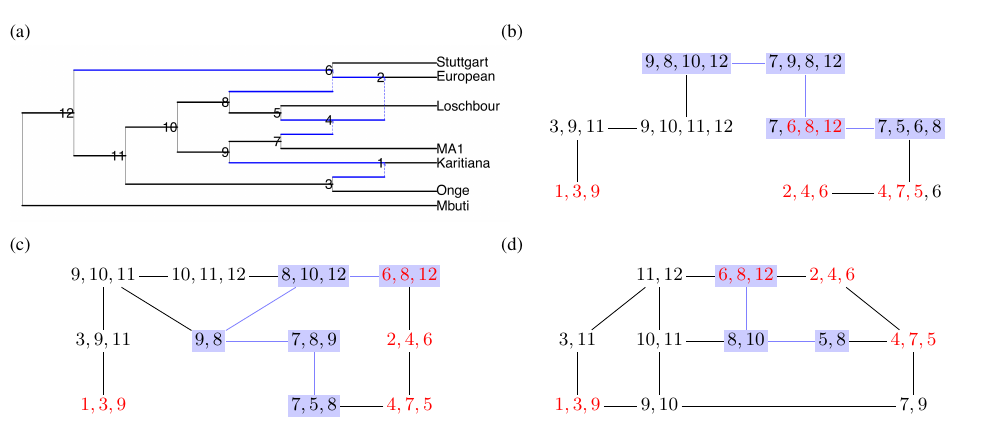}
  \caption{
  (a) Admixture graph $N$ from \cite[Fig.~3]{lazaridis2014ancient}
  with $h=4$ reticulations (hybrid edges are coloured).
$N$ has one non-trivial biconnected component (blob) $B$, induced by all its internal nodes except for the root. $B$ contains all 4
  reticulations so $N$ has level $\ell=4$.
  (b)-(d) Various cluster graphs for the moralized blob $B^m$:
  (b) clique tree,
  (c) join-graph structuring with the maximum cluster size set to 3,
  (d) LTRIP using the set of node families in $B$.
  Here sepsets (not shown) are the intersection of their incident clusters,
  and are small with 1 node only in (c) and (d). Purple boxes and edges:
  clusters and sepsets that contain node~8.
Red text: hybrid families.
}
  \label{fig:net2}
\end{figure}

\subsection{Calibration}
\label{sec:loopy-calibration}

Updating beliefs on a loopy cluster graph uses Algorithm~\ref{alg:bp} in the
same way as on a clique tree.
A cluster graph is said to be \emph{calibrated} when its normalized beliefs have converged
(i.e.\ are unchanged by Algorithm~\ref{alg:bp} along any edge). For calibration,
neighboring clusters $\mathcal{C}_i$ and $\mathcal{C}_j$ must have beliefs that
are marginally consistent over the variables in their sepset $\mathcal{S}_{i,j}$:
\[
  \int\beta_i d(\mathcal{C}_i\!\setminus\!\mathcal{S}_{i,j})
  = \tilde{\mu}_{i\rightarrow j} \propto \mu_{i,j} \propto \tilde{\mu}_{j\rightarrow i} =
  \int\beta_j d(\mathcal{C}_j\!\setminus\!\mathcal{S}_{i,j})\;.
\]
On a clique tree, calibration can be guaranteed at the end of a finite sequence
of messages passed. Clique and sepset beliefs are then proportional to the
posterior distribution over their variables, and can be integrated to compute
the common normalization constant
$\kappa=\kappa_i\left(=\int\beta_i d\mathcal{C}_i\right)=\kappa_{j,k}\left(=
\int\mu_{j,k}d\mathcal{S}_{j,k}\right)$, which equals the likelihood.
For loopy BP, calibration is not guaranteed. If it is attained, then we can
similarly view cluster and sepset beliefs as unnormalized approximations of the
posterior distribution over their variables, though the $\kappa_i$s and
$\kappa_{j,k}$s may differ, grow unboundedly, and generally do not equal or
estimate the likelihood.
Gaussian models enjoy the remarkable property that, if calibration can be
attained on a cluster graph, then the approximate posterior
means (ancestral values) are guaranteed to be exact.
In contrast, the posterior variances are generally inexact, and
are typically underestimated
\cite{weiss1999correctness,wainwright2003tree, malioutov2006walk},
although we found them overestimated in our phylogenetic examples below
(Fig.~\ref{fig:loopyBPapprox}).

Successful calibration depends on various aspects
such as the features of the loops in the cluster graph, the factors in the model,
and the \emph{scheduling} of messages.
For beliefs to converge, a proper message schedule requires that a message is
passed along every
sepset, in each direction, infinitely often (until stopping criteria are met)
\cite{malioutov2006walk}.
Multiple scheduling schemes have been devised to help reach calibration more
often and more accurately.
These can be data-independent (e.g.\ choosing a list of trees
nested in the cluster graph that together cover all clusters and edges, then
iteratively traversing each tree in both directions \cite{wainwright2003tree})
or adaptive (e.g.\ prioritizing messages between clusters that are further
from calibration \cite{elidan2006residualBP,sutton2007residualBP,
knoll2015weightdecayBP,aksenov2020parallelRBP}).

\subsection{Likelihood approximation}
\label{sec:likelihood-approx}

To approximate the log-likelihood $\loglik(\theta)=\log\int p_\theta(x) dx$
from calibrated beliefs on cluster graph $\mathcal{U}^*=({\cal V}^*, {\cal E}^*)$,
denoted together as
$q=\{\beta_i, \mu_{i,j}; {\cal C}_i\in{\cal V}^*, \{{\cal C}_i,{\cal C}_j\}\in{\cal E}^*\}$,
we can use the \emph{factored energy functional} \cite{koller2009probabilistic}: \begin{equation}\label{eq:factoredenergy}
  \tilde{F}(p_\theta,q)
  = \sum_{\mathcal{C}_i\in\mathcal{V}^*}
    \!\E_{\beta_i}(\log \psi_i) + \!\sum_{\mathcal{C}_i\in\mathcal{V}^*}
    \!\mathrm{H}(\beta_i)-\!\!\!\!\!\!\sum_{\{\mathcal{C}_i,\mathcal{C}_j\}\in
    \mathcal{E}^*}\!\!\!\!\!\!\mathrm{H}(\mu_{i,j})\;.
\end{equation}
Recall that $\psi_i$ is the product of factors $\phi_v$ assigned to cluster
${\cal C}_i$.
Here $\E_{\beta_i}$ denotes the expectation with respect to $\beta_i$
normalized to a probability distribution.
$\mathrm{H}(\beta_i)$ and $\mathrm{H}(\mu_{i,j})$ denote the entropy of
the distributions defined by normalizing $\beta_i$ and $\mu_{i,j}$
respectively. $\tilde{F}(p_\theta,q)$ has the advantage of involving local integrals that can
be calculated easily: each over the scope of a single cluster or sepset.
The justification for $\tilde{F}(p_\theta,q)$ comes from two approximations.
First, following the expectation-maximization (EM) decomposition, $\loglik(\theta)$
can be approximated by the \emph{evidence lower bound} (ELBO) used for
variational inference \cite{Ranganath2014}.
For any distribution $q$ over the full set of variables, which are here the
unobserved (latent) variables after absorbing evidence from the data, we have
\[
  \loglik(\theta) \geq \elbo(p_\theta, q) = \E_{q}(\log p_\theta) + H(q)\;.
\]
The gap $\loglik(\theta) - \elbo(p_\theta, q)$ is the Kullback-Leibler divergence
between $q$, and $p_\theta$ normalized to the distribution of the unobserved variables
conditional on the observed data. The first approximation comes from minimizing this gap over a class of
distributions $q$ that does not necessarily include the true conditional distribution.
The second approximation comes from pretending that for a given
distribution $q$ with a belief factorization
\[
  q\propto\frac{\prod_{\mathcal{C}_i\in\mathcal{V}^*}\beta_i}{
      \prod_{\{\mathcal{C}_i,\mathcal{C}_j\}\in\mathcal{E}^*}\mu_{i,j}}\;,
\]
its marginal over a given cluster (or a given sepset)
is equal to the normalized belief of that cluster (or sepset), simplifying
$\E_{q}(\log \psi_i)$ to $\E_{\beta_i}(\log \psi_i)$ and
simplifying $\E_{q}(-\log \beta_i)$ to $\mathrm{H}(\beta_i)$.
This simplification leads to the more tractable $\tilde{F}(p_\theta,q)$,
in which each integral is of lower dimension, within the scope of a single
cluster or sepset.

\subsection{Scalability versus accuracy: choice of cluster graph complexity}
\label{sec:scalability-vs-accuracy}

\subsubsection{Scalability, treewidth and phylogenetic network complexity}

At the cost of exactness, loopy cluster graphs can offer greater computational
scalability than clique trees because they allow for smaller cluster sizes, which
reduces the complexity associated with belief updates.
For example, consider a Gaussian model
for $p$ traits: $\dim(x_v)=p$ at all nodes $v$ in the network.
For a clique tree $\mathcal{U}$ with $m$ cliques and maximum
clique size $k$, passing a message between neighbor cliques has complexity
$\mathcal{O}(p^3k^3)$ and
calibrating $\mathcal{U}$ has
complexity $\mathcal{O}(m p^3k^3)$.
Now consider a cluster graph $\mathcal{U}^*$ with $m^*$ clusters,
$\mathcal{O}(m^*)$ edges, and maximum cluster size $k^* < k$.
Then passing a message between neighbor cliques of $\mathcal{U}^*$
has complexity $\mathcal{O}(p^3{k^*}^3)$ so it is faster than on $\mathcal{U}$.
But calibrating $\mathcal{U}^*$ now requires more belief updates because
each edge needs to be traversed more than twice. If each edge is traversed
in both directions $b$ times to reach convergence,
then calibrating $\mathcal{U}^*$
has complexity $\mathcal{O}(b m^* p^3 {k^*}^3)$.
So if $\mathcal{U}^*$ has smaller clusters than $\mathcal{U}$ and
if $(k/k^*)^3 \gg b m^*/m$,
then loopy BP on $\mathcal{U}^*$ runs faster than BP on $\mathcal{U}$.
Loopy BP could be particularly advantageous for complex networks
whose clique trees have large clusters.

Cluster graph construction
determines the balance between scalability and approximation quality.
At one end of the spectrum, the most scalable and least accurate
are the factor graphs, also known as Bethe cluster graphs
\cite{yedidia2005constructing}.
A factor graph has one cluster per factor $\phi_v$ and one cluster per variable,
and so has the smallest possible maximum clique size $k^*$
and each sepset reduced to a single variable.
Various algorithms have been proposed for constructing cluster graphs
along the spectrum (e.g.\ LTRIP \cite{streicher2017graph}) (Fig.~\ref{fig:net2}).
Notably, join-graph structuring \cite{mateescu2010join} spans the whole
spectrum because it is controlled by a user-defined maximum cluster size $k^*$,
which can be varied from its smallest possible value to a value large enough
to obtain a clique tree.

At the other end of the spectrum, the best maximum clique size $k$ is $1+\tw(G^m)$,
where $\tw(G^m)$ is the treewidth of the moralized graph.
Loopy BP becomes interesting when $\tw(G^m)$ is large, making exact BP costly.
Unfortunately, determining the treewidth of a general graph is NP-hard
\cite{arnborg1987complexity,bodlaender2010treewidth}.
Heuristics such as greedy minimum-fill
or nested dissection \cite{strasser2017computing,hamann2018graph} can be used
to obtain clique trees whose maximum clique size $k$ is near the optimum
$1+\tw(G^m)$.

Different cluster graph algorithms could potentially be applied to the different
biconnected components, or \emph{blobs} \cite{gusfield2007decomposition}
(e.g.\ LTRIP for one blob, clique tree for another),
perhaps based a blob's attributes that are easy to compute.
To choose between loopy versus exact BP, or between different cluster graph
constructions more generally, one could use traditional complexity measures of
phylogenetic networks as potential predictors of cost-effectiveness.
For example, the reticulation number $h$ is straightforward to compute.
In a \emph{binary} network, where all
internal non-root nodes have degree 3, $h$ is simply the
number of hybrid nodes. More generally
$h=|\{\mbox{hybrid edges}\}| - |\{\mbox{hybrid nodes}\}|$
\cite{van2010phylogenetic}.
The \emph{level} of a network is the maximum reticulation number
within a blob \cite{gambette2009structure}.
The network's level ought to predict treewidth better than $h$
because a graph's treewidth equals the maximum treewidth of its blobs
\cite{bodlaender1998partial},
and moralizing the network does not affect its nodes' blob membership.
These phylogenetic complexity measures do not predict treewidth perfectly \cite{2022ScornavaccaWeller-treewidth-parsimony}
except in simple cases as shown below, proved in SM section~B.

\begin{proposition}\label{thm:level-hybridladder-treewidth}
Let $N$ be a binary phylogenetic network with $h$ hybrid nodes, level $\ell$,
and let $t$ be the treewidth of the moralized network $N^m$ obtained from $N$.
For simplicity, assume that $N$ has no parallel edges and no degree-2 nodes
other than the root.
\begin{enumerate}
\item[(A0)] If $\ell=0$ then $h=0$ and $t=1$.
\item[(A1)] If $\ell=1$ then $h\geq 1$ and $t=2$.
\item[(A2)] Let $v_1$ be a hybrid node with non-adjacent parents $u_1,u_2$.
If $v_1$ has a descendant hybrid node $v_2$ such that one of its
parents is not a descendant of either $u_1$ or $u_2$,
then $\ell\geq 2$ and $t\geq 3$.
\end{enumerate}
\end{proposition}

\begin{figure}
    \begin{center}
      \includegraphics{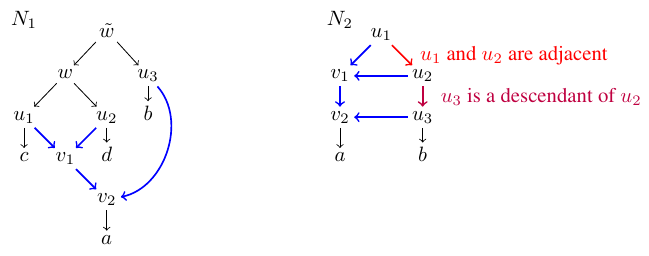}
    \end{center}
    \caption{Two binary networks with a hybrid ladder and $h=\ell=2$.
    $N_1$ satisfies (A2) of
    Proposition~\ref{thm:level-hybridladder-treewidth} and $N_1^m$ has treewidth $t=3$.
$N_2$ does not meet (A2) (see red/purple annotations) and
    $N_2^m$ has treewidth $t=2$.
    Stacking more hybrid
    ladders in the same way above $a$ and $b$ increases $h$ and $\ell$ but
    leaves $N_2^m$ outerplanar, keeping $t=2$.}
    \label{fig:hybridladder}
\end{figure}

Level-1 networks have received much attention in phylogenetics because
they are identifiable under various models under some mild restrictions
\cite{2016SolisLemus-SNaQ,Banos2019,2020Gross,2023XuAne_identifiability}.
Several inference methods limit the search to level-1 networks
\cite{2016SolisLemus-SNaQ,2016Oldman-TriLoNet,2019Allman-nanuq,2023Kong-phynest}.
Since moralized level-1 networks have treewidth 2, exact BP is guaranteed to
be efficient on them.

Beyond level-1, a network has a hybrid ladder
(also called \emph{stack} \cite{2018Semple-amalgamation2trees})
if a hybrid node $v_1$ has a hybrid child node $v_2$.
By Proposition~\ref{thm:level-hybridladder-treewidth},
a hybrid ladder has the potential
to increase treewidth of the moralized network and decrease BP scalability,
if the remaining conditions in (A2) are met.
Related results in \cite{chaplick2023snakes} are for undirected
graphs that do not require prior moralization, and contain ladders defined
as regular $2\times L$ grids. Their Observation~1, that a graph containing a
non-disconnecting grid ladder of length $L\geq 2$ has treewidth at least 3,
relies on a similar
argument as for (A2). However, structures leading to the conditions in (A2)
are more general, even before moralization.
It may be interesting to extend some of the results from \cite{chaplick2023snakes}
to moralized hybrid ladders in rooted networks.

In Fig.~\ref{fig:hybridladder} (right) $N_2$ has a hybrid ladder
that does not meet all conditions of (A2), and has $t=2$.
Generally, outerplanar networks have treewidth at most $2$ \cite{bodlaender1998partial},
and if bicombining (hybrid nodes have exactly 2 parents),
remain outerplanar after moralization.
Networks in which no hybrid node is the descendant of another
hybrid node in the same blob are called \emph{galled networks}
\cite{2010HusonRuppScornavacca}.
They provide more tractability to solve the
cluster containment problem \cite{huson2009computing}.
Here, galled networks would then never meet the assumptions of (A2)
and it would be interesting to study their treewidth after moralization.

\medskip

We performed an empirical investigation of how $h$ and $\ell$
can predict the treewidth $t$ of the moralized network.
Fig.~\ref{fig:realsim nets} shows that $t$
correlates with $h$ and $\ell$,
on networks estimated from real data using various inference methods and on
networks simulated under the biologically realistic birth-death-hybridization model
\cite{2024JustisonHeath-exploring,justison2023siphynetwork},
especially for complex networks.
For networks with hundreds of tips
(\cite{thorson2023identifying} lists several studies of this size),
large maximum clique sizes $k\ge 30$ are not uncommon.
In contrast, a Bethe cluster graph
would have maximum cluster size $k^*=3$, so that $(k/k^*)^3\ge 10^3$ would
provide a large computational gain for loopy BP to be considered.

\begin{figure}
  \begin{center}
  \includegraphics[width = 0.8\textwidth]{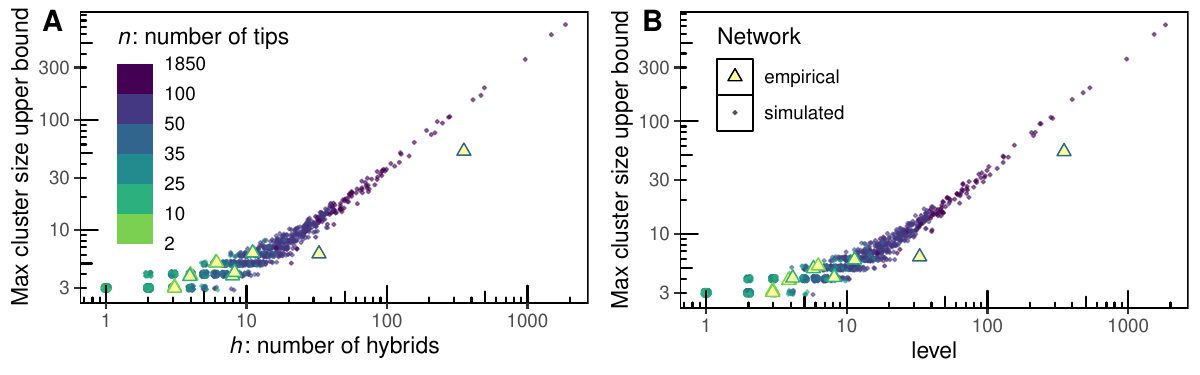}
  \end{center}
\caption{We observe a positive sublinear relationship between a maximum clique
  size upperbound (from the greedy minfill heuristic) and the number of hybrids $h$ (A) or
  network level $\ell$ (B) on a combined sample of 11 empirical networks and 2509 simulated
  birth-death-hybridization networks.
  The empirical networks were sampled from
  \cite[Figs.~3(a-c) (left), 4(a-c) (left)]{2023Maier} (reported as
  estimated by \cite{bergstrom2020origins,librado2021origins,hajdinjak2021initial,
  lipson2020ancient,wang2021genomic,sikora2019population}),
  \cite[Fig.~3]{lazaridis2014ancient},
  \cite[Fig.~3 (left)]{2023Nielsen-admixturebayes},
  \cite[Fig.~4(c)]{sun2023ancient}, \cite[Fig.~1(a)]{muller2022bayesian},
  \cite[Fig.~5(a)]{2022Neureiter};
  fit by these authors using
  \texttt{ADMIXTOOLS} \cite{2012Patterson,2023Maier}, \texttt{admixturegraph} \cite{leppala2017admixturegraph}, \texttt{OrientAGraph} \cite{molloy2021OrientAGraph}, \texttt{contacTrees} \cite{2022Neureiter}, \texttt{Recombination} \cite{muller2022bayesian}, \texttt{AdmixtureBayes} \cite{2023Nielsen-admixturebayes}. The simulated networks were obtained by subsampling 10 networks per parameter scenario
  simulated by \cite{2024JustisonHeath-exploring}, then filtering out networks of
  treewidth 1 (trees, possibly with parallel hybrid edges).
The graphs are similar at high values because most networks have most of their
  hybrids contained in one large blob, leading to $h\approx\ell$. For example,
  $|\ell - h| \le 2$
  in 99\% of the networks with $h>10$. }
  \label{fig:realsim nets}
\end{figure}

\subsubsection{Approximation quality with loopy BP}
\label{ex:lbp error}
    We simulated data on a complex graph (40 tips, 361 hybrids) \cite[Fig.~1(a)]{muller2022bayesian} and a simpler graph (12 tips, 12 hybrids) \cite[Extended Data Fig.~4]{lipson2020ancient},
    then compared estimates from exact and loopy BP. For both networks,
    edges of length 0 were assigned the minimum non-zero edge length
    after suppressing any non-root degree-2 nodes.
    Trait values $\mathrm{\mathbf{x}}=(\mathrm{x}_1,\dots,\mathrm{x}_n)$ at the
    tips were simulated from a BM with rate $\sigma^2=1$
    and $\mathrm{x}_\rho=0$ at the root.
    Figure~\ref{fig:loopyBPapprox} shows the exact and approximate log-likelihood
    and conditional mean and variance of $x_\rho$
    assuming a BM with rate $\sigma^2=1$ but improper prior
    $x_\rho\sim\mathcal{N}(0,\infty)$,
    using a greedy minimum-fill clique tree
    $\mathcal{U}$ and a cluster graph $\mathcal{U}^*$.
    Using a factor graph, calibration failed for the complex network
    (SM section C, Fig.~S3),
    so we used join-graph structuring to build $\mathcal{U}^*$.
    $\mathcal{U}$ can be calibrated in one iteration and the
    calculated quantities are exact (horizontal lines).
In contrast, $\mathcal{U}^*$ requires multiple passes and gives approximations.
Calibration required more iterations on the complex network ($h=361$)
    than on the simpler network ($h=11$), as expected.
    But for both networks, the factored energy \eqref{eq:factoredenergy}
    approximated the log-likelihood very well.
    The distribution of the root state $x_\rho$ conditional on the data
    seems more difficult to approximate. The conditional mean was correctly
    estimated but required more iterations than the log-likelihood approximation
    on the complex network.
    The conditional variance was severely overestimated on the complex network and
    very slightly overestimated on the simpler network.
    As desired, the average computing time per belief update was lower on
    $\mathcal{U}^*$, although modestly so due to the clique tree $\mathcal{U}$ having many small clusters of size
    similar to those in $\mathcal{U}^*$
    (Fig.~S4).

\begin{figure}[h]
    \centering
    \includegraphics[scale=0.6]{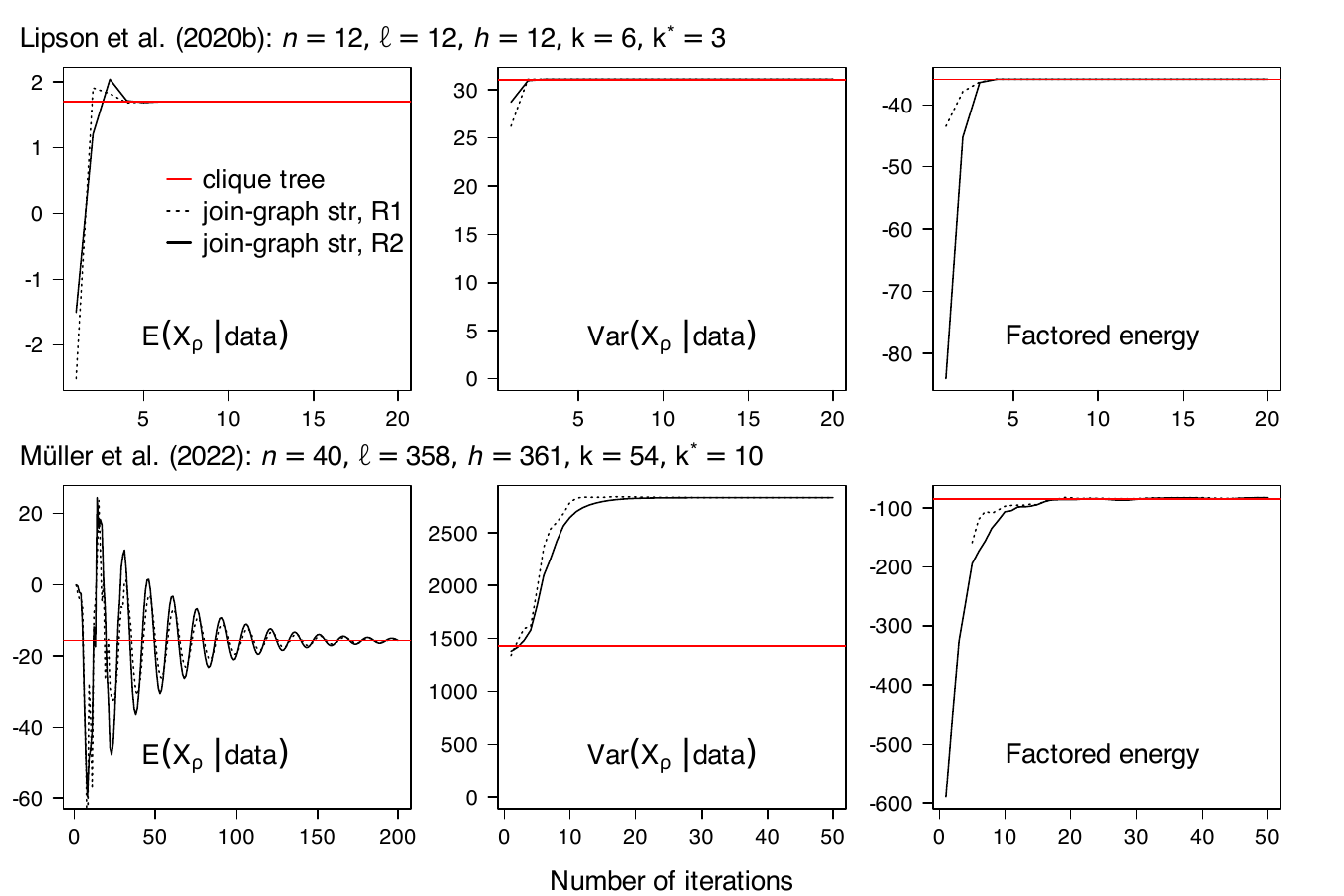}
  \caption{Accuracy of loopy BP.
  Approximation of the conditional distribution
  of the root state $X_\rho$ (left and center)
  and log-likelihood (right)
  using a greedy
  minimum-fill clique tree $\mathcal{U}$ and a join-graph structuring cluster
  graph $\mathcal{U}^*$ for two networks of varying complexity
  \cite{muller2022bayesian,lipson2020ancient} as measured by their
  number of tips ($n$), level ($\ell$), number of hybrids ($h$), maximum clique size ($k$),
  and maximum cluster size ($k^*$).
  For $\mathcal{U}$, estimates are exact after one iteration
  and shown as horizontal red lines.
For $\mathcal{U}^*$, estimates are shown over
  20 (first row), 50 or 200 (second row) iterations.
  Each iteration consists of two passes through each
  spanning tree in a minimal set that jointly covers $\mathcal{U}^*$.
In each plot, the two curves correspond to two different regularizations
  of initial beliefs (SM section~E,
  dotted: algorithm R1, solid: algorithm R2).
}
  \label{fig:loopyBPapprox}
\end{figure}

\section{Leveraging BP for efficient parameter inference}
\label{sec:inferencewithBP}

\subsection{BP for fast likelihood computation}

In some particularly simple models, such as the BM on a tree,
fast algorithms such as IC \cite{felsenstein85} or
\texttt{phylolm} \cite{2014HoAne-algo} can directly calculate
the best-fitting parameters that maximize the restricted likelihood (REML),
in a single tree traversal avoiding numerical optimization.
For more general models, such closed-form estimates are not available.
One product of BP is the likelihood of any fixed set of model parameters.
BP can hence be simply used as a fast algorithm for
likelihood computation, which can then be exploited by any statistical estimation technique,
in a Bayesian or frequentist framework.
Most of the tools cited in section~\ref{sec:tree_BP_continuous}
use either direct numerical optimization of the likelihood
\cite{Mitov2019,2023Boyko-hOUwie,2023Bartoszek-modelselection-mvOU}
or sampling techniques such as Markov Chain Monte Carlo (MCMC)
\cite{Pybus2012,fitzjohn2012diversitree}
for parameter inference.

BP also outputs the trait distribution
at internal, unobserved nodes conditioned on the observed data at the tips.
In addition to providing a tool for efficient ancestral state reconstruction,
these conditional means and variances can be used for parameter inference,
with approaches based on latent variable models such as
Expectation Maximization (EM) \cite{Bastide2017},
or Gibbs sampling schemes \cite{Cybis2015}.
Although not currently used in the field of evolutionary biology to our knowledge, 
approaches based on approximate EM algorithms \cite{Heskes2003}
and relying on loopy BP could also be used.

The linear Gaussian framework can also be useful for traits
that rely on latent Gaussian liabilities, such as the threshold model for discrete traits (Example~\ref{ex:threshold}),
factor analysis \cite{Tolkoff2018}, and phylogenetic structural equation models
\cite{thorson2023identifying}. In a Bayesian context, the latent liabilities can be sampled at the tips,
and then the conditional likelihood can be computed efficiently with Gaussian BP.
A similar approach was successful on trees \cite{Cybis2015,Zhang2021,Zhang2023}.
The BP framework can generalize such methods to phylogenetic networks.

\subsection{BP for fast gradient computation}

As we show below,
the conditional means and variances at ancestral nodes
can be used to efficiently compute the gradient of the likelihood
\cite{Salakhutdinov2003}.
The gradient of the likelihood can help speed up inference
in many different statistical frameworks \cite{Barber2012}.
In a phylogenetic context, gradients have been used to improve
maximum likelihood estimation \cite{1996FelsensteinChurchill,Ji2019},
Bayesian estimation through Hamiltonian Monte Carlo (HMC) approaches
\cite{Zhang2021,Fisher2021,2021Bastide-HMC},
or variational Bayes approximations \cite{Fourment2019a}.
Although automatic differentiation can be used on trees for
some models \cite{Swanepoel2022}, 
direct computations of the gradient using BP-like
algorithms have been shown to be more efficient in some
contexts \cite{Fourment2023}. 
After recalling Fisher's identity to calculate gradients
after BP calibration,
we illustrate its use
on the BM model (univariate or multivariate)
where it allows for the derivation of a new analytical formula for
the REML parameter estimates.

\subsubsection{Gradient Computation with Fisher's Identity}

In a phylogenetic context, latent variables are usually internal
nodes, while observed variables are leaves.
We write
$\vect{Y} = \set{X_{v,j} : \text{trait } j \text{ observed at } v\in V}$
the set of observed variables.
Fisher's identity provides a way to link the gradient of the
log-likelihood of the data $\loglik(\theta)=\log p_{\theta}(\vect{Y})$ at parameter $\theta$,
with the distribution
of all the variables conditional on the observations $\vect{Y}$.
We refer to \cite[chap.\ 10]{Cappe2005}
or \cite[chap.\ 11]{Barber2012} for general introductions
on Markov models with latent variables.
Under broad assumptions, Fisher's identity states
(see Proposition~10.1.6 in \cite{Cappe2005}, or section~11.6 in \cite{Barber2012}):
  \begin{equation*}\label{eq:fisher_identity}
      \left. 
        \gradient{\theta'} \left[
          \log p_{\theta'}(\vect{Y})
      \right]
      \right|_{\theta'=\theta}
      = \E_\theta \left[
        \cond{
          \left.
        \gradient{\theta'} \left[
          \log p_{\theta'}(\vect{X}_v; v\in V)
          \right]
          \right|_{\theta'=\theta}
        }{
          \vect{Y}
        }
        \right],
  \end{equation*}
where $\left.\gradient{\theta'} [f(\theta')]\right|_{\theta'=\theta}$ denotes the gradient 
of $f$
with respect to the generic parameters $\theta'$ and evaluated at $\theta' = \theta$,
and $\E_\theta[\cond{\bullet}{\vect{Y}}]$ the expectation
conditional on the observed data
under the model parametrized by $\theta$,
which is precisely where the output from BP can be used.
Plugging in the factor decomposition from the graphical
model \eqref{eq:factor_decomposition}
we get:
\begin{equation}\label{eq:fisher_gradient}
  \left. 
    \gradient{\theta'} \left[
      \log p_{\theta'}(\vect{Y})
  \right]
  \right|_{\theta'=\theta}
  = 
  \sum_{v\in V}
  \E_\theta \left[
    \cond{
      \left.
        \gradient{\theta'} \left[
          \log \phi_v(\vect{X}_v | \vect{X}_u, \theta'; u\in\pa(v))
          \right]
          \right|_{\theta'=\theta}
    }{
      \vect{Y}
    }
    \right].
\end{equation}
While \eqref{eq:fisher_gradient} applies to the full vector of
all model parameters, it can also be applied to take the gradient with respect
to a single parameter $\theta$ of interest,
keeping the other parameters fixed.
For instance, we can focus on one rate matrix
$\bm{\Sigma}$
of a BM model, or one primary optimum of an OU model.
Special care needs to be taken for
gradients with respect to structured matrices,
such as variance matrices that need to be symmetric (see e.g.\ \cite{2021Bastide-HMC})
or with a sparse inverse under structural equation modeling for
high dimentional traits \cite{2023Thorson-phylosem}.

For models where the conditional expectation of the factor in \eqref{eq:fisher_gradient}
has a simple
form, this formula is the key to an efficient gradient computation.
In particular, for discrete traits as in Example~\ref{ex:net1},
the expectation becomes a sum of a manageable number of terms,
local to a cluster, weighted by the normalized cluster belief after calibration
\cite[ch.~19]{koller2009probabilistic}.

\subsubsection{Gradient computation for linear Gaussian models}
For linear Gaussian models~\eqref{eq:lineargaussian},
log-factors can be written as quadratic forms~\eqref{eq:canonicalform},
so their derivatives have analytical formulas
(see SM section~D). The conditional expectation in~\eqref{eq:fisher_gradient}
then only depends on the joint first and second order moments of the variables
$(\vect{X}_v, \vect{X}_{\pa(v)})$
in a cluster, which are known as soon as the beliefs are calibrated.
When the graph is a tree, \cite{2021Bastide-HMC} exploited this
formula to derive gradients in the general linear Gaussian case.
However, they did not use the complete factor decomposition~\eqref{eq:factor_decomposition},
but instead an \emph{ad-hoc} decomposition that only works when the graph is a
(binary) tree, and exploits the split partitions defined by the tree.
In contrast, the present approach gives a recipe for the efficient gradient
computation of any linear Gaussian model on any network,
as soon as beliefs are calibrated.

In the special case where the process is a homogeneous BM
(univariate or multivariate) on a network with a weighted-average
merging rule~\eqref{eq:weighedaverage-merging},
a constant rate $\bm{\Sigma}$,
no missing data at the tips, and, if present, within-species variation
that is proportional to $\bm{\Sigma}$,
then the gradient with respect to $\bm{\Sigma}$
takes a particularly simple form.
Setting this gradient to zero, we find an analytical formula for the
REML estimate of $\bm{\Sigma}$
and for the ML estimate of the ancestral mean $\mu_\rho$
(SM section~D.3).
In a phylogenetic regression setting, a similar formula can be found
for the ML estimate of coefficients
(SM section~D.4).
Efficient algorithms such as IC and \texttt{phylolm} already exist
to compute these quantities on a tree, in a single traversal.
Here, our formulas need two traversals but
remain linear in the number of tips,
and because they rely on a general BP formulation,
they apply to networks with reticulations.
Fisher's identity and BP hence offer a general method for gradient
computation, and could lead to analytical formulas for
other simple models.
Such efficient formulas could alleviate numerical instabilities observed
in software such as \texttt{mvSLOUCH}, which experienced a significant failure rate
for the BM on trees with a large number of traits \cite{Bartoszek2024}.

\subsubsection{Hessian computation with Louis's identity}
Second order derivatives can improve both maximum likelihood
and Bayesian statistical inference methods. Efficient Hessian computation methods
have been developed for branch lengths in sequence evolution models \cite{1996FelsensteinChurchill,Ji2019},
and recently for continuous trait model parameters on trees \cite{2024Kiang-hessian}.
Using BP, the Hessian of the log-likelihood with respect
to the parameters can also be obtained as a conditional expectation
of the Hessian of the complete log-likelihood:
\begin{multline*}\label{eq:louis_identity}
  \left.
    \left\{
  \gradient{\theta'}^2 \left[
      \log p_{\theta'}(Y)
  \right]
  +
  \gradient{\theta'} \left[
      \log p_{\theta'}(Y)
  \right]
  \left[
  \gradient{\theta'} \left[
      \log p_{\theta'}(Y)
  \right]\right]^{\top}
  \right\}
  \right|_{\theta'=\theta}
  = \\
  \E_\theta \left[
    \cond{\big\{
    \gradient{\theta'}^2 \left[
      \log p_{\theta'}(X_v; v\in V)
      \right]
      +
      \gradient{\theta'} \left[
        \log p_{\theta'}(X_v; v\in V)
      \right]
      \left[
      \gradient{\theta'} \left[
        \log p_{\theta'}(X_v; v\in V)
      \right]\right]^{\top}
    \big\}|_{\theta'=\theta}}
    {
      Y
    }
    \right].
\end{multline*}
This so-called Louis identity \cite{Cappe2005} also simplifies under
the factor decomposition~\eqref{eq:factor_decomposition},
and leads to tractable formulas in simple Gaussian or discrete cases.

\subsection{BP for direct Bayesian parameter inference}

Likelihood or gradient-based approaches require careful
analytical computations to get exact formulas in any new model
within the class of linear Gaussian graphical models, depending on the parameters
of interest \cite{2021Bastide-HMC}.
One way to alleviate this problem is to use a Bayesian framework,
and expand the graphical model to include
both the phylogenetic network and the evolutionary parameters,
which are seen as random variables themselves,
as e.g., in \cite{hohna2014probabilistic}.
Then, inferring parameters amounts to learning their conditional distribution in
this larger graphical model. In this setting, the output of interest
from BP is not the likelihood but the distribution of random variables
(evolutionary parameters primarily) conditional on the observed data.

Exact computation may not be possible in this extended graphical model,
because it is typically not linear Gaussian and the graph's treewidth can be much larger than that of the phylogenetic network,
when one parameter (e.g.\ the evolutionary rate) affects multiple node families.
Therefore, approximations may need to be used.
For example, ``black box'' optimization techniques
rely on variational approaches to reach a tractable approximation
of the posterior distribution of model parameters \cite{Ranganath2014}.
The conditional distribution of unobserved variables, provided
by BP, facilitates the noisy approximation of the variational gradient
that can be used to speed up the optimization of the variational Bayes approximation.

\section{Challenges and Extensions}
\label{sec:discussion}

\subsection{Degeneracy}

While our implementation
provides a proof-of-concept, various technical challenges still need to be solved.
Much of the literature on BP focuses on factor graphs, which failed to
converge for one of our example phylogenetic networks. More work is needed to better
understand the convergence and accuracy of alternative cluster graphs,
and on other choices that
can substantially affect loopy BP's efficiency, such as scheduling.
Below, we focus on implementation challenges due to degeneracies.

For the message $\tilde{\mu}_{i\rightarrow j}$ to be well-defined in step 1 of Gaussian BP,
the belief of the sending cluster must have a precision matrix
$\bm{K}$ in~\eqref{eq:canonicalform}
with a full-rank submatrix with respect to the
variables to be integrated out ($\bm{K}_I$ in Algorithm~\ref{alg:gbp}).
This condition can fail under realistic phylogenetic models,
due to two different types of degeneracy.

The first type arises from deterministic factors:
when $\bm{V}_v=0$ in \eqref{eq:lineargaussian}
and $X_v$ is determined by the states at parent nodes $X_{\pa(v)}$ without noise,
e.g.\ when all of $v$'s
parent branches have length 0 in standard phylogenetic models.
This is expected at hybridization events when both parents have sampled
descendants in the phylogeny, because the parents and hybrid need
to be contemporary of one another.
This situation is also common in admixture graphs \cite{2023Maier}
due to a lack of identifiability of hybrid edge lengths from $f$ statistics,
leading to a ``zipped-up'' estimated network in which the estimable
composite length parameter is assigned to the hybrid's child edge
\cite{2023XuAne_identifiability}.
With this degeneracy, $X_v$ has infinite precision given
its parents, that is, $\bm{K}$ has some infinite values.
The complications are technical, but not numerical.
For example, one can use a generalized canonical form that includes a Dirac distribution
to capture the deterministic equation of
$X_v$ given $X_{\pa(v)}$ from \eqref{eq:lineargaussian}.
Then BP operations need to be extended to these generalized canonical forms,
as done in \cite{schoeman2022degenerate}
(illustrated in SM section~F).
One could also modify the network by contracting internal tree edges of length~0.
At hybrid nodes, adding a small variance to
$\bm{V}_v$ would be an approximate yet biologically realistic approach.

The second type of degeneracy arises when the precision submatrix $\bm{K}_I$ is
finite but not of full rank.
In phylogenetic models, this is frequent at initialization \eqref{eq:clusterfactor}.
For example, consider a cluster of 3 nodes: a hybrid $v$ and its 2 parents.
By \eqref{eq:canonicalformGGM} we see that $\rk(\bm{K}_v)\leq p$.
So at initialization with belief $\phi_v$,
$\bm{K}_I$ is degenerate if we seek to integrate out $|I|=2$ nodes,
which would occur if the cluster is adjacent to a sepset
containining only one parent of $v$.
This situation is typical of factor graphs.
Initial beliefs would also be degenerate with $\bm{K}=\bm{0}$
for any cluster that is not assigned any factor by \eqref{eq:clusterfactor}.
This may occur if there are more clusters
than node families, or if the graph
has nested redundant clusters (e.g.\ from join-graph structuring).
In some cases, a schedule may avoid these degeneracies,
guaranteeing a well-defined message at each BP update.
On a clique tree, a schedule based on a postorder traversal has this guarantee,
provided that all $p$ traits are observed at all leaves
or that trait $j$ at node $v$ is removed from scope if it is unobserved
at all its descendants.
But generally, it is unclear how to find such a schedule.
Another approach is to simply skip a BP update if its message
is ill-defined, though there is no guarantee that the sending cluster will
eventually have a well-behaved belief to pass the message later.
A robust option is to regularize cluster beliefs, right after initialization
\eqref{eq:clusterfactor} or during BP, by increasing some diagonal elements of
$\bm{K}$ to make $\bm{K}_I$ of full rank. To maintain the probability model,
this cluster belief regularization is balanced by a similar modification
to a corresponding sepset.
SM section~E describes two such approaches that
appear to work well in practice, although theoretical guarantees have not been
established.

\subsection{Loopy BP is promising for discrete traits}

We focused on Gaussian models in this paper, for which the `lazy' matrix
approach is polynomial.
For discrete trait models, the computational gains from loopy BP can be much
greater, because alternative approaches are not polynomial on general networks.
For a trait with $c$ states ($c=2$ for a binary trait as in Example~\ref{ex:net1}),
passing a message has complexity $\mathcal{O}(c^k)$
where $k$ is the sending cluster size.
Thus, cluster graphs with small clusters can bring exponential computational gains.
Even exact BP can bring significant computational gains to existing approaches
that rely on other means to reduce complexity.
For example, the model without ILS used in
\cite{lutteropp2022netrax,allen2020estimating} is a mixture model, so the network likelihood can be calculated as a weighted average of tree likelihoods
for which exact BP takes linear time. This approach scales exponentially with $h$ because
there are typically $\mathcal{O}(2^h)$ trees displayed in a network.
In contrast, the complexity of BP on a clique tree of maximum clique size $k$
is $\mathcal{O}(nc^k)$, thus parametrized by the treewidth $t$ of the moralized network instead of $h$
($t = k-1$ for an optimal clique tree).
Given our empirical evidence that
$t$ grows more slowly than $h$ or the network's level $\ell$
in biologically-realistic networks (Fig.~\ref{fig:realsim nets}),
exact BP could achieve significant computational gains
and loopy BP substantially more.

A BP approach is already used in \texttt{momi2} \cite{kamm2020momi2},
who use a clique tree built from a node ordering by age from
youngest to oldest, to get conditional likelihoods of
the derived allele count under a Moran model (without mutation).
The mutation-with-ILS model in \texttt{SnappNet} can be also reframed
as a graphical model on a graph expanded from the phylogenetic network
(as shown in Example~\ref{ex:tree2} and SM section~A). Accordingly, the BP-like algorithm in \cite{2021Rabier-snappnet}
has complexity controlled by the network's scanwidth, a parameter introduced
by \cite{berry2020scanning}.
Using regular BP on more optimal clique trees and loopy BP on cluster graphs
may help speed up computations even more.

Also related is the algorithm in \cite{2022ScornavaccaWeller-treewidth-parsimony},
who use a clique tree to solve a parsimony problem.
In this non-probabilistic setting, it is unclear how cluster graphs could be
leveraged to speed up algorithms as they do in loopy BP.

\medskip
To deal with computational intractability, the most widely-used probabilistic
methods to infer networks from DNA sequences are based on composite likelihoods
\cite{2016SolisLemus-SNaQ,2015yu} or summary statistics like
$f$ statistics \cite{2023Maier,2023Nielsen-admixturebayes},
leading to a lack of identifiability for parts of the network topology
and some of its parameters
\cite{2016SolisLemus-SNaQ,Banos2019,2023XuAne_identifiability,
      2024Ane-anomalies,2024ABGarrotelopesR,2014Rhodes-circular}.
These identifiability issues should be alleviated
if using the full data becomes tractable thanks to exact or loopy BP.

\section*{Supplementary Material}

Technical derivations are available in the Supplementary Material (SM).
Code to reproduce Figures~\ref{fig:net2}a, \ref{fig:realsim nets} and
\ref{fig:loopyBPapprox} is available at
\url{https://github.com/bstkj/graphicalmodels_for_phylogenetics_code}.
A julia package for Gaussian BP on phylogenetic networks
is available at \url{https://github.com/JuliaPhylo/PhyloGaussianBeliefProp.jl}.

\maketitle

\enlargethispage{20pt}

\section*{Acknowledgements}
This work was supported in part by the National Science Foundation (DMS 2023239 to C.A.)
and
by the University of Wisconsin-Madison Office of the Vice Chancellor for Research and
Graduate Education with funding from the Wisconsin Alumni Research Foundation.
C.A. visited P.B. at the University of Montpellier thanks to support from
the I-SITE MUSE through the Key Initiative ``Data and Life Sciences''.

\clearpage{}\let\olditemize\itemize
\renewcommand\itemize{\olditemize\setlength{\itemsep}{1pt}}

\setcounter{equation}{0}
\renewcommand{\theequation}{SM-\arabic{equation}} 
\setcounter{theorem}{1}
\renewcommand{\thealgorithm}{R\arabic{algorithm}}

\appendix
\renewcommand{\thesection}{\Alph{section}}
\setcounter{figure}{0}
\renewcommand{\thefigure}{S\arabic{figure}}

\hrule
\begin{center}
\textsc{\Large Supplementary material}
\end{center}
\hrule

\renewcommand{\shorttitle}{SM: Leveraging graphical model techniques to
study evolution on phylogenetic networks}

\section{Example 3: Binary trait with ILS on a tree (\texttt{SNAPP}) and on a network (\texttt{SnappNet})}
\label{sec:snapp_snappnet_BP}

In this section, we recast the likelihood calculations done
in \texttt{SNAPP} and \texttt{SnappNet} as belief propagation (BP),
on a graph $G$ that differs from the species phylogeny due to
incomplete lineage sorting (ILS).

Following Example~3 in the main text,
we consider a sample of one or more individuals from each species in a species
phylogeny (which may be reticulate).
For a given gene, we assume that the gene tree is
generated according to the coalescent model along the species phylogeny
(see \cite{1982kingman} for the coalescent in one population,
\cite{2003RannalaYang} for the multispecies coalescent along a species tree, and
\cite{2023fogg_phylocoalsimulations} for the network multispecies coalescent
along a species network).
Finally, consider a binary trait evolving along this gene tree,
with 2 states (or alleles) ``green'' and ``red'',
to re-use terminology by \cite{2012Bryant-snapp}.

We denote the total number of alleles ($n$) and number of red alleles
($r$) ancestral to the sampled individuals at the beginning and at the end of
edge $e$, by $n_{\overline{e}},r_{\overline{e}}$ and $n_{\underline{e}},r_{\underline{e}}$
respectively. Here ``beginning'' and ``end'' of $e$ refer to points
in time with time going forward, along the species lineage that the edge $e$
represents. The ``beginning'' of $e$ corresponds to the species (or population)
right after the speciation or hybridization event at which $e$ starts,
closer to the root.
The ``end'' of $e$ corresponds to the species (or population) right before
the speciation or hybridization event at which $e$ ends, closer to the leaves.
We only observe the total number of alleles and the number of red alleles
at the end of pendant edges, that is, in the sampled extant species.

Below, we start by formulating this evolutionary model as a graphical model
for the allele counts $n_{\overline{e}},r_{\overline{e}}$
and $n_{\underline{e}},r_{\underline{e}}$,
which require a graph $G$ that is more complex than the original species
phylogeny (e.g.\ Figs.~\ref{fig:tree2} and \ref{fig:net1SM}).

\subsection{\texttt{SNAPP}}
\label{sec:snapp_BP}
\cite{2012Bryant-snapp} discovered conditional independencies among
the allele counts tracked by the model, 
$n_{\overline{e}},r_{\overline{e}}$ (at the beginning of edge $e$)
and $n_{\underline{e}},r_{\underline{e}}$ (at the end of $e$),
when the species phylogeny is a tree $T$.
We show here that these conditional independencies
correspond to a graphical model on a graph $G$ that is different from the
original phylogenetic tree.
This is illustrated in Fig.~\ref{fig:tree2}
with $T$ in (a), valid graphs in (b-c), and BP run along the clique tree $\mathcal{U}$ in (d).

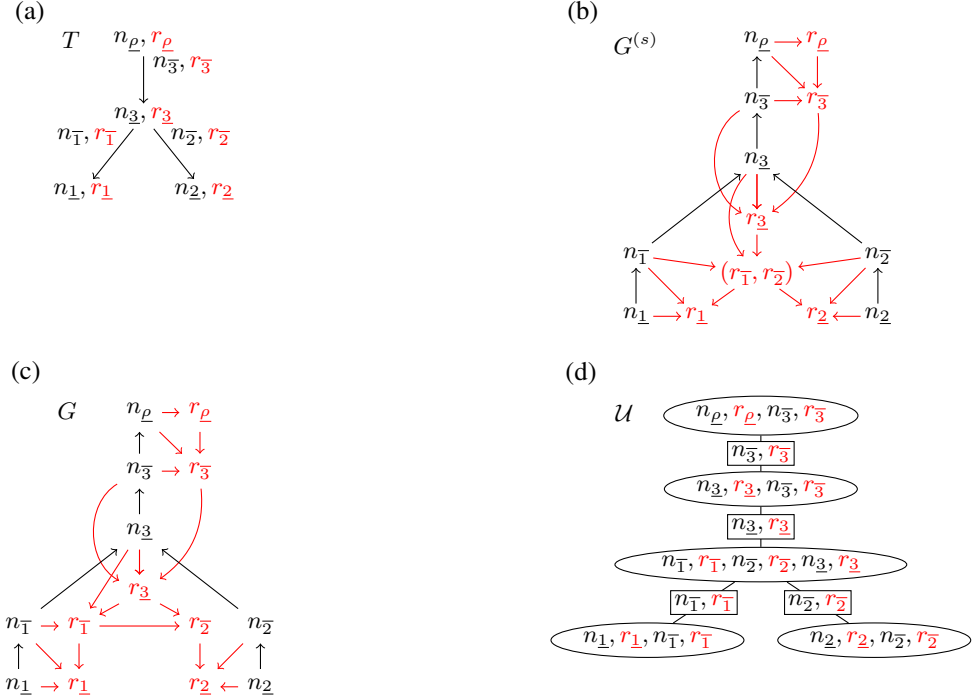
\begin{figure}[h]
  \begin{center}
    \begin{tikzpicture}[-]
\matrix (m0)[matrix,anchor=north,
            column sep={0.8cm,between origins},
            row sep={1cm,between origins},
            nodes={font=\small, inner sep=.5mm}] at (1,0)
    {\node[label={[xshift=-2em]:\normalsize (a)}](){\hspace{-1em}$T$}; &
    \node(x0){$n_{\underline{\rho}},\textcolor{red}{r_{\underline{\rho}}}$}; & \\
    & \node(x3){$n_{\underline{3}},\textcolor{red}{r_{\underline{3}}}$}; & \\
    \node(x1){$n_{\underline{1}},\textcolor{red}{r_{\underline{1}}}$}; & &
    \node(x2){$n_{\underline{2}},\textcolor{red}{r_{\underline{2}}}$}; \\
    };

    \matrix (m1)[matrix,anchor=north,
      column sep={8mm,between origins},
      row sep={8mm,between origins},
      nodes={font=\small, inner sep=.5mm}] at (9,0)
    {\node[label={[xshift=-2em]:\normalsize (b)}]{$G^{(s)}$}; & &
    \node(n0b_1){$n_{\underline{\rho}}$}; &
    \node[red](r0b_1){$r_{\underline{\rho}}$}; & \\
    & & \node(n3t_1){$n_{\overline{3}}$}; & \node[red](r3t_1){$r_{\overline{3}}$}; & \\
    & & \node(n3b_1){$n_{\underline{3}}$}; & & \\
    & & \node[red](r3b_1){$r_{\underline{3}}$}; & & \\[-1em]
    \node(n1t_1){$n_{\overline{1}}$}; & &
    \node[red,yshift=-.5em](r1tr2t_1){$(r_{\overline{1}},r_{\overline{2}})$}; & &
    \node(n2t_1){$n_{\overline{2}}$}; \\
    \node(n1b_1){$n_{\underline{1}}$}; & \node[red](r1b_1){$r_{\underline{1}}$}; & &
    \node[red](r2b_1){$r_{\underline{2}}$}; & \node(n2b_1){$n_{\underline{2}}$}; \\
    };

    \matrix (m2)[matrix,anchor=north,
        column sep={8mm,between origins},
        row sep={8mm,between origins},
        nodes={font=\small}] at (1.2,-4.7)
    {
    & \node[label={[xshift=-2em]:\normalsize (c)}]{\hspace{-1em}$G$};
    & \node(n0b){$n_{\underline{\rho}}$}; & \node[red](r0b){$r_{\underline{\rho}}$}; & \\
    & & \node(n3t){$n_{\overline{3}}$}; & \node[red](r3t){$r_{\overline{3}}$}; & \\
    & & \node(n3b){$n_{\underline{3}}$}; & & \\
    & & \node[red](r3b){$r_{\underline{3}}$}; & & \\[-1em]
    \node(n1t){$n_{\overline{1}}$}; & \node[red](r1t){$r_{\overline{1}}$}; & &
    \node[red](r2t){$r_{\overline{2}}$}; & \node(n2t){$n_{\overline{2}}$}; \\
    \node(n1b){$n_{\underline{1}}$}; & \node[red](r1b){$r_{\underline{1}}$}; & &
    \node[red](r2b){$r_{\underline{2}}$}; & \node(n2b){$n_{\underline{2}}$}; \\
    };

    \matrix (m3)[matrix,anchor=north ,
        column sep={1.5cm,between origins},
        row sep={1cm,between origins},
        nodes={font=\small,ellipse,draw,inner sep=.5mm}] at (9.42,-4.7)
    {\node[draw=none,label={[xshift=-2.5em]:\normalsize (d)}](){\hspace{-1.9em}$\mathcal{U}$}; &
      \node(c0b){$n_{\underline{\rho}},\textcolor{red}{r_{\underline{\rho}}},
      n_{\overline{3}},\textcolor{red}{r_{\overline{3}}}$}; & \\
      & \node(c3t){$n_{\underline{3}},\textcolor{red}{r_{\underline{3}}},
      n_{\overline{3}},\textcolor{red}{r_{\overline{3}}}$}; & \\
      & \node(c3b){$n_{\overline{1}},\textcolor{red}{r_{\overline{1}}},
      n_{\overline{2}},\textcolor{red}{r_{\overline{2}}},
      n_{\underline{3}},\textcolor{red}{r_{\underline{3}}}$}; & \\
      \node(c1t){$n_{\underline{1}},\textcolor{red}{r_{\underline{1}}},
      n_{\overline{1}},\textcolor{red}{r_{\overline{1}}}$}; & &
      \node(c2t){$n_{\underline{2}},\textcolor{red}{r_{\underline{2}}},
      n_{\overline{2}},\textcolor{red}{r_{\overline{2}}}$}; \\
    };
    \begin{scope}[every node/.style={font=\small\itshape}]
    \draw[->] (x0) -- node[above right,
        yshift=-.1em]{$n_{\overline{3}},\textcolor{red}{r_{\overline{3}}}$}(x3);
    \draw[->] (x3) -- node[left,above,
        xshift=-1em]{$n_{\overline{1}},\textcolor{red}{r_{\overline{1}}}$}(x1);
    \draw[->] (x3) -- node[right,above,
        xshift=1em]{$n_{\overline{2}},\textcolor{red}{r_{\overline{2}}}$}(x2);

    \draw[->,red] (n0b_1) -- (r0b_1);
    \draw[->,red] (n0b_1) -- (r3t_1);
    \draw[->,red] (r0b_1) -- (r3t_1);
    \draw[->,red] (n3t_1) -- (r3t_1);
    \draw[->] (n3t_1) -- (n0b_1);
    \draw[->] (n3b_1) -- (n3t_1);
    \draw[->,red,bend right=60] (n3t_1) to (r3b_1);
    \draw[->,red] (n3b_1) -- (r3b_1);
    \draw[->,red,bend left=30] (r3t_1) to (r3b_1);
    \draw[->,red] (n1b_1) -- (r1b_1);
    \draw[->] (n1b_1) -- (n1t_1);
    \draw[->,red] (n1t_1) -- (r1b_1);
    \draw[->,red] (n2b_1) -- (r2b_1);
    \draw[->] (n2b_1) -- (n2t_1);
    \draw[->,red] (n2t_1) -- (r2b_1);
    \draw[->] (n1t_1) -- (n3b_1);
    \draw[->] (n2t_1) -- (n3b_1);
    \draw[->,red] (n3b_1) -- (r3b_1);
    \draw[->,red, bend right=40] (n3b_1) to (r1tr2t_1); \draw[->,red] (r1tr2t_1) -- (r1b_1);
    \draw[->,red] (r1tr2t_1) -- (r2b_1);
    \draw[->,red] (n1t_1) -- (r1tr2t_1);
    \draw[->,red] (n2t_1) -- (r1tr2t_1);
    \draw[->,red] (r3b_1) -- (r1tr2t_1);

    \draw[->,red] (n0b) -- (r0b);
    \draw[->,red] (n0b) -- (r3t);
    \draw[->,red] (r0b) -- (r3t);
    \draw[->,red] (n3t) -- (r3t);
    \draw[->] (n3t) -- (n0b);
    \draw[->] (n3b) -- (n3t);
    \draw[->,red,bend right=60] (n3t) to (r3b);
    \draw[->,red] (n3b) -- (r3b);
    \draw[->,red,bend left=30] (r3t) to (r3b);
    \draw[->,red] (n1b) -- (r1b);
    \draw[->] (n1b) -- (n1t);
    \draw[->,red] (n1t) -- (r1t);
    \draw[->,red] (r1t) -- (r1b);
    \draw[->,red] (n1t) -- (r1b);
    \draw[->,red] (n2b) -- (r2b);
    \draw[->] (n2b) -- (n2t);
    \draw[->,red] (r2t) -- (r2b);
    \draw[->,red] (n2t) -- (r2b);
    \draw[->] (n1t) -- (n3b);
    \draw[->] (n2t) -- (n3b);
    \draw[->,red] (n3b) -- (r3b);
    \draw[->,red] (n3b) -- (r1t);
    \draw[->,red] (r3b) -- (r1t);
    \draw[->,red] (r3b) -- (r2t);
    \draw[->,red] (r1t) -- (r2t);

    \draw[-] (c1t) -- (c3b) node[midway,rectangle,draw,fill=white,
        inner sep=.5mm]{$n_{\overline{1}},\textcolor{red}{r_{\overline{1}}}$};
    \draw[-] (c2t) -- (c3b) node[midway,rectangle,draw,fill=white,
        inner sep=.5mm]{$n_{\overline{2}},\textcolor{red}{r_{\overline{2}}}$};
    \draw[-] (c3b) -- (c3t) node[midway,rectangle,draw,fill=white,
        inner sep=.5mm]{$n_{\underline{3}},\textcolor{red}{r_{\underline{3}}}$};
    \draw[-] (c3t) -- (c0b) node[midway,rectangle,draw,fill=white,
        inner sep=.5mm]{$n_{\overline{3}},\textcolor{red}{r_{\overline{3}}}$};
    \end{scope}
    \end{tikzpicture}
  \end{center}
  \caption{
  (a) Phylogenetic species tree $T$, used to generate gene genealogies for
  $n_{\underline{1}}$ and $n_{\underline{2}}$ individuals sampled from species 1
  and 2 respectively.
  (b) and (c) show two valid graphs for the graphical model associated with the evolution of
  a binary trait on a gene tree drawn from the multispecies coalescent model.
  Both are DAGs with 2 sources (roots)
  $n_{\underline{1}}$ and $n_{\underline{2}}$, and 2 sinks (leaves)
  $r_{\underline{1}}$ and $r_{\underline{2}}$.
  The model restricted to variables $n_{\overline{e}}$ and $n_{\underline{e}}$
  (in black) can be described by the subgraph $G_n$
  whose nodes and edges are in black, identical in (b) and (c).
  It is a tree similar to $T$ but with reversed edge directions.
  In (b) the graph $G^{(s)}$ is symmetric, invariant to swapping
  species 1 and 2, like $T$. One of its nodes corresponds to the
  2-dimensional variable $(r_{\overline{1}},r_{\overline{2}})$.
  The edge $n_{\underline{3}}\rightarrow (r_{\overline{1}},
  r_{\overline{2}})$ is redundant since
  $n_{\underline{1}}\rightarrow (r_{\overline{1}},r_{\overline{2}})$ and
  $n_{\underline{2}}\rightarrow (r_{\overline{1}},r_{\overline{2}})$
  are already present, and the model imposes
  $n_{\overline{1}}+n_{\overline{2}}=n_{\underline{3}}$.
  In (c) the graph $G$ is not symmetric but all its nodes
  contain a univariate variable.
  (d)
  Clique tree $\mathcal{U}$ for both $G^{(s)}$ and $G$.
  The 6-variable clique is overparametrized because
  $n_{\overline{1}}+n_{\overline{2}}=n_{\underline{3}}$
  and $r_{\overline{1}}+r_{\overline{2}}=r_{\underline{3}}$,
  but reflects the symmetry of the model.
  }\label{fig:tree2}
\end{figure}

If we only consider the ancestral number of individuals $n$,
then the graph $G_n$ for the associated graphical model is as follows,
thanks to the description of the coalescent model going back in time.
For each edge $e$ in $T$, an edge is created in $G_n$ but with the reversed
direction (black subgraph in Fig.~\ref{fig:tree2}(b-c)).
On this edge, the coalescent edge factor
$\phi_{n_{\overline{e}}} = \mathbb{P}(n_{\overline{e}} \mid n_{\underline{e}})$
was derived by \cite{1984Tavare} and
is given in \cite[eq. (6)]{2012Bryant-snapp}.
Each internal node $v$ in $T$ is triplicated in $G_n$ to hold the variables
$n_{\underline{e}}$, $n_{\overline{c_1}}$ and $n_{\overline{c_2}}$,
where $e$ denotes the parent edge of $v$ and $c_1,c_2$ denote its child edges
(assuming that $v$ has only 2 children, without loss of generality).
These nodes are then connected in $G_n$ by edges from each $n_{\overline{c_i}}$
to $n_{\underline{e}}$.
The speciation factor $\phi_{n_{\underline{e}}} = \mathbbm{1}_{\{n_{\overline{c_1}}
+ n_{\overline{c_2}}\}}(n_{\underline{e}})$ expresses the relationship
$n_{\underline{e}} = n_{\overline{c_1}} + n_{\overline{c_2}}$.
Overall, $G_n$ is a tree with a single sink (leaf), multiple sources (roots),
and data at the roots.

To calculate the likelihood of the data, we add to the model the number of red alleles $r$
ancestral to the sampled individuals.
The full graph (Fig.~\ref{fig:tree2}(b-c)) contains $G_n$, with extra nodes
for the $r$ variables, and extra edges to model the process along edges and at
speciations.
In Fig.~\ref{fig:tree2}(b), $G^{(s)}$ is one such valid graph
to express the graphical model. It maintains the symmetry
found in $T$ by having one node for the 2-dimensional variable
$(r_{\overline{1}},r_{\overline{2}})$.
In Fig.~\ref{fig:tree2}(c), $G$ is another valid graph, not symmetric
but with univariate nodes. In $G$,
The node family for $r_{\underline{e}}$ includes $r_{\overline{e}}$ and both
$n_{\overline{e}}$ and $n_{\underline{e}}$. The mutation edge factor
$\phi_{r_{\underline{e}}} = \mathbb{P}(r_{\underline{e}} \mid r_{\overline{e}},n_{\overline{e}},n_{\underline{e}})$
was derived by \cite{1994GriffithsTavare} using both the coalescent and mutation
processes, and is given in \cite[eq. (16)]{2012Bryant-snapp}.
For edge $e$ in $T$ with child edges $c_1$ and $c_2$ in $T$,
the speciation factors for red alleles
$\phi_{r_{\overline{c_1}}} =
\mathbb{P}(r_{\overline{c_1}} \mid n_{\overline{c_1}},n_{\underline{e}}, r_{\underline{e}})$
and
$\phi_{r_{\overline{c_2}}}=\mathbbm{1}_{\{r_{\underline{e}} -r_{\overline{c_1}}\}}(r_{\overline{c_2}})$
describe a hypergeometric distribution
where $n_{\overline{c_1}}$ individuals, $r_{\overline{c_1}}$ of which are red,
are sampled from a pool of $n_{\underline{e}}$ individuals, $r_{\underline{e}}$ of which are red, and
$r_{\overline{c_2}} = r_{\underline{e}} - r_{\overline{c_1}}$.

While Fig.~\ref{fig:tree2} illustrates the model on a minimal
phylogeny, the description above shows that generally, the likelihood calculations
used in \texttt{SNAPP} \citep{2012Bryant-snapp}
correspond to BP along a clique tree.

\bigskip

In this example, beliefs are not always partial (or full) likelihoods at every step of BP,
unlike in Examples~1 and ~2 from the main text. For example, consider the first iteration of BP, with
the tip clique ${\mathcal C}_1$ containing
$(n_{\overline{1}},r_{\overline{1}})$ in $\mathcal{U}$ (Fig.~\ref{fig:tree2}(d))
sending a message to its large neighbor clique.
The belief of ${\mathcal C}_1$ is initialized with the factors
$\phi_{n_{\overline{1}}}$ and $\phi_{r_{\underline{1}}}$,
which are the probabilities of $n_{\overline{1}}$ and of
$r_{\underline{1}}$ conditional on their parents in graph $G$.
From fixing $(n_{\underline{1}},r_{\underline{1}})$ to their
observed values $(\mathrm{n}_{\underline{1}},\mathrm{r}_{\underline{1}})$, the
message sent by ${\mathcal C}_1$ in step 1 is
\[\tilde{\mu}(n_{\overline{1}},r_{\overline{1}}) =
  \mathbb{P}(n_{\overline{1}} \mid \mathrm{n}_{\underline{1}}) \,
  \mathbb{P}(\mathrm{r}_{\underline{1}} \mid r_{\overline{1}},n_{\overline{1}},
  \mathrm{n}_{\underline{1}})\;.
\]
This message is the quantity denoted by $\mathrm{F}^{\mathrm{T}}(n,r)$ in
\cite{2012Bryant-snapp}.
It is \emph{not} a partial likelihood, because it is not the
likelihood of some partial subset of the data conditional
on some ancestral values in the phylogeny. Intuitively,
this is because nodes with data below $n_{\overline{1}}$ in $G$
include both $r_{\overline{1}}$ and $r_{\overline{2}}$,
yet ${\mathcal C}_1$ does not include $r_{\overline{2}}$.
Information about $r_{\overline{2}}$ will be passed to the
root of $\mathcal U$ separately.
More generally, during the first traversal of $\mathcal U$,
each sepset belief corresponds to an $\mathrm{F}$ value in \cite{2012Bryant-snapp}:
$\mathrm{F}^{\mathrm{T}}$ for sepsets at the top of a branch $(n_{\overline{e}},r_{\overline{e}})$,
and $\mathrm{F}^\mathrm{B}$ for sepsets at the bottom of a branch
$(n_{\underline{e}},r_{\underline{e}})$.
The beauty of BP on a clique tree is that beliefs are guaranteed to converge to
the likelihood of the \emph{full} data, conditional on the state of the clique
variables.
After messages are passed down from the root to ${\mathcal C}_1$, the updated
belief of ${\mathcal C}_1$ will indeed be the likelihood of the full data
conditional on $n_{\overline{1}}$ and $r_{\overline{1}}$.

\subsection{\texttt{SnappNet}}

\texttt{SnappNet} \citep{2021Rabier-snappnet}
extends the model described in \texttt{SNAPP}
\citep{2012Bryant-snapp} to binary phylogenetic networks with reticulations.
We show here that our graphical model framework extends
to phylogenies with reticulations,
with additional edges in $G$ and hybridization factors to model the process at
hybrid nodes for the $n$ and $r$ variables.
The likelihood calculations used in \texttt{SnappNet} correspond to BP for
this graphical model.
Again, the graph $G$
is more complicated than the phylogeny.
See Fig.~\ref{fig:net1SM} for a 2-taxon phylogeny with 1 reticulation.
In addition to the
coalescent and speciation factors described in section~\ref{sec:snapp_BP},
we also need to describe hybridization factors. Consider an edge $e$ that is the child of a hybrid node,
whose parent hybrid edges $p_1$ and $p_2$ have inheritance probabilities
$\gamma_1$ and $\gamma_2$.
The hybridization factors for the total allele count
$\phi_{n_{\underline{p_1}}}=\mathbb{P}(n_{\underline{p_1}}\mid n_{\overline{e}})$
and
$\phi_{n_{\underline{p_2}}}=\mathbbm{1}_{\{n_{\overline{e}}-n_{\underline{p_1}}\}}(n_{\underline{p_2}})$
describe a binomial distribution for each $n_{\underline{p_i}}$ ($i=1,2$)
with $n_{\underline{p_1}} + n_{\underline{p_1}} =  n_{\overline{e}}$,
because each of the $n_{\overline{e}}$
individuals has a $\gamma_i$ chance of being assigned to edge $p_i$ ($i=1,2$).
The hybridization factor for the red allele count is simply
$\phi_{r_{\overline{e}}}=
\mathbbm{1}_{\{r_{\underline{p_1}}+r_{\underline{p_2}}\}}(r_{\overline{e}})$
because
$r_{\overline{e}}= r_{\underline{p_1}}+r_{\underline{p_2}}$.
\begin{figure}[h]
\begin{center}
  \begin{tikzpicture}[-]
    \matrix [matrix,anchor=north,
                column sep={1.2cm,between origins},
                row sep={0.5cm,between borders},
                nodes={font=\footnotesize,inner sep=1mm}] at (-5.5,0)
    {\node[label={[xshift=-2em]:\normalsize (a)}](){$N$}; & & & \\[-1.5em]
    & \node(x0){$\underline{\rho},\textcolor{red}{\underline{\rho}}$}; & & \\
    \node[yshift=-1em](x3){}; & & \node(x4){}; & \\
    \node(x1){$\underline{1},\textcolor{red}{\underline{1}}$}; & &
    \node(x2){$\underline{2},\textcolor{red}{\underline{2}}$}; & \\
    };
    \matrix [matrix,anchor=north,
                column sep={0.25cm,between borders},
                row sep={0.2cm,between borders},
                nodes={font=\footnotesize,inner sep=1mm}] at (0,0)
    {\node[label={[xshift=-2em]:\normalsize (b)}](){$G^{(s)}$};
    & \node(npb){$\underline{\rho}$}; & \node(rpb){$\textcolor{red}{\underline{\rho}}$}; & & \\
    \node(n5t){$\overline{5}$}; & \node(r5tr4t){$\textcolor{red}{(\overline{5},\overline{4})}$}; & \node(n4t){$\overline{4}$}; & & \\
    & & \node(r4b){$\textcolor{red}{\underline{4}}$}; & \node(n4b){$\underline{4}$}; & \\
    \node(n5bn3b){$(\underline{5},\underline{3})$}; & \node(r5b){$\textcolor{red}{\underline{5}}$}; & \node(n3t){$\overline{3}$}; & \node(r3tr2t){$\textcolor{red}{(\overline{3},\overline{2})}$}; & \node(n2t){$\overline{2}$}; \\
    \node(n1t){$\overline{1}$}; & \node(r1t){$\textcolor{red}{\overline{1}}$}; & \node(r3b){$\textcolor{red}{\underline{3}}$}; & \node(r2b){$\textcolor{red}{\underline{2}}$}; & \node(n2b){$\underline{2}$}; \\
    \node(n1b){$\underline{1}$}; & \node(r1b){$\textcolor{red}{\underline{1}}$}; & & & \\
    };

    \matrix [matrix,anchor=north,
                column sep={0.25cm,between borders},
                row sep={0.25cm,between borders},
                nodes={font=\footnotesize,inner sep=1mm}] at (5.5,0)
    {\node[label={[xshift=-2em]:\normalsize (c)}](){\hspace{-1em}$G$}; & & & & \\   
    \node(r5t_2){$\textcolor{red}{\overline{5}}$}; &
    \node(n5t_2){$\overline{5}$}; & \node(npb_2){$\underline{\rho}$}; &
    \node(rpb_2){$\textcolor{red}{\underline{\rho}}$}; & \node(n4t_2){$\overline{4}$}; &
    \node(r4t_2){$\textcolor{red}{\overline{4}}$}; \\
    \node(r5b_2){$\textcolor{red}{\underline{5}}$}; &
    \node(n5b_2){$\underline{5}$}; & \node(n3b_2){$\underline{3}$}; &
    \node(n3t_2){$\overline{3}$}; & \node(n4b_2){$\underline{4}$}; &
    \node(r4b_2){$\textcolor{red}{\underline{4}}$}; \\
    \node(r1t_2){$\textcolor{red}{\overline{1}}$}; &
    \node(n1t_2){$\overline{1}$}; &
    \node(r3b_2){$\textcolor{red}{\underline{3}}$}; &
    \node(r3t_2){$\textcolor{red}{\overline{3}}$}; &
    \node(n2t_2){$\overline{2}$}; & \node(r2t_2){$\textcolor{red}{\overline{2}}$}; \\
    \node(r1b_2){$\textcolor{red}{\underline{1}}$}; &
    \node(n1b_2){$\underline{1}$}; & & & \node(n2b_2){$\underline{2}$}; &
    \node(r2b_2){$\textcolor{red}{\underline{2}}$}; \\
    };

    \matrix [matrix,anchor=north,
                column sep={1.6cm,between borders},
                row sep={1.2cm,between borders},
                nodes={font=\small,fill=gray!20,inner sep=.5mm}] at (-2.12,-5)
    {\node[label={[xshift=-2em]:\normalsize (d)},fill=none]{$\mathcal{U}$}; & &
    \node[label=$\mathcal{C}_3$](c5){$\begin{matrix}\underline{5}, \ \underline{3}, \
      \overline{5}, \ \overline{3} \\
      \textcolor{red}{\underline{5}}, \
      \textcolor{red}{\underline{3}}, \
      \textcolor{red}{\overline{5}}, \
      \textcolor{red}{\overline{3}}\end{matrix}$}; &
    \node[label=$\mathcal{C}_2$](c3){$\begin{matrix}\overline{1}, \ \underline{5}, \
      \underline{3} \\ \textcolor{red}{\overline{1}}, \
      \textcolor{red}{\underline{5}}, \
      \textcolor{red}{\underline{3}}\end{matrix}$}; &[-2em]
    \node[label=$\mathcal{C}_1$](c1){$\begin{matrix}\underline{1}, \ \overline{1} \\
      \textcolor{red}{\underline{1}}, \ \textcolor{red}{\overline{1}}
    \end{matrix}$}; \\ 
    \node[label=$\mathcal{C}_7$](c7){$\begin{matrix}\overline{5}, \ \overline{4}, \
      \underline{\rho} \\ \textcolor{red}{\overline{5}}, \
      \textcolor{red}{\overline{4}}, \ \textcolor{red}{\underline{\rho}}
    \end{matrix}$}; &
    \node[label=$\mathcal{C}_6$](c8){$\begin{matrix}\overline{5}, \ \underline{4}, \
      \overline{4} \\ \textcolor{red}{\overline{5}}, \
      \textcolor{red}{\underline{4}}, \ \textcolor{red}{\overline{4}}
    \end{matrix}$}; &
    \node[label={[xshift=-1em]:$\mathcal{C}_5$}](c4){$\begin{matrix}\overline{5}, \
      \overline{3}, \ \overline{2}, \ \underline{4} \\
      \textcolor{red}{\overline{5}}, \ \textcolor{red}{\overline{3}}, \
      \textcolor{red}{\overline{2}}, \ \textcolor{red}{\underline{4}}
    \end{matrix}$}; & \node[label=$\mathcal{C}_4$](c2){$\begin{matrix}\underline{2}, \
      \overline{2} \\ \textcolor{red}{\underline{2}}, \
      \textcolor{red}{\overline{2}}\end{matrix}$}; & \\};

    \begin{scope}[every node/.style={font=\footnotesize\itshape}]
        \draw[->] (x3) -- node[above left,yshift=-.1em]
            {$\overline{1},\textcolor{red}{\overline{1}}$}(x1);
        \draw[blue,->] (x0) -- node[black,above left,yshift=.5em,xshift=1em]
            {$\overline{5},\textcolor{red}{\overline{5}}$}
            node[left,yshift=.2em]{\scriptsize $\gamma=0.6$}
            node[black,below left,xshift=-1.2em]{$\underline{5},
            \textcolor{red}{\underline{5}}$}(x3);
        \draw[blue,->] (x4) -- node[black,above,xshift=1em]
            {$\overline{3},\textcolor{red}{\overline{3}}$}
            node[black,below left,xshift=0em,yshift=0em]{$\underline{3},
            \textcolor{red}{\underline{3}}$}
            node[below right]{\scriptsize $\gamma=0.4$}(x3);
        \draw[->] (x4) -- node[above right,yshift=-.1em]
            {$\overline{2},\textcolor{red}{\overline{2}}$}(x2);
        \draw[->] (x0) -- node[above right,yshift=0em,xshift=-1em]
            {$\overline{4},\textcolor{red}{\overline{4}}$}
            node[below right,xshift=1em,yshift=0.3em]{$\underline{4},
            \textcolor{red}{\underline{4}}$}(x4);

        \draw[->,red] (npb) -- (r5tr4t);
        \draw[->,red] (rpb) -- (r5tr4t);
        \draw[->,red] (r5tr4t) -- (r4b);
        \draw[->,red] (r4b) -- (r3tr2t);
        \draw[->,red] (n4b) -- (r3tr2t);
        \draw[->,red] (r3tr2t) -- (r2b);
        \draw[->,red] (r3tr2t) -- (r3b);
        \draw[->,red] (r3b) -- (r1t);
        \draw[->,red] (r5b) -- (r1t);
        \draw[->,red] (r1t) -- (r1b);
        \draw[->,red] (r5tr4t) -- (r5b);
        \draw[->,red] (npb) -- (rpb);
        \draw[->,red] (n5t) -- (r5tr4t);
        \draw[->,red] (n4t) -- (r5tr4t);
        \draw[->,red] (n5t) -- (r5b);
        \draw[->,red] (n5bn3b) -- (r5b);
        \draw[->,red] (n4t) -- (r4b);
        \draw[->,red] (n4b) -- (r4b);
        \draw[->,red] (n2t) -- (r2b);
        \draw[->,red] (n2b) -- (r2b);
        \draw[->,red] (n3t) -- (r3tr2t);
        \draw[->,red] (n2t) -- (r3tr2t);
        \draw[->,red] (n3t) -- (r3b);
        \draw[->,red] (n1t) -- (r1b);
        \draw[->,red] (n1b) -- (r1b);
        \draw[->,red,bend left=5] (n5bn3b) to (r3b);
        \draw[<-] (npb) -- (n5t);
        \draw[<-] (npb) -- (n4t);
        \draw[<-] (n4t) -- (n4b);
        \draw[<-] (n4b) -- (n2t);
        \draw[<-] (n2t) -- (n2b);
        \draw[<-] (n4b) -- (n3t);
        \draw[<-,bend right=40] (n3t) to (n5bn3b);
        \draw[<-] (n5t) -- (n5bn3b);
        \draw[<-] (n5bn3b) -- (n1t);
        \draw[<-] (n1t) -- (n1b);

        \draw[->,red] (r4t_2) -- (r4b_2);
        \draw[->,red] (r5t_2) -- (r5b_2);
        \draw[->,red] (r2t_2) -- (r2b_2);
        \draw[->,red] (r1t_2) -- (r1b_2);
        \draw[->,red] (npb_2) -- (rpb_2);
        \draw[->,red] (n1t_2) -- (r1b_2);
        \draw[->,red] (n2t_2) -- (r2b_2);
        \draw[->,red] (n1b_2) -- (r1b_2);
        \draw[->,red] (n2b_2) -- (r2b_2);
        \draw[->,red,bend right=25] (rpb_2) to (r4t_2);
        \draw[->,red,bend left=20] (rpb_2) to (r5t_2);
        \draw[->,red] (n5t_2) -- (r5t_2);
        \draw[->,red,bend right=25] (npb_2) to (r5t_2);
        \draw[->,red,bend left=30] (r5t_2) to (r4t_2);
        \draw[->,red] (n5t_2) -- (r5b_2);
        \draw[->,red] (n5b_2) -- (r5b_2);
        \draw[->,red] (r5b_2) -- (r1t_2);
        \draw[->,red,bend left=50] (r3b_2) to (r1t_2);
        \draw[->,red] (n3b_2) -- (r3b_2);
        \draw[->,red] (n3t_2) -- (r3t_2);
        \draw[->,red] (r3t_2) -- (r3b_2);
        \draw[->,red] (n3t_2) -- (r3b_2);
        \draw[->,red] (n4t_2) -- (r4b_2);
        \draw[->,red] (n4b_2) -- (r4b_2);
        \draw[->,red] (n4b_2) -- (r3t_2);
        \draw[->,red,bend left=3] (r4b_2) to (r3t_2);
        \draw[->,red] (r4b_2) -- (r2t_2);
        \draw[->,red,bend right=50] (r3t_2) to (r2t_2);
        \draw[->] (n1b_2) -- (n1t_2);
        \draw[->] (n2b_2) -- (n2t_2);
        \draw[->] (n4b_2) -- (n4t_2);
        \draw[->] (n5b_2) -- (n5t_2);
        \draw[->] (n5t_2) -- (npb_2);
        \draw[->] (n1t_2) -- (n5b_2);
        \draw[->] (n1t_2) -- (n3b_2);
        \draw[->] (n3b_2) -- (n5b_2);
        \draw[->] (n2t_2) -- (n4b_2);
        \draw[->] (n3b_2) to (n3t_2);
        \draw[->] (n3t_2) to (n4b_2);
        \draw[->,bend right=30] (n4t_2) to (npb_2);
    \end{scope}
    \begin{scope}[every node/.style={font=\small\itshape,fill=orange!30,
      inner sep=.5mm,midway}]
        \draw[-] (c1) -- (c3) node[label={[yshift=1.1em]:$\mathcal{S}_{1,2}$}]
          {$\begin{matrix}\overline{1} \\
          \textcolor{red}{\overline{1}}\end{matrix}$};
        \draw[-] (c3) -- (c5) node[label={[yshift=1.1em]:$\mathcal{S}_{2,3}$}]
          {$\begin{matrix}\underline{5}, \
          \underline{3} \\ \textcolor{red}{\underline{5}}, \
          \textcolor{red}{\underline{3}}\end{matrix}$};
        \draw[-] (c5) -- (c4) node[label={[xshift=-2.3em,yshift=-.5em]:$\mathcal{S}_{3,5}$}]
          {$\begin{matrix}\overline{5}, \
          \overline{3} \\ \textcolor{red}{\overline{5}}, \
          \textcolor{red}{\overline{3}}\end{matrix}$};;
        \draw[-] (c2) -- (c4) node[label={[yshift=1.1em]:$\mathcal{S}_{4,5}$}]
          {$\begin{matrix}\overline{2} \\
          \textcolor{red}{\overline{2}}\end{matrix}$};
        \draw[-] (c4) -- (c8) node[label={[yshift=1.1em]:$\mathcal{S}_{5,6}$}]
          {$\begin{matrix}\overline{5}, \
          \underline{4} \\ \textcolor{red}{\overline{5}}, \
          \textcolor{red}{\underline{4}}\end{matrix}$};
        \draw[-] (c7) -- (c8) node[label={[yshift=1.1em]:$\mathcal{S}_{6,7}$}]
          {$\begin{matrix}\overline{5}, \
          \overline{4} \\ \textcolor{red}{\overline{5}}, \
          \textcolor{red}{\overline{4}}\end{matrix}$};
    \end{scope} 
  \end{tikzpicture}
  \end{center}
  \vspace{-1em}
  \caption{
  (a) Phylogenetic network $N$ with hybrid edges in blue.
  (b) and (c) show two valid graphs for the graphical model associated
  with $N$. Both are DAGs with 2 roots $\underline{1}$ and $\underline{2}$,
  and 2 leaves $\textcolor{red}{\underline{1}}$ and $\textcolor{red}{\underline{2}}$.
  In (b) the graph $G^{(s)}$ is symmetric and has a bivariate node
  for each of
  $(\textcolor{red}{\overline{5},\overline{4}})$,
  $(\textcolor{red}{\overline{3},\overline{2}})$, and
  $(\underline{5},\underline{3})$ to describe their conditional distributions.
  In (c), $G$ is asymmetric but has univariate nodes for each of
  $\textcolor{red}{\overline{5}},\textcolor{red}{\overline{4}},
  \textcolor{red}{\overline{3}},\textcolor{red}{\overline{2}},
  \underline{5},\underline{3}$.
(d) Clique tree $\mathcal{U}$ for $G$ (but not for $G^{(s)}$),
  with clusters $\mathcal{C}_j$ in grey
  and sepsets $\mathcal{S}_{i,j}$ in orange.
  To reduce clutter and simplify notations in this figure,
  $n_{\underline{e}}$ and $r_{\underline{e}}$ are both abbreviated as
  $\underline{e}$ and are distinguished by colours
  ($n$'s in black, $r$'s in red). Similarly,
  $n_{\overline{e}}$ and $r_{\overline{e}}$ are both denoted as
  $\overline{e}$ and distinguished by colours.
  }
  \label{fig:net1SM}
\end{figure}

We show that applying the \texttt{SnappNet} algorithm to the network in
Fig.~\ref{fig:net1SM}(a) is equivalent to BP on the clique tree in
Fig.~\ref{fig:net1SM}(d). To start, we assign the initial beliefs
$\phi_{n_{\overline{1}}}\phi_{r_{\underline{1}}}$ to cluster $\mathcal{C}_1$,
\, $\phi_{n_{\overline{2}}}\phi_{r_{\underline{2}}}$ to  $\mathcal{C}_4$, \,
$\phi_{n_{\underline{\rho}}} \phi_{r_{\underline{\rho}}}\phi_{r_{\overline{5}}}
\phi_{r_{\overline{4}}}$
to $\mathcal{C}_7$; \,
$\phi_{n_{\overline{5}}}\phi_{n_{\overline{3}}} \phi_{r_{\underline{5}}}\phi_{r_{\underline{3}}}$
to $\mathcal{C}_3$, \,
$\phi_{n_{\underline{4}}}\phi_{r_{\overline{3}}}\phi_{r_{\overline{2}}}$ to $\mathcal{C}_5$,
and
$\phi_{n_{\overline{4}}}\phi_{r_{\underline{4}}}$ to $\mathcal{C}_6$.
Finally, all hybridization factors
$\phi_{n_{\underline{5}}}\phi_{n_{\underline{3}}}\phi_{r_{\overline{1}}}$
are assigned to $\mathcal{C}_2$.
Each sepset is assigned an initial belief of 1.
The total allele counts $\mathrm{n}_{\underline{1}}$,
$\mathrm{n}_{\underline{2}}$ (fixed by design)
and the observed red allele counts
$\mathrm{r}_{\underline{1}}$, $\mathrm{r}_{\underline{2}}$ at the tips are
absorbed as evidence into the mutation factors $\phi_{r_{\underline{1}}}$,
$\phi_{r_{\underline{2}}}$ and coalescent factors $\phi_{n_{\overline{1}}}$,
$\phi_{n_{\overline{2}}}$ for the terminal edges. This is denoted as
$\phi[\cdot]$, with $[\cdot]$ containing the evidence absorbed.
BP messages are then passed on the clique tree from $\mathcal{C}_1$ and
$\mathcal{C}_4$ (considered as leaves) towards $\mathcal{C}_7$ (considered as
root) as follows:
\begin{equation*}
  \begin{split}
    \tilde{\mu}_{1\rightarrow 2} &= \phi_{n_{\overline{1}}}[\mathrm{n}_{\underline{1}}]
      \phi_{r_{\underline{1}}}[\mathrm{n}_{\underline{1}},\mathrm{r}_{\underline{1}}]=
      \mathrm{F}_{\overline{1}}(\mathcal{S}_{1,2}) \\
    \tilde{\mu}_{2\rightarrow 3} &= \sum_{\mathcal{C}_2\setminus
      \mathcal{S}_{2,3}}\phi_{n_{\underline{5}}}\phi_{n_{\underline{3}}}
      \phi_{r_{\overline{1}}}\cdot\tilde{\mu}_{1\rightarrow 2}=
      \mathrm{F}_{\underline{5},\underline{3}}(\mathcal{S}_{2,3}) \\
    \tilde{\mu}_{3\rightarrow 5} &= \sum_{\mathcal{C}_3\setminus
      \mathcal{S}_{3,5}}
      \phi_{n_{\overline{5}}}
      \phi_{n_{\overline{3}}}\phi_{r_{\underline{5}}}
      \phi_{r_{\underline{3}}}\cdot\tilde{\mu}_{2\rightarrow 3}=
      \mathrm{F}_{\overline{5},\overline{3}}(\mathcal{S}_{3,5}) \\
    \tilde{\mu}_{4\rightarrow 5} &= \phi_{n_{\overline{2}}}[\mathrm{n}_{\underline{2}}]
      \phi_{r_{\underline{2}}}[\mathrm{n}_{\underline{2}},\mathrm{r}_{\underline{2}}]=
      \mathrm{F}_{\overline{2}}(\mathcal{S}_{4,5}) \\
    \tilde{\mu}_{5\rightarrow 6} &= \sum_{\mathcal{C}_5\setminus
      \mathcal{S}_{5,6}}
      \phi_{n_{\underline{4}}}\phi_{r_{\overline{3}}}\phi_{r_{\overline{2}}}\cdot
      \tilde{\mu}_{3\rightarrow 5}\tilde{\mu}_{4\rightarrow 5}
      =\mathrm{F}_{\overline{5},\underline{4}}(\mathcal{S}_{5,6}) \\
    \tilde{\mu}_{6\rightarrow 7} &= \sum_{\mathcal{C}_6\setminus
      \mathcal{S}_{6,7}}
      \phi_{n_{\overline{4}}}
      \phi_{r_{\underline{4}}}\cdot\tilde{\mu}_{5\rightarrow 6}=
      \mathrm{F}_{\overline{5},\overline{4}}(\mathcal{S}_{6,7}) \\
    \beta_7^\text{final} &= \phi_{n_{\underline{\rho}}}\phi_{r_{\underline{\rho}}}
      \phi_{r_{\overline{5}}}\phi_{r_{\overline{4}}}\cdot
      \tilde{\mu}_{6\rightarrow 7} \\
    \sum_{\mathcal{C}_7}\beta_7^\text{final} &=
      \sum_{\mathcal{C}_7\setminus\mathcal{S}_{6,7}}\phi_{r_{\underline{\rho}}}
      \sum_{\mathcal{S}_{6,7}}\phi_{n_{\underline{\rho}}}\phi_{r_{\overline{5}}}
      \phi_{r_{\overline{4}}}\cdot\tilde{\mu}_{6\rightarrow 7}=\sum_{\mathcal{C}_7
      \setminus\mathcal{S}_{6,7}}\phi_{r_{\underline{\rho}}}
      \mathrm{F}_{\underline{\rho}}(\mathcal{C}_7\setminus\mathcal{S}_{6,7})
  \end{split}
\end{equation*}
where the $\mathrm{F}_{\mathbf{z}}$ functions are defined in
\citet{2021Rabier-snappnet} for different population interface sets $\mathbf{z}$
(each population interface is the top or bottom of some branch, e.g.\
$\overline{\rho}$, $\underline{\rho}$). That this correspondence does not always
hold (e.g.\ suppose the messages were sent towards $\mathcal{C}_1$ instead),
highlights that BP is more general. The $\mathrm{F}_{\mathbf{z}}$s
recursively compose the likelihood according to rules described in
\citet{2021Rabier-snappnet}, and can be expressed at the top-level as
$\sum_{n_{\underline{\rho}},r_{\underline{\rho}}}\phi_{r_{\underline{\rho}}}
\mathrm{F}_{\underline{\rho}}$. This is precisely the quantity from marginalizing
$\beta^\text{final}_7$ above, the final belief of $\mathcal{C}_7$.

\section{Bounding the moralized network's treewidth}

\begin{proof}[Proof of Proposition~1]
Using the notations in the main text,
let $N$ be as a binary phylogenetic network with $h$ hybrid nodes, level $\ell$,
no parallel edges and no degree-2 nodes other than the root.
Let $t$ be the treewidth of the moralized graph $N^m$ obtained from $N$.

(A0) is well-known: $t=1$ exactly when $N$ is a tree.
If $\ell=1$ then $N$ has at least one non-trivial blob and every such blob
is a cycle. So $N^m$ has outerplanar blobs and
$t=2$ \cite{biedl2015triangulating}, proving (A1).

Now consider hybrid nodes $v_1$ and $v_2$ as in (A2).
Let $u_3,u_4$ be the parents of $v_2$ such that $u_3$ is not a descendant of
$u_1$ or $u_2$
(see Fig.~4, in which $u_4=v_1$). Then there must be a path $p_4$ from $v_1$ to $v_2$ through $u_4$
since $v_2$ is a descendant of $v_1$.
For a directed path $p$, let $p^u$ denote the corresponding undirected path.
Let $w$ be a strict common ancestor of $u_1$ and of $u_2$ such that there exist
disjoint paths $p_1$, $p_2$, with $p_i$ from $w$ to $u_i$. Such $w$ exists because $u_1$ has a parent other than $u_2$ and vice versa.
Let $C$ be the cycle in $N^m$ formed by concatenating $p^u_1$, $p^u_2$ and the
moral edge $\{u_1,u_2\}$.
Next, pick any path $p$ from the root of $N$ to $u_3$.
If $p$ does not share any node with $C$, then we can find a common ancestor
$\tilde{w}$ of $w$ and $u_3$, and
paths $p_w$ and $p_3$ from $\tilde{w}$ to $w$ and $u_3$ respectively,
that do not intersect $p_1$ nor $p_2$. Then we can see that
$N^m$ contains the complete graph on $\{w,u_1,u_2,v_1\}$ as a graph minor,
by contracting $p^u_w + p^u_3 + \{u_3,v_2\} + p^u_4$ into a single edge
between $w$ and $v_1$.
If instead $p$ intersects $C$, then let $w'$ be the lowest node at which $p$
and $C$ intersect.
Then $w'\neq u_1$ because otherwise $u_3$ would be its descendant.
Similarly $w'\neq u_2$.
Let $p_3$ denote the subpath of $p$ from $w'$ to $u_3$.
Then $N^m$ contains the complete graph on $\{w',u_1,u_2,v_1\}$ as a graph minor,
as $C$ can be contracted into the cycle $\{w',u_1,u_2\}$ and
$p^u_3 + \{u_3,v_2\} + p^u_4$ can be contracted into an edge $\{w',v_1\}$.
In both cases, $N^m$ contains the complete graph on 4 nodes as
a graph minor, therefore its treewidth is $t \ge 3$ \cite{bodlaender1998partial}.
Also, in both cases $v_1$ and $v_2$ are in a common undirected
cycle in $N$, so in the same blob and $\ell\geq 2$.
\end{proof}

\clearpage

\section{Approximation quality with loopy BP}\label{sec:approx_loopy_BP}
\begin{figure}[!h]
  \centering
  \includegraphics[scale=0.65]{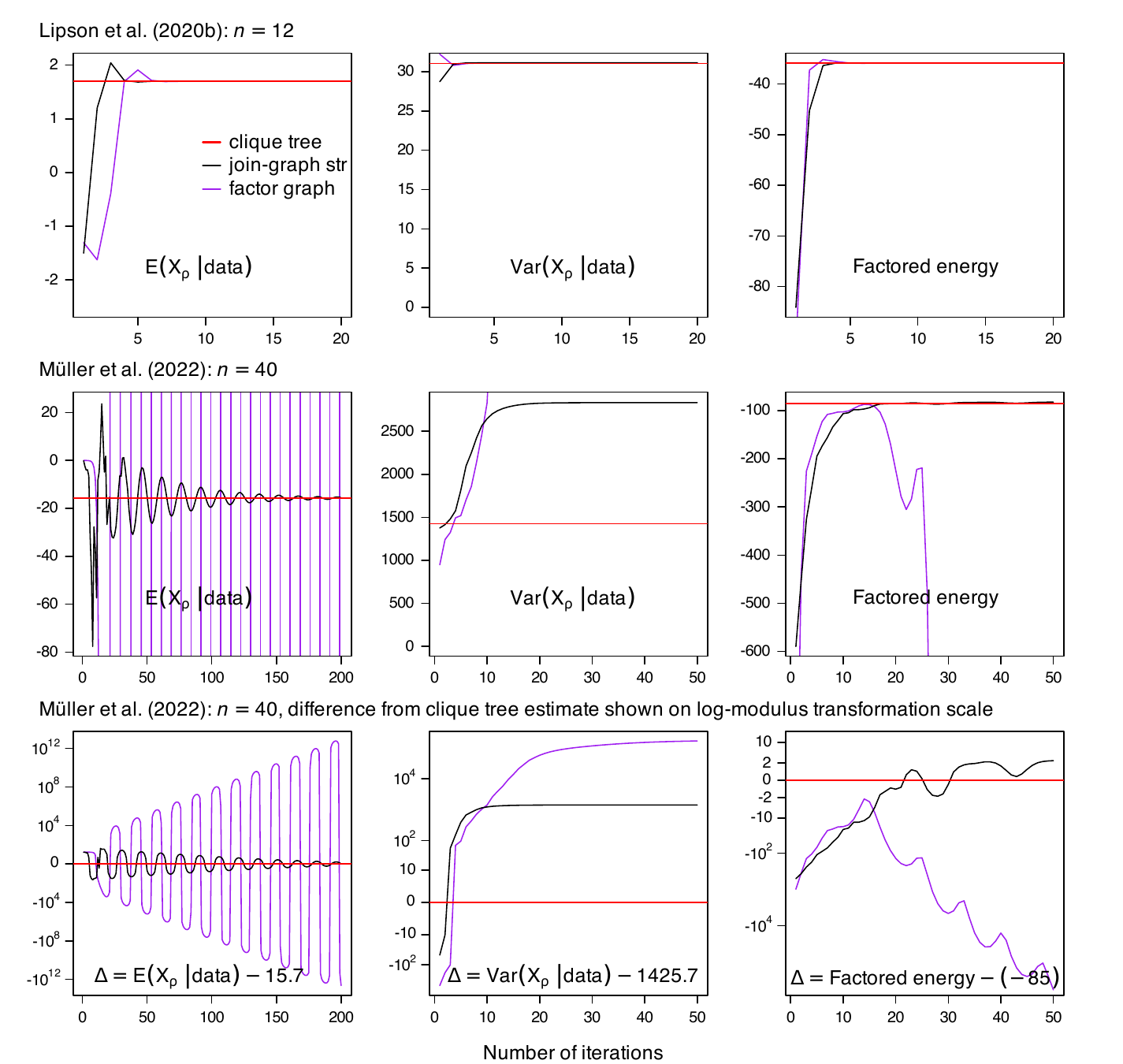}
  \caption{Comparing the accuracy of loopy BP between different cluster graphs:
  built from join-graph structuring $\mathcal{U}^*$ as in Fig.~6 (black),
  or a factor graph (purple).
  For both, initial beliefs are regularized using algorithm~\ref{alg:R2}.
  The true values, obtained using a clique tree, are shown in red.
  The plots in the first row are for the simpler phylogenetic network,
  and the plots in the other rows are for the complex phylogeny.
  The last row shows the difference $\Delta$ between the loopy BP estimate and
  the true value,
  displayed on the log-modulus scale using the transformation
  $\mbox{sign}(\Delta)\log(1+|\Delta|)$.
  For the simpler phylogenetic network, convergence speed and accuracy are similar between
  $\mathcal{U}^*$ and the factor graph, which is unsurprising given their similarly small cluster
  sizes ($\le 3$).
  For the complex network, the factor graph did not reach calibration as its
  iterates diverged for the conditional mean and factored energy.
  }\label{SMfig:loopyBPapproxfactorgraph}
\end{figure}

\begin{figure}[H]
\begin{minipage}[c]{0.45\textwidth}
  \centering
  \includegraphics[scale=0.5]{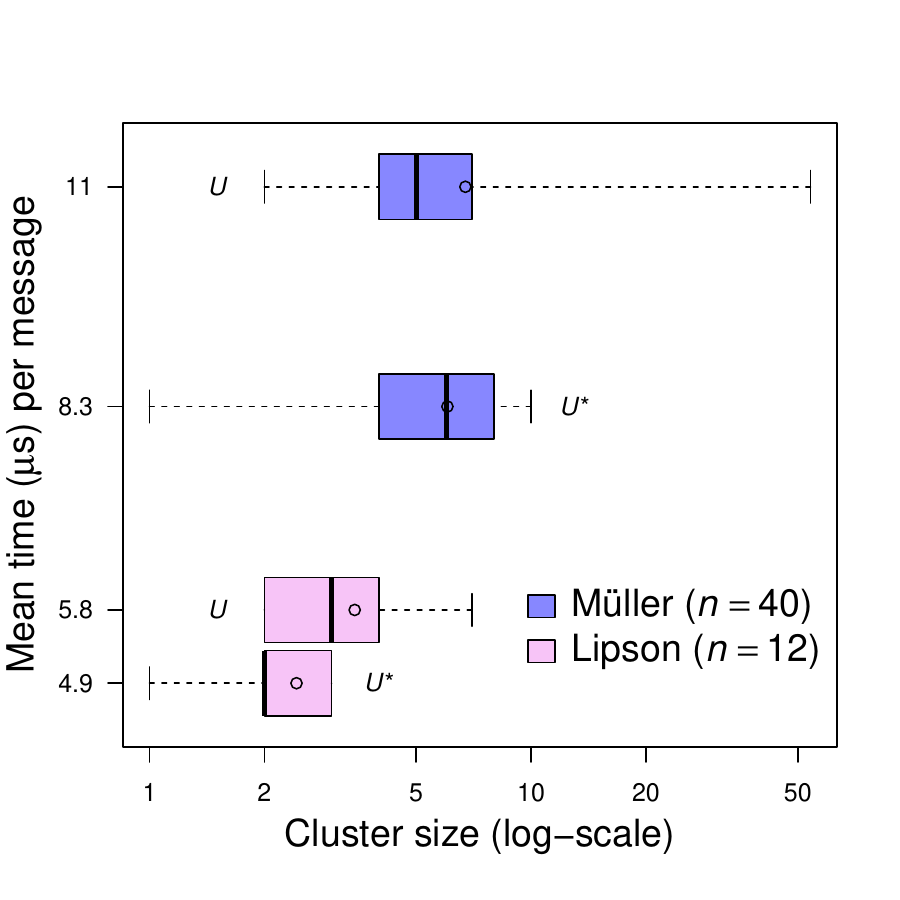}
\end{minipage}\hfill\begin{minipage}[c]{0.5\textwidth}
  \caption{Boxplots (with means as points) showing the distribution of
  cluster sizes in the join-graph structuring cluster graph
  $\mathcal{U}^*$ and in the clique tree $\mathcal{U}$ from Fig.~6. The factor graph has clusters of size between 1 and 3 (not displayed).
The time for 100 iterations (defined in Fig.~6) was benchmarked over 20 replicates on a MacBook Pro M2 2022,
and divided by the number of messages per 100 iterations to obtain an estimate
  of the mean time per belief update (vertical axis).
  }\label{SMfig:loopyBPapproxclustersizes}
\end{minipage}
\end{figure}

\section{Gradient and parameter estimates under the BM}\label{sec:gradient_BM}

\newcommand{\partialDerBis}[2]{\frac{\partial}{\partial #1} #2 }
\newcommand{\partialDerPrime}[2]{\left. \frac{\partial #2}{\partial #1'} \right|_{#1'=#1}}
\newcommand{\partialDerPrimeBis}[2]{\left. \frac{\partial}{\partial #1'} #2 \right|_{#1'=#1}}
\newcommand{\crossprod}[1]{\transpose{#1}#1}
\newcommand{\tcrossprod}[1]{#1\transpose{#1}}
\newcommand{\mtov}[0]{\mathrm{vec}}
\newcommand{\nodePos}[1]{#1}
\newcommand{\nodePosComplex}[1]{k_{\mathcal{C}}(#1)}
\newcommand{\nodePrecision}[1]{\branchLength{#1}^{-1} \varBMmm^{-1}}
\newcommand{\nodePrecisionPrime}[1]{\branchLength{#1}'^{-1} \varBMmm'^{-1}}
\newcommand{\nodeVariance}[1]{\branchLength{#1} \varBMmm}
\newcommand{\nodeMean}[1]{E_{#1}}
\newcommand{\nodeVarianceSmall}[1]{[\mathbf{J}_{#1}^{-1}]_{11}}
\newcommand{\nodeCovarianceSmall}[3]{[\mathbf{J}_{#1}^{-1}]_{\nodePos{#2}\nodePos{#3}}}
\newcommand{\nodeRemainder}[1]{\mathbf{g}_{#1}}
\newcommand{\nodeMeanDown}[1]{\nodeMean{#1}}
\newcommand{\nodeMeanDownTrans}[2]{\transpose{\nodeMeanDown{#1}}}
\newcommand{\nodeVarianceSmallDown}[1]{\nodeVarianceSmall{#1}}
\newcommand{\nodeCovarianceSmallDown}[3]{\nodeCovarianceSmall{#1}{#2}{#3}}
\newcommand{\gammauv}[1]{\gamma_{#1v}}
\def\allTraitOne{X}
\newcommand{\allTrait}[1]{\allTraitOne_{#1}}
\newcommand{\allTraitm}[1]{\allTraitOne_{#1}}
\newcommand{\allTraitPar}[1]{\allTraitOne_{\pa(#1)}}
\newcommand{\allTipTrait}[0]{Y}
\def\rootMean{\mu}
\def\rootMeanm{\mu}
\def\varBMmOne{\Sigma}
\def\varBMmm{\bm{\varBMmOne}}
\def\dimTrait{p}
\def\dimCluster{s}

\newcommand{\conditional}[2]{\left.#1\mathrel{}\middle|\mathrel{}#2\right.}
\newcommand{\Normal}[2]{\mathcal{N}\hspace*{-0.2em}\left(#1;\;#2\right)}
\newcommand{\densGaussian}[3]{\phi\left(#1; #2, #3\right)}
\newcommand{\vvect}[1]{#1}
\newcommand{\matr}[1]{\mathbf{#1}}
\newcommand{\kro}{\otimes}
\newcommand{\card}[1]{\left\lvert#1\right\rvert}
\newcommand{\cbigdot}{\bullet}

\def\tipTrait{Y}
\def\tipTraitVect{\vect{V}}
\def\fixedEffectsOne{\beta}
\def\fixedEffects{\vvect{\fixedEffectsOne}}
\def\regressor{\matr{U}}
\def\regressorTrans{\regressor^{\text{tree}}}
\def\phyloEffects{\vect{E}^{\mathrm{phy}}}
\def\phyloParams{\vvect{\nu}^{\text{phy}}}
\def\hatPhyloParams{\hat{\vvect{\nu}}^{\text{phy}}}
\def\BMEffects{\vect{E}^{\text{BM}}}
\def\iidEffects{\vect{E}^{\text{iid}}}
\def\errors{\vect{E}}
\def\nEffects{q}
\def\varBM{\sigma^2}

\subsection{The homogeneous BM model}

We consider here the simple case of a multivariate BM of dimension $\dimTrait$
on a network: with $\bm{V}_v=\branchLength{e} \varBMmm$ at
a tree node $v$ with parent edge $e$.
At a hybrid node, we assume
a weighted average merging rule as in (3.2) with a possible extra hybrid variance proportional to $\varBMmm$:
$\widetilde{\bm{V}}_h = \widetilde{\ell}(h) \varBMmm$
for some scalar $\widetilde{\ell}(h)\geq 0$.
To simplify equations, we define $\widetilde{\ell}(v)=0$ if $v$ is a tree node.
Then for each node $v \in V$ the Gaussian linear model~(3.1) simplifies to:
\begin{equation}\label{eq:BM_factors}
    \conditional{\allTraitm{v}}{\allTraitPar{v}} 
    \sim 
    \Normal{\sum_{u \in \pa(v)} \gammauv{u} \allTraitm{u}}{\branchLength{v} \varBMmm},
\end{equation}
with $\gammauv{u}$ the inheritance probability associated with the branch
going from $u$ to $v$, and
\[
  \branchLength{v} =
  \widetilde{\ell}(v) + \sum_{u \in \pa(v)} \gammauv{u}^2 \branchLength{e_{u \to v}}
  \;.
\]
At the root, we assume a prior variance proportional to $\varBMmm$:
$\allTraitm{\rho} \sim \Normal{\mu_\rho}{\branchLength{\rho} \varBMmm}$
which may be improper (and degenerate) with infinite variance
$\branchLength{\rho}=\infty$ or $\branchLength{\rho}=0$.

This model can also accommodate within-species variation,
by considering each individual as one leaf in the phylogeny, whose parent node corresponds to the species to
which the individual belongs. The edge $e$ from the species to the individual is assigned length
$\branchLength{e} = w$ and variance proportional to $\varBMmm$ conditional on the parent node
(species average): $\branchLength{e} \varBMmm$.
This model, then, assumes equal phenotypic (within species) correlation
and evolutionary (between species) correlation between the $\dimTrait$ traits.
The derivations below assume a fixed variance ratio $w$, to be estimated
separately.

All results in this section use this homogeneous BM model,
and make the following assumption.
\begin{assumption}\label{ass:nomissing}
  At each leaf, the trait vector (of length $\dimTrait$)
  is either fully observed or fully missing,
  i.e.\ there are no partially observed nodes.
\end{assumption}

\subsection{Belief Propagation}

Gaussian BP Algorithm~2 can be applied in the simple BM case
to get the calibrated beliefs, with two traversals of a clique tree
(or convergence with infinitely many traversals of a cluster graph).
The following result states that the conditional moments of all the nodes
obtained from this calibration have a very special form.
We will use it to derive analytical formulas for the maximum likelihood estimators
of the parameters of the BM.

\begin{proposition}\label{prop:conditionalmoments}
  Assume the homogeneous BM~\eqref{eq:BM_factors} and
  Assumption~\ref{ass:nomissing}.
  The expectation of the trait at each node conditional on the observed data
  does not depend on the assumed $\varBMmm$ parameter.
  In addition, the conditional variance matrix
  and the conditional covariance matrix of a node trait and any of its parent's
  is proportional to $\varBMmm$.
\end{proposition}

To prove this proposition, we need the following technical lemma,
which we prove later.

\begin{lemma}\label{lemma:kroneckerform}
  Consider a homogeneous $\dimTrait$-dimensional BM
  on a network and Assumption~\ref{ass:nomissing}.
  At each iteration of the
  calibration, each cluster and sepset of $s$ nodes
  has a belief whose canonical parameters
  are of the form:
  \begin{equation}\label{eq:kroneckerform}
    \matr{K} = \matr{J} \kro \varBMmm^{-1}
    \quad\mbox{ and }\quad
    h = (\matr{I}_\dimCluster \kro \varBMmm^{-1}) \begin{bmatrix}\vect{m}_1\\ \vdots \\\vect{m}_\dimCluster \end{bmatrix}
    =(\matr{I}_\dimCluster \kro \varBMmm^{-1}) \mtov(\matr{M}) = \mtov(\varBMmm^{-1} \matr{M})
  \end{equation}
  for some $\dimCluster\times \dimCluster$ symmetric matrix $\matr{J}$
  and
  vectors $\vect{m}_i$ ($i=1\ldots \dimCluster$) of size $\dimTrait$,
  where
  $\matr{M}$ is the $\dimTrait \times \dimCluster$ matrix
  with $\vect{m}_i$ on column $i$,
  and where $\mtov$ denotes the vectorization operation formed by stacking columns.
  Further, $\matr{M}$ depends linearly on the data $\allTipTrait$
  (from stacking the trait vectors at the tips) and
  $\mu_\rho$, separately across traits, in the sense that
  \begin{equation}\label{eq:kroneckerform-mi}
  \vect{m}_i = (w_i \kro \matr{I}_\dimTrait)
  \begin{bmatrix}\mu_\rho\\ \allTipTrait \end{bmatrix}.
  \end{equation}
  for some $1\times (n+1)$ vector of weights $w_i$.
  $\matr{J}$ and vectors $w_i$ ($i=1,\ldots,s$)
  are independent of the variance rate $\varBMmm$, the data $Y$ and $\mu_\rho$.
  They only depend on the network, the chosen cluster graph,
  the chosen cluster or sepset in this graph,
  and the iteration number.
\end{lemma}

\begin{proof}[Proof of Proposition~\ref{prop:conditionalmoments}]
  First note that, for any belief
  with form \eqref{eq:kroneckerform}, the mean of the
  associated normalized Gaussian distribution can be expressed as follows.
  Let the vector $\vvect{\mu}_j$ of size $\dimTrait$ be the
  mean for the node indexed $j$.
  Then
  \[
    \begin{bmatrix}\vvect{\mu}_1\\ \vdots \\\vvect{\mu}_\dimCluster \end{bmatrix} 
    = \matr{K}^{-1} h 
    = (\matr{J}^{-1} \kro \matr{I}_p) \begin{bmatrix}\vect{m}_1\\ \vdots \\\vect{m}_\dimCluster \end{bmatrix} 
    =(\matr{J}^{-1} \kro \matr{I}_p) \mtov(\matr{M}) = \mtov(\matr{M} \matr{J}^{-1})
    \,.
  \]
  Let $\matr{E}$ be the $\dimTrait\times s$ matrix of means
  with $\vvect{\mu}_j$ on column $j$.
  Then the expression above simplifies to
  \[\matr{E} = \matr{M} \, \matr{J}^{-1}\;.\]

  Assume Lemma \ref{lemma:kroneckerform},
  from which we re-use notations here.
  Let $\mathcal{C}$ be a cluster, and $\matr{K}$, $\matr{J}$ and
  $\matr{M}$ be its matrices from \eqref{eq:kroneckerform}.
  For any node $v$ in $\mathcal{C}$,
  let $\nodePosComplex{v}$ be the index of $v$
  in $\mathcal{C}$'s matrices.
 Then, writing
$[\matr{J}^{-1}]_{\cbigdot k}$ for the $k^{th}$ 
 column vector
 of $\matr{J}^{-1}$, we get:
  \begin{equation}\label{eq:conditional_expectations}
    \begin{aligned}
      \E\left[\cond{\allTraitm{v}}{\allTipTrait}\right] 
      = \mu_{\nodePosComplex{v}}
      = \matr{M} [\matr{J}^{-1}]_{\cbigdot\nodePosComplex{v}}
      = E_{v},
      \end{aligned}
  \end{equation}
  where $E_v$ denotes the column of $\matr{E}$ for node $v$ (i.e.\ the
  conditional expectation of its trait), and
  does not depend on $\varBMmm$. Assuming that calibration is reached, $E_v$ does not
  depend on the cluster $\mathcal{C}$ (or sepset) containing $v$.
  Further, note that~\eqref{eq:conditional_expectations} is exact on any cluster graph
  at calibration, not simply approximate,
  because we are using a Gaussian graphical model
  \citep{weiss1999correctness}.

  Similarly, for nodes $u,v$ in $\mathcal{C}$:
  \begin{equation}\label{eq:conditional_covariances}
    \begin{aligned}
    \var\left[\cond{\allTraitm{v}}{\allTipTrait}\right] 
    &= [\matr{K}^{-1}]_{\nodePosComplex{v}\nodePosComplex{v}}
    &&= [\matr{J}^{-1}]_{\nodePosComplex{v}\nodePosComplex{v}}\varBMmm,\\
    \cov\left[\cond{\allTraitm{v}, \allTraitm{u}}{\allTipTrait}\right]
    &= [\matr{K}^{-1}]_{\nodePosComplex{v}\nodePosComplex{u}}
    &&= [\matr{J}^{-1}]_{\nodePosComplex{v}\nodePosComplex{u}}\varBMmm.\\
    \end{aligned}
  \end{equation}
  Therefore, their conditional variances and covariances
  are proportional to $\varBMmm$.
\end{proof}

In the following, with a slight abuse of notation,
for any two nodes $u,v$ in $\mathcal{C}$, we will write
$[\matr{J}^{-1}]_{\nodePos{u}\nodePos{v}} = [\matr{J}^{-1}]_{\nodePosComplex{u}\nodePosComplex{v}}$
and
$[\matr{K}^{-1}]_{\nodePos{u}\nodePos{v}} = [\matr{K}^{-1}]_{\nodePosComplex{u}\nodePosComplex{v}}$
for the submatrices corresponding to the indices for $u$ and $v$
in $\mathcal{C}$.

Since~\ref{eq:conditional_expectations} requires inverting $\matr{J}$,
whose size $s$ depends on the cluster, calculating the conditional means $E_v$
has complexity $\mathcal{O}(s^3)$ typically.
As the $\matr{J}$ and $\matr{M}$ matrices appearing
in~\eqref{eq:conditional_expectations} and~\eqref{eq:conditional_covariances}
do not depend on $\varBMmm$ by Lemma~\ref{lemma:kroneckerform},
they can be computed by running BP with any $\varBMmm$ value,
and we have the following.

\begin{corollary}\label{cor:getJM-using-IdentifySigma}
  The $\matr{J}$ and $\matr{M}$ matrices in Lemma~\ref{lemma:kroneckerform},
  used
  in~\eqref{eq:conditional_expectations} and~\eqref{eq:conditional_covariances},
  are obtained as a direct output of BP
  using $\varBMmm = \matr{I}_p$ to calibrate the cluster graph.
\end{corollary}

Using~\eqref{eq:kroneckerform-mi} in Lemma~\ref{lemma:kroneckerform}
and the derivation of vectors $w_i$ at each BP update, given in the proof below,
we obtain the following result.

\begin{corollary}\label{cor:get-wi-foranytrait}
  For each cluster, the weights $w_i$ appearing in~\eqref{eq:kroneckerform-mi}
  can be obtained alongside BP for any trait
  using updates \eqref{eq:update_w_tips} and \eqref{eq:update_w_propagation} below,
  until convergence of all $w_i$ weight vectors and $\matr{J}$ matrices.
  These quantities can then be used to obtain conditional expectations and
  conditional (co)variances for any trait using~\eqref{eq:kroneckerform-mi}, \eqref{eq:conditional_expectations}
  and~\eqref{eq:conditional_covariances}.
\end{corollary}

For example, this result implies that
obtaining calibrated conditional expectations
for a large-dimensional trait can be done without handling large
$\dimTrait\times\dimTrait$ matrices: by first calculating the $w_i$'s and
$\matr{J}$ for each cluster until convergence,
and then re-using them repeatedly for each of the $\dimTrait$ traits
separately (without re-calibration).

\begin{proof}[Proof of Lemma~\ref{lemma:kroneckerform}]
  We now show that the properties stated in Lemma~\ref{lemma:kroneckerform} hold
  for each factor at initialization,
  and continue to hold after each step
  of Algorithm~2: belief initialization, evidence absorption, and propagation.

  \textit{Factor Initialization.}
  Using the notations from the main text, each factor
  $\phi_v(x_v\mid x_{\pa(v)})$
  has canonical form over its full scope:
  \[\phi_v(x_v\mid x_{\pa(v)}) = \text{C}\left(
      \begin{bmatrix}x_v \\ x_{\pa(v)}\end{bmatrix};\, \matr{K}_v, h_v, g_v\right).
  \]
  For any internal node or any leaf $v$ before evidence absorption,
  we get from (4.4) and \eqref{eq:BM_factors}:
  \begin{equation}\label{eq:canonhybrid}
      \matr{K}_v =
      \matr{J} \kro \varBMmm^{-1} \mbox{ and }
      h_v = \matr{0},
      \mbox{ where }
      \matr{J} = \frac{1}{\branchLength{v}} \begin{bmatrix} 1 & -\transpose{\vect{\gamma}} \\
        -\vect{\gamma} & \vect{\gamma}\transpose{\vect{\gamma}} \end{bmatrix}
  \end{equation}
  and
  \(
    \transpose{\vect{\gamma}} = (\gammauv{u}; u \in \pa(v)).
  \)
Hence all the node family factors have form \eqref{eq:kroneckerform}
  at initialization and neither $\matr{J}$ nor any $m_i=\bf{0}$
  depend on the data.
  We can initialize $w_i=\bm{0}$ for each $i$ in the factor's scope.

  At the root $\rho$, the formulas above still hold
  using that $\pa(\rho)$ is empty and $\vect{\gamma}$
  has length $0$,
  for any $0 < \branchLength{\rho} < +\infty$,
  and also for $\branchLength{\rho} = +\infty$ in which case
  $\matr{J} = [0]$, independent of the data. For $h_\rho$, we have
  $h_\rho = \frac{1}{\branchLength{\rho}} \varBMmm^{-1} \mu_\rho$,
  which satisfies \eqref{eq:kroneckerform} with $m_\rho= \mu_\rho/ \branchLength{\rho}$.
  $m_\rho$ is linear in $\mu_\rho$ and satisfies~\eqref{eq:kroneckerform-mi}
  with $w_\rho=(\frac{1}{\branchLength{\rho}},0,\dots,0)$
  (or simply $\bm{0}$ if $\branchLength{\rho} = +\infty$).
  If $\branchLength{\rho} = 0$, the root factor is not assigned
  to any cluster at initialization because it is instead handled during
  evidence absorption below. This is because $\branchLength{\rho} = 0$
  implies that $X_\rho$ is fixed to the $\mu_\rho$ value,
  and this is handled similarly to leaves fixed at their observed values.

  \textit{Belief Initialization.}
  Now consider a cluster $\mathcal{C}$ that is assigned factors
  $\phi_v$ for $k\geq 0$ nodes $\{v_1,\ldots,v_k\}$, by (4.2). Before this assignment, the belief for  $\mathcal{C}$ (and for all sepsets)
  is set to the constant function equal to 1,
  which trivially satisfies~\eqref{eq:kroneckerform}
  with $\matr{J}=\matr{0}$ and every $\vect{m}_i=\bf{0}$, and satisfies
  \eqref{eq:kroneckerform-mi} with every $w_i=\bm{0}$.
  To assign factor $\phi_{v_j}$ to $\mathcal{C}$,
  we first extend scope of $\phi_{v_j}$ to the scope of $\mathcal{C}$
  (re-ordering the rows and columns of $\matr{K}_{v_j}$ and $h_{v_j}$
  to match that of $\mathcal{C}$).
  We then multiply $\mathcal{C}$'s belief by the extended factor.
  So we now prove that \eqref{eq:kroneckerform}
  is preserved by these two operations: extension and multiplication.

  \textit{Belief Extension.}
  Consider extending the scope of a belief with parameters $(\matr{K},\vect{h},g)$
  satisfying \eqref{eq:kroneckerform} to include all $\dimTrait$ traits of an
  extra $(\dimCluster+1)^{\mathrm{th}}$ node.
  Without loss of generality, we assign these extra variables the last $\dimTrait$ indices.
  Then the canonical parameters of the extended belief can be written as
  \[\tilde{\matr{K}} =
    \begin{bmatrix}\matr{J} & \matr{0} \\\matr{0} & \matr{0} \end{bmatrix} \kro \varBMmm^{-1}
 \quad \mbox{and} \quad
 \tilde{h} = (\matr{I}_{\dimCluster+1} \kro \varBMmm^{-1})
   \begin{bmatrix}\vect{m}_1\\ \vdots \\\vect{m}_\dimCluster \\ \bm{0} \end{bmatrix}
  \]
  and continue to be of form \eqref{eq:kroneckerform}.
  $\tilde{\matr{K}}$ continues to be independent of $\varBMmm$ and of the data.
  All $m_i$ vectors involved in $\tilde{h}$ continue to be independent of
  $\varBMmm$ and linear in the data according to~\eqref{eq:kroneckerform-mi}:
  with $w_i$ unchanged for $i\leq s$ and $w_i=\bm{0}$ for $i=s+1$.

  \textit{Beliefs Product and Quotient.}
  Next, if $(\matr{K},\vect{h},g)$ and $(\matr{K}',\vect{h}',g')$ are
  the parameters of two
  beliefs on the same scope satisfying Lemma~\ref{lemma:kroneckerform},
  then their product also satisfies Lemma~\ref{lemma:kroneckerform}
  because the canonical form
  of the product has parameters $(\matr{K}+\matr{K}', \vect{h}+\vect{h}', g+g')$,
  and can be expressed with \eqref{eq:kroneckerform} using
  $\matr{J}+\matr{J}'$ and
  $m_j + m'_j$.
  \eqref{eq:kroneckerform-mi} continues to hold using weight vectors
  $w_j + w'_j$.
  Similarly, the ratio of the two beliefs has parameters
  $(\matr{K}-\matr{K}', \vect{h}-\vect{h}', g-g')$ and
  continues to satisfy Lemma~\ref{lemma:kroneckerform}.

  \textit{Evidence Absorption.}
  Assume that a belief
  satisfies Lemma~\ref{lemma:kroneckerform},
  and that we want to 
  absorb the evidence from one node $u$, $1 \leq u \leq \dimCluster$.
  This node $u$ can be a leaf, or the root if $\branchLength{\rho} = 0$.
  If $u=\rho$ then the data to be absorbed is $\mathrm{x}_u=\mu_\rho$.
  We need to express the canonical form of the belief as a function of
  $\vect{x}_{-u} = [\vect{x}_1^{\top}, \dotsc, \vect{x}_{u-1}^{\top}, \vect{x}_{u+1}^{\top}, \dotsc, \vect{x}_s^{\top}]^{\top}$ only,
  letting the data $\mathrm{x}_u$ appear in the canonical parameters.
  By Assumption~\ref{ass:nomissing},
  $\mathrm{x}_u$ is of full length $\dimTrait$, which maintains the block structure.
  We have:
  \begin{align*}
    \begin{bmatrix}\transpose{\vect{x}}_1 \cdots \transpose{\vect{x}}_s \end{bmatrix}
    (\matr{J} \kro \varBMmm^{-1})
    \begin{bmatrix}\vect{x}_1 \\ \vdots \\ \vect{x}_\dimCluster \end{bmatrix}
    &=
    \transpose{\vect{x}_{-u}}
    (\matr{J}_{-u} \kro \varBMmm^{-1})
    \vect{x}_{-u}
    +
    2 \sum_{t\neq u} \transpose{\mathrm{x}_u} J_{ut} \varBMmm^{-1} \vect{x}_t
    +
    \transpose{\mathrm{x}_u} J_{uu} \varBMmm^{-1} \mathrm{x}_u
    \\
    &=
    \transpose{\vect{x}_{-u}}
    (\matr{J}_{-u} \kro \varBMmm^{-1})
    \vect{x}_{-u}
    +
    2 \transpose{(\vect{J}_{-u,u} \kro \mathrm{x}_u)}
    (\matr{I}_{\dimCluster-1} \kro \varBMmm^{-1})
    \vect{x}_{-u}
    +
    \transpose{\mathrm{x}_u} J_{uu} \varBMmm^{-1} \mathrm{x}_u,
  \end{align*}
  where $\matr{J}_{-u}$ is the $(\dimCluster-1) \times (\dimCluster-1)$ matrix $\matr{J}$
  without the row and column for $u$;
  and $\vect{J}_{-u,u}$ is the $(\dimCluster-1) \times 1$ column vector of $\matr{J}$ for $u$
  without the row entry for $u$.
  Likewise:
  \begin{equation*}
    \begin{bmatrix}\transpose{\vect{m}}_1 \cdots \transpose{\vect{m}}_\dimCluster \end{bmatrix}
    (\matr{I}_\dimCluster \kro \varBMmm^{-1})
    \begin{bmatrix}\vect{x}_1 \\ \vdots \\ \vect{x}_\dimCluster \end{bmatrix}
    =
    \transpose{\vect{m}_{-u}}
    (\matr{I}_{\dimCluster-1} \kro \varBMmm^{-1})
    \vect{x}_{-u}
    +
    \transpose{\vect{m}}_u \varBMmm^{-1} \mathrm{x}_u,
  \end{equation*}
  where $m_{-u}$ is similarly defined as $x_{-u}$, so that the canonical
  form of the factor satisfies:
  \[
    \log\text{C}\left(
      \vect{x}_{-u};\, \matr{K}_{-u}, h_{-u}, g_{-u}
      \right)  
      =
    -\frac{1}{2} \transpose{\vect{x}_{-u}}
    (\matr{J}_{-u} \kro \varBMmm^{-1})
    \vect{x}_{-u}
    +
    \transpose{(\vect{m}_{-u} - \vect{J}_{-u,u} \kro \mathrm{x}_u)}
    (\matr{I}_{s-1} \kro \varBMmm^{-1})
    \vect{x}_{-u}
    + g_{-u},
  \]
  where $g_{-u}$ does not depend on $\vect{x}_{-u}$,
  \[
    \matr{K}_{-u} = (\matr{J}_{-u} \kro \varBMmm^{-1})
    \qquad \text{and} \qquad
    h_{-u} = (\matr{I}_{s-1} \kro \varBMmm^{-1}) (\vect{m}_{-u} - \vect{J}_{-u,u} \kro \mathrm{x}_u),
  \]
  have form~\eqref{eq:kroneckerform}. $\matr{J}_{-u}$
  continues to be independent of $\varBMmm$ and of the data and $\mu_\rho$. All $m_i$ vectors involved in $h_{-u}$ continue to be independent of
  $\varBMmm$ and linear in the data ---with linear dependence on $\mathrm{x}_u$
  introduced in this step.
  Namely, \eqref{eq:kroneckerform-mi} holds with weight vector $w_j$
  updated to:
  \begin{equation}\label{eq:update_w_tips}
    w_j - \vect{J}_{i,u} e_u\quad\mbox{ where }\quad e_u = (0,\dots,0,1,0,\ldots)
  \end{equation}
  is the basis (row) vector of $\mathbb{R}^{n+1}$ with coordinate 1 at the
  position indexing tip $u$.

  Note that, for a tip $v$ with data on all $\dimTrait$ traits,
  we recover (4.4) for the factor associated with
  the external edge to $v$, whose scope is reduced to
  $x_{\pa(v)}$ after absorbing the evidence from $\mathrm{x}_v$ (with $\dimCluster=2$):
\[
  \matr{K}_v = \frac{1}{\branchLength{v}} \varBMmm^{-1}, \quad
  h_v = \frac{1}{\branchLength{v}} \varBMmm^{-1} \mathrm{x}_v,\quad
  \mbox{ and }\quad
  g_v = -\frac{1}{2}\left(\log|2\pi\branchLength{v} \varBMmm| +
  \frac{1}{\branchLength{v}} \|\mathrm{x}_v\|^2_{\varBMmm^{-1}}\right)\;.
\]

  \textit{Propagation.}
  Next, we show that beliefs 
  continue to satisfy Lemma~\ref{lemma:kroneckerform}
  after any propagation step of Algorithm~2. The first propagation step
  consists of marginalizing a belief, to calculate the message
  $\tilde{\mu}_{i\rightarrow j}$ from cluster $i$ to cluster $j$.
  Suppose that a belief with
  parameters $(\matr{K},\vect{h},g)$ satisfies \eqref{eq:kroneckerform},
  and that we marginalize out all
  traits of one or more nodes in its scope.
  Let $I$ be the indices corresponding to nodes
  (or their traits, depending on the context, with some abuse of notation)
  to be marginalized and $S$ the indices corresponding to the remaining
  nodes (or their traits).
Then, the marginal belief has canonical parameters
  $(\tilde{\matr{K}},\tilde{\vect{h}})$ with:
  \begin{eqnarray*}
    \tilde{\matr{K}} &=& \matr{K}_{\mathrm{S}}
      -
      \matr{K}_{\mathrm{S},\mathrm{I}}
      \matr{K}_{\mathrm{I}}^{-1}
      \matr{K}_{\mathrm{I},\mathrm{S}}\\
    &=& J_{\mathrm{S}} \kro \varBMmm^{-1}
      -
      \left(J_{\mathrm{S},\mathrm{I}} \kro \varBMmm^{-1}\right)
      \left({J_{\mathrm{I}}}^{-1} \kro \varBMmm\right)
      \left(J_{\mathrm{I},\mathrm{S}} \kro \varBMmm^{-1}\right)\\
    &=& \left(
      J_{\mathrm{S}} - J_{\mathrm{S},\mathrm{I}}{J_{\mathrm{I}}}^{-1}J_{\mathrm{I},\mathrm{S}}
      \right) \kro \varBMmm^{-1}
    = \tilde{J} \kro \varBMmm^{-1}
  \end{eqnarray*}
  and
  \begin{eqnarray*}
    \tilde{h} &=& h_{\mathrm{S}}
      -
      \matr{K}_{\mathrm{S},\mathrm{I}}
      \matr{K}_{\mathrm{I}}^{-1}
      h_{\mathrm{I}}
    = h_{\mathrm{S}}
      -
      \left(J_{\mathrm{S},\mathrm{I}}\kro\varBMmm^{-1}\right)
      \left({J_{\mathrm{I}}}^{-1} \kro \varBMmm\right) h_{\mathrm{I}}
    = h_{\mathrm{S}}
      -
      \left(J_{\mathrm{S},\mathrm{I}}{J_{\mathrm{I}}}^{-1} \kro\bm{I}_p\right)
      h_{\mathrm{I}}
    =  \begin{bmatrix}\varBMmm^{-1}\tilde{\vect{m}}_1\\ \vdots \\\varBMmm^{-1}\tilde{\vect{m}_s} \end{bmatrix}
  \end{eqnarray*}
  where, for $j\in S$:
  \[\tilde{\vect{m}}_j =
    \vect{m}_j - \sum_{i\in I}
    \left(J_{\mathrm{S},\mathrm{I}}{J_{\mathrm{I}}}^{-1}\right)_{ji}\vect{m}_i.
  \]
  So~\eqref{eq:kroneckerform-mi} holds with updated weights:
  \begin{equation}\label{eq:update_w_propagation}
    \tilde{w}_j = w_j - \sum_{i\in I}
    \left(J_{\mathrm{S},\mathrm{I}}{J_{\mathrm{I}}}^{-1}\right)_{ji} w_i \,,
  \end{equation}
  and
  the marginalized belief (message) is still of the form \eqref{eq:kroneckerform}
  and continues to satisfy Lemma~\ref{lemma:kroneckerform}.
  The remaining propagation steps consist of
  dividing the message by the current sepset belief;
  extending the resulting quotient to the scope
  of the receiving cluster;
  and multiplying the receiving cluster's current belief with the
  extended quotient.
  Each of these steps was already proved to preserve
  the properties of Lemma~\ref{lemma:kroneckerform},
  therefore the receiving cluster's new belief still satisfies Lemma~\ref{lemma:kroneckerform}.
  The sepset belief does too because it is updated with the message that was passed.
\end{proof}

\subsection{Gradient computation and analytical formula for parameter estimates}\label{sec:gradient_formula_multivariate_BM}

\subsubsection{Gradients of factors}

When the factors are linear Gaussian as in~(3.1), their derivarive with respect to any vector of parameters $\theta$ can
be written as:
\begin{multline}\label{eq:derivative_gaussian}
\gradient{\theta} \left[
        \log \phi_v(\vect{X}_v | \vect{X}_{\pa(v)}, \theta)
    \right]
= 
    \partialDerBis{\theta}{\transpose{[\bm{q}_v \vect{X}_{\pa(v)} + {\omega}_v]}}
    \inverse{\bm{V}_v}
    \left(\vect{X}_v - \bm{q}_v \vect{X}_{\pa(v)} - {\omega}_v\right)
    \\
    +
    \frac12
    \partialDer{\theta}{\transpose{\mtovh(\inverse{\bm{V}_v})}}
    \mtovh
    \left(
      \bm{V}_v
        - \tcrossprod{(\vect{X}_v - \bm{q}_v \vect{X}_{\pa(v)} - {\omega}_v)}
    \right),
  \end{multline}
  where $\mtovh$ is the symmetric vectorization operation \citep{Magnus1986}.
  In the BM case, (3.1) simplifies to~\eqref{eq:BM_factors},
  so that, for non-root nodes:
  \begin{multline*}
\gradient{\theta} \left[
        \log \phi_v(\vect{X}_v | \vect{X}_{\pa(v)}, \theta)
    \right]
= 
    \partialDerBis{\theta}{\transpose{\left[\sum_{u \in \pa(v)} \gammauv{u} \allTraitm{p}\right]}}
    \nodePrecision{v}
    \left(\vect{X}_v - \sum_{u \in \pa(v)} \gammauv{u} \allTraitm{p}\right)
    \\
    +
    \frac12
    \partialDer{\theta}{\transpose{\mtovh(\nodePrecision{v})}}
    \mtovh
    \left(
      \nodeVariance{v}
        - \tcrossprod{\left(\vect{X}_v - \sum_{u \in \pa(v)} \gammauv{u} \allTraitm{p}\right)}
    \right),
  \end{multline*}
  and, for the root $\rho$, assuming $0 < \branchLength{\rho} < +\infty$,
  \begin{multline*}
\gradient{\theta} \left[
        \log \phi_{\rho}(\vect{X}_\rho | \theta)
    \right]
= 
    \partialDer{\theta}{\transpose{\left[\mu_{\rho}\right]}}
    \nodePrecision{\rho}
    \left(\vect{X}_\rho - \mu_{\rho}\right)
    \\
    +
    \frac12
    \partialDer{\theta}{\transpose{\mtovh(\nodePrecision{\rho})}}
    \mtovh
    \left(
      \nodeVariance{\rho}
        - \tcrossprod{\left(\vect{X}_\rho - \mu_{\rho}\right)}
    \right).
  \end{multline*}

  \subsubsection{Estimation of $\mu_{\rho}$}
  Note that $\mu_{\rho}$ has no impact on the model
  and needs not be estimated if
  $\branchLength{\rho}=\infty$ (improper flat prior).
  We assume here that $0 < \branchLength{\rho} < +\infty$,
  and will consider the case $\branchLength{\rho}=0$ later.
  Only the root factor depends on $\mu_{\rho}$.
  Taking its gradient with respect to $\mu_{\rho}$, we get:
  \begin{equation*}
\gradient{\mu_{\rho}} \left[
        \log \phi_{\rho}(\vect{X}_\rho | \theta)
    \right]
= 
    \nodePrecision{\rho}
    \left(\vect{X}_{\rho} - \mu_{\rho}\right).
  \end{equation*}
  To apply Fisher's formula~(6.1), we take the
  expectation $\E_\theta[\cond{\cbigdot}{\vect{Y}}]$
  of this gradient conditional on all the data $\vect{Y}$:
  \begin{equation*}
    \left.
    \gradient{\mu_{\rho}'} \left[
      \log p_{\theta'}(\vect{Y})
  \right]
  \right|_{\mu_{\rho}'=\mu_{\rho}}
  =
    \E_\theta\left[
      \cond{
    \left.
    \gradient{\mu_{\rho}'} \left[
        \log \phi_{\rho}(\vect{X}_\rho | \theta')
    \right]
    \right|_{\mu_{\rho}'=\mu_{\rho}}
      }{\vect{Y}}
    \right]
    = 
    \nodePrecision{\rho}
    \left(
      \E_\theta\left[\cond{\vect{X}_{\rho}}{\vect{Y}}\right]
      - \mu_{\rho}
      \right).
  \end{equation*}
  Setting this gradient to 0,
  we get:
  \begin{equation}\label{eq:mu_rho_hat}
    \hat{\mu}_{\rho} 
    = \E_\theta\left[\cond{\vect{X}_{\rho}}{\vect{Y}}\right]
    = \nodeMeanDown{\rho}
  \end{equation}
  where $\mathcal{C}$ is any cluster containing $\rho$ in its scope,
  and $\matr{J}$ and $\matr{M}$ are the matrices in Lemma~\ref{lemma:kroneckerform}
  for its belief.
  Note that by Lemma~\ref{lemma:kroneckerform},
  this estimate is independent of the assumed $\mathbf{\Sigma}$ used
  during calibration.
  This procedure corresponds to maximum likelihood estimation
  under the assumption that $\branchLength{\rho}$ is known.
  Under this model, $\mu_{\rho}$ represents the ancestral state at time
  $\branchLength{\rho}$ prior to the root node $\rho$, which is typically taken
  as the most recent common ancestor of the sampled leaves. This is equivalent
  to considering an extra root edge of length $\branchLength{\rho}$ above $\rho$,
  whose parent node has ancestral state $\mu_{\rho}$.
  Then $\hat{\mu}_{\rho}$ is a maximum likelihood estimate of
  the ancestral state at $\rho$,
  or an approximation thereof if a cluster graph is used instead of a clique tree.
  Note that, in a Bayesian setting, when fixing $\branchLength{\rho}$ to a given value,
  and fixing $\mu_{\rho} = 0$, this model can be seen as setting a
  Gaussian prior on the value at the root of the tree.
  This is the model used e.g.\ in BEAST \citep{Fisher2021}.

  \subsubsection{Estimation of $\mathbf{\Sigma}$} \label{si:sec:estimation_sigma}

  We now take the gradient with respect to the vectorized precision parameter
  $\vect{P} = \mtovh(\varBMmm^{-1})$, of length $p(p+1)/2$.
  For $v\neq \rho$, we get:
  \[ \gradient{\vect{P}} \left[
        \log \phi_v(\vect{X}_{\pa(v)} | \vect{X}_u, \theta)
    \right]
= \frac12
    \branchLength{v}^{-1}
    \mtovh
    \left(
      \nodeVariance{v}
        - \tcrossprod{\left(\vect{X}_v - \sum_{u \in \pa(v)} \gammauv{u} \allTraitm{u}\right)}
    \right)
  \] and, for the root $\rho$:
\begin{equation*}
\gradient{\vect{P}} \left[
      \log \phi_{\rho}(\vect{X}_\rho | \theta)
  \right]
= 
  \frac12
  \branchLength{\rho}^{-1}
  \mtovh
  \left(
    \nodeVariance{\rho}
      - \tcrossprod{\left(\vect{X}_\rho - \mu_{\rho}\right)}
  \right).
\end{equation*}
Applying again Fisher's formula~(6.1), we get:
  \begin{equation*}
    \left.
    \gradient{\vect{P}'} \left[
      \log p_{\theta'}(\vect{Y})
  \right]
  \right|_{\vect{P}'=\vect{P}}
    = 
    \frac12 
    \sum_{v\in V}
    \mtovh
    \left(
      \varBMmm - \branchLength{v}^{-1}\matr{F}_{v}
    \right),
\end{equation*}
where $\matr{F}_{v}$ is derived next, using that
$\E\left[\vect{Z}\transpose{\vect{Z}}\right] 
= 
\var\left[\vect{Z}\right] + \E\left[\vect{Z}\right]\transpose{\E\left[\vect{Z}\right]}$
and using~\eqref{eq:conditional_expectations} and~\eqref{eq:conditional_covariances}
on a cluster $\mathcal{C}$ containing $v$ and its parents in its scope,
with $\matr{J}_v$ and $\matr{M}_v$ from Lemma~\ref{lemma:kroneckerform} for $\mathcal{C}$.
For $v\neq \rho$, we get:
\begin{align*}
  \matr{F}_{v}
  &=
  \E_{\theta}\left[\cond{
    \tcrossprod{\left(\vect{X}_v - \sum_{u \in \pa(v)} \gammauv{u} \allTraitm{u}\right)}
  }{\vect{Y}}\right]
  \\
  &=
  \var_{\theta}\left[\cond{
    \vect{X}_v - \sum_{u \in \pa(v)} \gammauv{u} \allTraitm{u}
      }{\vect{Y}}
  \right]
  +
  \tcrossprod{
  \E_{\theta}\left[\cond{
    \vect{X}_v - \sum_{u \in \pa(v)} \gammauv{u} \allTraitm{u}
      }{\vect{Y}}\right]
      }
  \\
  &= 
  \left(
    \nodeCovarianceSmallDown{v}{v}{v}
    + \sum_{u_1,u_2 \in \pa(v)} \gammauv{u_1}\gammauv{u_2} \nodeCovarianceSmallDown{v}{u_1}{u_2}
    -2 \sum_{u \in \pa(v)} \gammauv{u} \nodeCovarianceSmallDown{v}{v}{u}
    \right) \varBMmm
  \\
  &\qquad\qquad
  + \tcrossprod{\left(\nodeMeanDown{v} - \sum_{u \in \pa(v)} \gammauv{u} \nodeMeanDown{u}\right)}.
\end{align*}
For the root $\rho$:
\begin{align*}
  \matr{F}_{\rho}
  &=
  \E_{\theta}\left[\cond{
    \tcrossprod{\left(\vect{X}_\rho - \mu_{\rho}\right)}
  }{\vect{Y}}\right]
  \\
  &=
  \var_{\theta}\left[\cond{
    \vect{X}_\rho - \mu_{\rho}
      }{\vect{Y}}
  \right]
  +
  \tcrossprod{
  \E_{\theta}\left[\cond{
    \vect{X}_\rho - \mu_{\rho}
      }{\vect{Y}}\right]
      }
  \\
  &=
  \nodeCovarianceSmallDown{\rho}{\rho}{\rho} \varBMmm
  + \tcrossprod{\left(\nodeMeanDown{\rho} - \mu_{\rho}\right)}.
\end{align*}
Setting this gradient to 0 with respect to $\varBMmm$,
we get the following maximum
likelihood estimate for the rate matrix:
\begin{multline}\label{eq:closed_form_BM_variance}
    \hat{\varBMmm}
    =
    \left[
      \begin{aligned}
      &\branchLength{\rho}^{-1} \tcrossprod{\left(\nodeMeanDown{\rho} - \hat{\mu}_{\rho}\right)}
      +
\sum_{v\in V, v \neq \rho}
    \branchLength{v}^{-1}
    \tcrossprod{\left(\nodeMeanDown{v} - \sum_{u \in \pa(v)} \gammauv{u} \nodeMeanDown{u}\right)}
      \end{aligned}
    \right]
    \\
    \times 
    \left[
      \sum_{v\in V}
      1 -
      \branchLength{v}^{-1}
      \left(
          \nodeCovarianceSmallDown{v}{v}{v}
      + \sum_{u_1,u_2 \in \pa(v)} \gammauv{u_1}\gammauv{u_2} \nodeCovarianceSmallDown{v}{u_1}{u_2}
      -2 \sum_{u \in \pa(v)} \gammauv{u} \nodeCovarianceSmallDown{v}{v}{u}
      \right)
    \right]^{-1}
  \end{multline}
  where we use the convention that a sum over an empty set (here $\pa(\rho)$) is $0$.

  Note that this formula only uses the calibrated moments computed at each cluster.
  After calibration, then, calculating $\hat{\varBMmm}$
  with~\ref{eq:closed_form_BM_variance} has complexity 
  $\mathcal{O}(|V|(k^3 + p^2))$
  where $k$ is the maximum cluster size,
  since~\ref{eq:closed_form_BM_variance} requires inverting at most $|V|$
  matrices of size $k\times k$ at most
  and the crossproduct of at most $|V|$ vectors of size $p$.
  The final product is a scalar scaling of a $p\times p$ matrix.
  Calibrating the clique tree or cluster graph is more complex,
  because each BP update has complexity up to $\mathcal{O}(k^3p^3)$.
  If the phylogeny is a tree, a clique tree has $k=2$ and $|V| = 2n-1$,
  so that~\ref{eq:closed_form_BM_variance}
  has complexity linear in the number of tips.
  While PIC can get these estimates in only one traversal of the tree,
  this formula requires two traversals of the clique tree,
  but is more general as it applies to any phylogenetic network.

  \subsubsection{ML and REML estimation}
  Restricted maximum likelihood (REML) estimation
  can be framed as integrating out fixed effects \citep{harville1974bayesian},
  here $\mu_\rho$,
  to estimate covariance parameters, here the BM variance rate $\mathbf{\Sigma}$.
  This model corresponds to placing an improper prior on the root using
  $\branchLength{\rho} = +\infty$, in which case $\mu_{\rho}$ is irrelevant.
  Then \eqref{eq:closed_form_BM_variance} remains valid
  (with vanishing terms for the root) and gives an analytical
  formula for the REML estimate of $\mathbf{\Sigma}$.

  For maximum likelihood (ML) estimation of
  $\mu_\rho$, considered as the state $X_{\rho}$ at the root node $\rho$,
  we need to consider the case $\branchLength{\rho} = 0$ to fix
  $X_{\rho}=\mu_\rho$.
  Under $\branchLength{\rho} = 0$, \eqref{eq:mu_rho_hat} cannot be calculated
because $X_\rho=\mu_{\rho}$ was absorbed as evidence
  and $\rho$ removed from scope.
Instead, we note that under an improper root with infinite variance,
  the posterior density of the root trait conditional on all the tips
  is proportional to the likelihood
  \begin{equation*}
      p(\cond{X_{\rho}}{\vect{Y}})
          \propto
          p(\cond{\vect{Y}}{X_{\rho}})
      \times
      p(X_{\rho})
  \end{equation*}
  because $p(X_{\rho})\equiv 1$ under an improper prior on $X_{\rho}$.
  Therefore, maximizing the likelihood 
  $p(\cond{\vect{Y}}{X_{\rho}})$
  in the root parameter $X_{\rho} = \mu_{\rho}$
  amounts to maximizing the density
  $p(\cond{X_{\rho}}{\vect{Y}})$ in $X_{\rho}$.
  This density
  is Gaussian with expectation $\nodeMeanDown{\rho}$
  by~\eqref{eq:conditional_expectations}
so its maximum is attained at
$\hat{\mu}_{\rho} = \E_\theta\left[\cond{\vect{X}_v}{\vect{Y}}\right] = \nodeMeanDown{\rho}$.
In summary, the ML estimate of $\mu_\rho$ is still given by
\eqref{eq:mu_rho_hat}, but calculated by running BP
under an \emph{improper} prior at the root.

\subsection{Analytical formula for phylogenetic regression}

In the previous section, we derived analytical formulas
\eqref{eq:mu_rho_hat} and \eqref{eq:closed_form_BM_variance}
for estimating the parameters of a homogeneous multivariate BM on a
phylogenetic network,
using the output of only one BP calibration thanks to
Corollary~\ref{cor:getJM-using-IdentifySigma}.

Instead of fitting a multivariate process, it is often of interest to look at the
distribution of one particular trait conditional on all others.
This phylogenetic regression setting is for instance used on a network in \citet{2018Bastide-pcm-net}.
Writing $\tipTraitVect$ the (univariate) trait of interest measured at the $n$ tips of a network,
and $\regressor$ the $p\times n$ matrix of regressors,
we are interested in the model:
\begin{equation}\label{si:eq:phylo_regression}
\tipTraitVect = \transpose{\regressor} \beta + \epsilon,
\end{equation}
with
$\beta$ a vector of $p$ coefficients, and
$\epsilon$ a vector of residuals with expectation $0$
and a variance-covariance matrix that is given
by a univariate BM on the network with variance rate $\sigma^2$,
and a root fixed to $0$.
The $v^\mathrm{th}$ column $\regressor_{\cbigdot v}$
corresponds to the predictors at leaf $v$ and will be denoted as $U_v$.

In this setting, explicit maximum likelihood estimators for $\beta$ and $\sigma$ 
are available, but they involve the inverse of an $n\times n$ matrix,
with $\mathcal{O}(n^3)$ complexity.
Our goal is to get these estimators in linear time.

\subsubsection{Parameter estimation using the joint distribution}

To build on section~\ref{sec:gradient_formula_multivariate_BM},
we first look at the joint distribution of the reponse and predictors $V$ and $U$.
Setting the intercept aside, we slightly rewrite model~\ref{si:eq:phylo_regression}
(with a slight change of notation for $\regressor$ and $p$) to:
\begin{equation}\label{si:eq:phylo_regression_intercept}
  V = \alpha\mathbf{1} + \transpose\regressor \beta + \epsilon,
\end{equation}
with
$\alpha$ a scalar, $\mathbf{1}$ the vector of ones,
$\beta$ a vector of $p$ coefficients, and
$\epsilon$ a vector of residuals with expectation $\mathbf{0}$
and a variance-covariance matrix given
by a univariate BM on the network with variance rate $\sigma^2$,
and a root fixed to $0$.
Assuming that the joint trait $X = (V,U)$,
of dimension $p+1$, is jointly Gaussian
and evolving on the network with variance rate
$\varBMmm_X = \begin{bmatrix}\varBMmm_{VV}&\varBMmm_{VU}\\\varBMmm_{UV}&\varBMmm_{UU}\end{bmatrix}$,
we obtain the regression model above
with
\begin{equation}
\label{eq:regressionfromcorrelation-jointgaussian-true}
\beta = \varBMmm_{UU}^{-1}\varBMmm_{UV} \quad\mbox{ and }\quad
\sigma^2=\varBMmm_{VV} - \varBMmm_{VU}\varBMmm_{UU}^{-1}\varBMmm_{UV}.
\end{equation}
This is because a joint BM evolution for $X$ implies that the evolutionary
changes in $V$ and $U$ along each branch $e$, $(\Delta V)_e$ and $(\Delta U)_e$,
are jointly Gaussian $\mathcal{N}(\bm{0},\ell(e)\varBMmm_X)$ and
independent of previous evolutionary changes.
By classical Gaussian conditioning, this means that
\[
  (\Delta V)_e = \transpose{(\Delta U)_e} \beta + (\Delta\epsilon)_e
\]
where $(\Delta\epsilon)_e \sim \mathcal{N}(0,\ell(e)\sigma^2)$ and
independent of $(\Delta U)_e$.
At a hybrid node, the merging rule holds for both $V$ and $U$
with the same inheritance weights, so by induction on the nodes (in preorder)
we get that $V_u = \alpha + \transpose{U_u} \beta + \epsilon_u$ at every node $u$
in the network, with $\alpha=V_\rho-\transpose{U_\rho} \beta$,
and with $\epsilon$ following a BM process with variance rate $\sigma^2$
starting at $\epsilon_\rho=0$.
Therefore~\ref{si:eq:phylo_regression_intercept} holds at the tips.

Consequently,
we can apply formulas~\eqref{eq:mu_rho_hat} and~\eqref{eq:closed_form_BM_variance}
to get maximum likelihood (or REML) estimates $\hat{\mu}_{X}$ and $\hat{\varBMmm}_{X}$
of the joint expectation and variance rate matrix of $X$.
We can then plug in these estimates
in~\ref{eq:regressionfromcorrelation-jointgaussian-true} to get:
\begin{equation}\label{eq:regressionfromcorrelation-jointgaussian}
  \hat{\alpha} = \hat{\mu}_V - \hat{\varBMmm}_{VU}\hat{\varBMmm}_{UU}^{-1} \hat{\mu}_U\;,\quad
  \hat{\beta} = \hat{\varBMmm}_{UU}^{-1}\hat{\varBMmm}_{UV}\;,\quad\mbox{and }
  \hat{\sigma}^2 = \hat{\varBMmm}_{VV} -
  \hat{\varBMmm}_{VU}\hat{\varBMmm}_{UU}^{-1}\hat{\varBMmm}_{UV}\;,
\end{equation}
where $\hat{\mu}_V$ and $\hat{\mu}_U$ are, respectively, the scalar
and vector of size $p$
extracted from $\hat{\mu}_X$ for traits $V$ and $U$, and, similarly,
$\hat{\varBMmm}_{VU}$, $\hat{\varBMmm}_{UV}$, $\hat{\varBMmm}_{VV}$ and $\hat{\varBMmm}_{UU}$ are the sub-matrices of
dimension $1\times p$, $p\times 1$, $1 \times 1$ and $p \times p$ extracted from $\hat{\varBMmm}_X$.
As calculating $\hat{\mu}_X$ and $\hat{\varBMmm}_X$
via~\eqref{eq:mu_rho_hat} and~\eqref{eq:closed_form_BM_variance}
has complexity $\mathcal{O}(|V|(k^3 + p^3))$ where $|V|$ is the number of nodes
in the network and $k$ is the maximum cluster size,
obtaining $\hat{\alpha}$, $\hat{\beta}$ and $\hat{\sigma}^2$
with~\ref{eq:regressionfromcorrelation-jointgaussian}
has that same complexity, which can be much smaller than $\mathcal{O}(n^3)$.
If the phylogeny is a tree, this complexity depends linearly on $n$.

\subsubsection{Direct parameter estimation using the marginal distribution}

Going back to model~\eqref{si:eq:phylo_regression},
we do not assume that $X=(V,U)$ is jointly Gaussian
and make no assumption about $U$.
The distribution assumption is solely on the residual $\epsilon$.
  Model~\ref{si:eq:phylo_regression} then
  amounts to a trait $Y^{\beta} = \tipTraitVect - \transpose{\regressor}\fixedEffects$
  at the tips (for a given $\fixedEffects$) evolving
  under a homogeneous univariate BM model with variance $\varBM$.
  We denote by $X$ the corresponding trait at all network nodes,
  whose values $Y^{\beta}$ at tips depends on $\beta$.

  We can apply Fisher's formula~(6.1) to this model, taking the derivative with respect to $\beta$:
  \begin{equation}
    \left.
      \gradient{\fixedEffects'} \left[
        \log p(\vect{\tipTraitVect - \transpose{\regressor}}\fixedEffects')
      \right]
      \right|_{\fixedEffects'=\fixedEffects}
      = 
      \sum_{v\in V}
      \E_{\theta}\left[
        \cond{
          \left.
          \gradient{\fixedEffects'} \left[
              \log 
              \phi_v(X_v | X_{\pa(v)}, \theta')
          \right]
          \right|_{\fixedEffects'=\fixedEffects}
      }{
        \tipTraitVect - \transpose{\regressor}\fixedEffects
      }
      \right]
  \end{equation}
  In this sum, the only factors $\phi_v$ that depend on $\fixedEffects'$
  are the factors at the tips. In phylogenies, leaves are typically
  constrained to have a single parent, although extending our derivation
  to the case of hybrid leaves would be straightforward. For a leaf $v$
  with parent $\pa(v)=\{u\}$, we have:
  $Y^{\beta'}_{v} | X_{\pa(v)}
  \sim
  \Normal{X_u}{\varBM \branchLength{v}}$,
  and
  $Y^{\beta'}_{v} = \tipTraitVect_v - \transpose{U_v}\fixedEffects'$,
  so that
  $\phi_v(Y^{\beta'}_v | X_{\pa(v)}, \fixedEffects')
  = \phi_v(\tipTraitVect_v | \transpose{U_v}\fixedEffects' + X_u, \fixedEffects')
  $.
  Using the Gaussian derivative formula~\eqref{eq:derivative_gaussian},
  we get:
  \begin{align*}
    \left.
    \gradient{\fixedEffects'} \left[
      \log 
      \phi_v(Y^{\beta'}_v | X_{\pa(v)}, \theta')
    \right]
    \right|_{\fixedEffects'=\fixedEffects}
      &= 
      \left.
      \gradient{\fixedEffects'} \left[
        \log 
        \phi_v(\tipTraitVect_v | \transpose{U}_v\fixedEffects' + X_u, \theta')
    \right]
    \right|_{\fixedEffects'=\fixedEffects}
    \\
    &=
      \partialDerPrime{\fixedEffects}{(X_u + \transpose{U}_v\fixedEffects')}^{\top}
      [\varBM \branchLength{v}]^{-1}
      \; \left(
        \tipTraitVect_v -
        \left(
        X_u + \transpose{U}_v\fixedEffects
        \right)
      \right)
      \\
      &= 
      U_v \; [\varBM \branchLength{v}]^{-1}
      \left(
        \tipTraitVect_v - \transpose{U}_v \fixedEffects -
        X_u
      \right),
  \end{align*}
Using Lemma~\ref{lemma:kroneckerform},
    the expectation of $X_u$ conditional on the observed values at the tips,
    $\E_{\theta}(\cond{X_u}{Y^{\beta}})$,
    is linear in the data
    so that by~\eqref{eq:conditional_expectations}:
    \[
      \E_{\theta}(\cond{X_u}{Y^{\beta}})
      = \E_{\theta}(\cond{X_u}{Y^{\beta}})^{\top}
      = [\matr{J}_{u}^{-1}]_{\nodePos{u}\cbigdot} (\matr{M}_{u}^{\tipTraitVect})^{\top}
      - [\matr{J}_{u}^{-1}]_{\nodePos{u}\cbigdot} (\matr{M}_{u}^{U})^{\top} \beta
      = E_u^{\tipTraitVect} - (E_u^{U})^{\top} \beta,
    \]
  where $\matr{M}_{u}^{\tipTraitVect}$, $E_u^{\tipTraitVect}$ and $\matr{M}_{u}^{U}$, $E_u^{U}$ denote,
  respectively, the BP quantities of Lemma~\ref{lemma:kroneckerform}
  when applied to the traits $\tipTraitVect$ and $U$ separately.
  Note that $\matr{M}_{u}^{U}$ can also be obtained by running BP on each of the
  $\dimTrait$ rows of $\regressor$ independently,
  because $\matr{J}_{u}$ does not depend on the data
  and $\matr{M}_{u}$ depends linearly on the data.
  As $\tipTraitVect$ is a trait of dimension $1$,
  $\matr{M}_{u}^{\tipTraitVect}$ is a row vector of size $s$,
  the number of nodes in the chosen cluster containing $u$;
  and $[\matr{J}_{u}^{-1}]_{\nodePos{u}\cbigdot}(\matr{M}_{u}^{\tipTraitVect})^{\top} = E_u^{\tipTraitVect}$
  is a scalar.
  Also, $[\matr{J}_{u}^{-1}]_{\nodePos{u}\cbigdot}(\matr{M}_{u}^{U})^{\top} = (E_u^{U})^{\top}$
  is a row-matrix of size $1 \times p$,
  so that
  $[\matr{J}_{u}^{-1}]_{\nodePos{u}\cbigdot}(\matr{M}_{u}^{U})^{\top}\beta = (E_u^{U})^{\top}\beta$
  is also a scalar.
  We can hence write, for leaf $v$ with parent $u$:
    \[ \E_{\theta}\left[
      \left.
    \gradient{\fixedEffects'} \left[
      \log 
      \phi_v(X_v | X_u, \theta)
  \right]
  \right|_{\fixedEffects'=\fixedEffects}
  \right]
      = 
      U_v \; [\varBM \branchLength{v}]^{-1}
      \left(
      \tipTraitVect_v - \transpose{U}_v\fixedEffects -
(E_u^{\tipTraitVect} - (E_u^{U})^{\top}\beta)
      \right).
    \] Taking the sum and cancelling the gradient in $\fixedEffects$, we get:
  \begin{equation}\label{eq:regressionfromcorrelation-direct}
  \widehat{\fixedEffects}
  =
  \inverse{
      \left(
      \sum_{\mathrm{leaf}\;v} 
      \frac{1}{\branchLength{v}} U_v
      (U_v - E_{\pa(v)}^{U})^{\top}
\right)
      }
  \sum_{\mathrm{leaf}\;v} 
  \frac{1}{\branchLength{v}} U_v
  (\tipTraitVect_v 
- E_{\pa(v)}^{\tipTraitVect}).
  \end{equation}
  Note that the first term of the product involves the inversion of
  a $p \times p$ matrix, and that this formula outputs
  a vector of size $p$.
To get all the quantities needed in this formula, we just need one
  BP calibration of the cluster graph with multivariate traits
  $(\tipTraitVect, U)$
  to get the conditional means and variances,
  which can be done efficiently using only univariate traits 
  thanks to Corollary~\ref{cor:get-wi-foranytrait}.
  
  Finally, to get an estimator of the residual variance $\sigma^2$,
  we can run another BP calibration, taking
  $\hat{\epsilon} = \tipTraitVect - \transpose{\regressor} \widehat{\fixedEffects}$
  as the tip trait values, and then use the formulas from the previous section.
  Using an infinite root variance for this last BP traversal gives us the 
  REML estimate of the variance.

  If the phylogeny is a tree,
  this algorithm involves $p+2$ univariate BP calibrations,
  each requiring two traversals of the tree, sums of $\mathcal{O}(n)$
  terms in~\ref{eq:regressionfromcorrelation-direct} and other formulas,
  and a $p \times p$ matrix
  inversion, so calculating $\widehat{\fixedEffects}$
  and $\hat{\sigma}^2$
  is linear in the number of tips.
  Comparatively, the algorithm used in the \texttt{R} package \texttt{phylolm}
  \citep{2014HoAne-algo} only needs one multivariate traversal of the tree.
  Our algorithm is more general however, as it applies to any phylogenetic network
  and to any associated cluster graph.

\section{Regularizing initial beliefs}
\label{sec:regularize_beliefs}

At initialization, each factor is assigned to a cluster whose scope includes
all nodes from that factor.
Then the initial belief $\beta_i$ of a cluster $\mathcal{C}_i$
is the product of all factors assigned to it
by (4.2). Sepsets are not assigned any factors
so their beliefs $\mu_{i,j}$ are initialized to 1.
This assignment guarantees that the
the joint density $p_\theta$ of the graphical model equals
the following quantity at initialization:
\begin{equation}\label{eq:graphinvariant}
  \frac{\prod_{\mathcal{C}_i\in\mathcal{V}^*}\beta_i}{
    \prod_{\{\mathcal{C}_i,\mathcal{C}_j\}\in\mathcal{E}^*}\mu_{i,j}} \;.
\end{equation}
\eqref{eq:graphinvariant} is called the \emph{graph invariant} because BP
modifies cluster and sepset beliefs without changing the value of this quantity,
and hence keeps it equal to $p_\theta$
\citep{koller2009probabilistic}.
Initialization with (4.2) can lead to degenerate messages, as highlighted in section~7(a).
However, other belief assignments are permitted,
provided that \eqref{eq:graphinvariant}
equals $p_\theta$ at initialization.
Modifying beliefs between BP iterations is also permitted,
provided that \eqref{eq:graphinvariant} is unchanged.

Regularization modifies the belief precisions to make them non-degenerate.
To maintain the graph invariant, every modification to a cluster belief
is balanced by a modification to an adjacent sepset belief.
We describe two basic regularization algorithms below, but many others
could also be considered.

\begin{algorithm}
  \caption{Regularization along variable subtrees} \label{alg:R1}
  \begin{algorithmic}[1]
    \ForAll{variable $x$}
    \State
    $\mathcal{T}_x \gets$ subtree induced by all clusters containing $x$
    \State fix $\epsilon>0$
    \For{all sepsets and \textbf{for} all but one cluster in $\mathcal{T}_x$}{}
    \State
    add $\epsilon$ to the diagonal entry of its belief's precision matrix
    corresponding to $x$
    \EndFor
    \EndFor
  \end{algorithmic}
\end{algorithm}
\begin{algorithm}
  \caption{Regularization on a schedule} \label{alg:R2}
  \begin{algorithmic}[1]
    \State Choose an ordering of clusters: $\mathcal{C}_1,\dots,\mathcal{C}_{|\mathcal{V}|}$
    \State For each cluster $\mathcal{C}_i$ and each neighbor $C_{j}$
    of $\mathcal{C}_{i}$, set $i\rightarrow j$ as unvisited
\ForAll{$i=1,\dots,|\mathcal{V}|$}
    \ForAll{neighbor $\mathcal{C}_j$ of $\mathcal{C}_i$}
    \If{$j\rightarrow i$ is unvisited}
    \State fix $\epsilon>0$
\State \label{alg:step1-cluster}
    add $\epsilon \bm{I}$ to the precision matrix of the sepset $\mathcal{S}_{i,j}$
    \State \label{alg:step1-sepset}
    add $\epsilon$ to the diagonal entry of $\mathcal{C}_i$'s precision
    matrix corresponding to each variable in $\mathcal{S}_{i,j}$
    \State mark $j\rightarrow i$ as visited
    \EndIf
    \EndFor
    \ForAll{neighbor $\mathcal{C}_k$ of $\mathcal{C}_i$}
    \If{$i\rightarrow k$ is unvisited}
    \State \label{alg:step2-BP}
    propagate belief from $\mathcal{C}_i$ to $\mathcal{C}_k$ by Algorithm 2
\State mark $i\rightarrow k$ as visited
    \EndIf
    \EndFor
    \EndFor
  \end{algorithmic}
\end{algorithm}

In Algorithm~\ref{alg:R1}, each modified belief is multiplied
by a regularization factor $\exp\left(-\frac{1}{2}\epsilon x^2\right)$.
The graph invariant is satisfied because
$\mathcal{T}_x$ must be a tree (by the running intersection property),
so the same number of clusters and sepsets are modified
and the regularization factors cancel out in \eqref{eq:graphinvariant}.
In Algorithm~\ref{alg:R2}, the same argument applies to
modifications on lines~\ref{alg:step1-cluster} and~\ref{alg:step1-sepset},
which cancel out in~\eqref{eq:graphinvariant}
so the graph invariant is maintained. It is also maintained on 
line~\ref{alg:step2-BP}, which uses BP.

The choice of the regularization constant $\epsilon$ is not
specified above, but should be adapted to the magnitude of
entries in the affected precision matrices.

Both algorithms performed comparably well on the join-graph structuring
cluster graphs used in Fig.~6 and Fig.~\ref{SMfig:loopyBPapproxfactorgraph}.
On the factor graph for the complex network, however,
Algorithm~\ref{alg:R2} was found to work better than \ref{alg:R1}.
Namely, beliefs remained persistently degenerate
after initial regularization with \ref{alg:R1}, such that the estimated
conditional means and factored energy could not be computed.

Both algorithms are illustrated in Figure~\ref{fig:algoR1R2}.

\begin{figure}[H]
  \centering
  \begin{tikzpicture}
    \matrix ()[matrix of math nodes,
            column sep={5.5cm,between origins},
            row sep={2cm,between origins},
            nodes={font=\small,inner sep=.5mm}] at (8,0)
      { & |(c6)| \begin{matrix}x_6 \\ x_8 \\ x_{12}\end{matrix}\!\!
      \begin{bmatrix}1 &\hspace{-1em} -1/2 &\hspace{-1em} -1/2 \\ -1/2 &\hspace{-1em} 1/4{\color{red}+3\tilde{\epsilon}-\epsilon} &\hspace{-1em} 1/4 \\
      -1/2 &\hspace{-1em} 1/4 &\hspace{-1em} 1/4\end{bmatrix} \\
      |(c10)| \begin{matrix}x_{10} \\ x_{11}\end{matrix}\!\!
      \begin{bmatrix}1{\color{red}+3\tilde{\epsilon}-\epsilon} &\hspace{-1em} -1 \\ -1 &\hspace{-1em} 1\end{bmatrix} &
      |(c8)| \begin{matrix}x_8 \\ x_{10}\end{matrix}\!\!
      \begin{bmatrix}1\textcolor{blue}{+2\epsilon} &\hspace{-1em} -1 \\ -1 &\hspace{-1em} 1{\color{blue}+\epsilon}\end{bmatrix} \\
      & |(c5)| \begin{matrix}x_5 \\ x_8\end{matrix}\!\!
      \begin{bmatrix}1 &\hspace{-1em} -1 \\ -1 &\hspace{-1em} 1{\color{red}+3\tilde{\epsilon}-\epsilon}\end{bmatrix} \\};
    \matrix ()[matrix of math nodes,
            column sep={4cm,between origins},
            row sep={2cm,between origins},
            nodes={font=\small,inner sep=.5mm}] at (0,0)
      { |(c6_2)| \begin{matrix}x_6 \\ x_8 \\ x_{12}\end{matrix}\!\!
      \begin{bmatrix}1 &\hspace{-1em} -1/2 &\hspace{-1em} -1/2 \\ -1/2 &\hspace{-1em} 1/4{\color{blue}+\epsilon} &\hspace{-1em} 1/4 \\
      -1/2 &\hspace{-1em} 1/4 &\hspace{-1em} 1/4\end{bmatrix} \\
      |(c8_2)| \begin{matrix}x_8 \\ x_{10}\end{matrix}\!\!
      \begin{bmatrix}1{\color{blue}+\epsilon} &\hspace{-1em} -1 \\ -1 &\hspace{-1em} 1\end{bmatrix} \\
      |(c5_2)| \begin{matrix}x_5 \\ x_8\end{matrix}\!\!
      \begin{bmatrix}1 &\hspace{-1em} -1 \\ -1 &\hspace{-1em} 1\end{bmatrix} \\};
    \begin{scope}[every node/.style={font=\small,midway,inner sep=.5mm}]
      \draw (c10) -- node[above,yshift=.1em](){$x_{10}\!\!\begin{bmatrix}
        0\color{blue}+\epsilon\color{red}+3\tilde{\epsilon}-\epsilon\end{bmatrix}$}(c8);
      \draw (c8) -- node[left](){$x_8\!\!\begin{bmatrix}{0\color{blue}+\epsilon}\color{red}+3\tilde{\epsilon}-\epsilon\end{bmatrix}$}(c5);
      \draw (c6) -- node[left](){$x_8\!\!\begin{bmatrix}{0\color{blue}+\epsilon}\color{red}+3\tilde{\epsilon}-\epsilon\end{bmatrix}$}(c8);
      \draw (c8_2) -- node[left](){$x_8\!\!\begin{bmatrix}{0\color{blue}+\epsilon}\end{bmatrix}$}(c5_2);
      \draw (c6_2) -- node[left](){$x_8\!\!\begin{bmatrix}{0\color{blue}+\epsilon}\end{bmatrix}$}(c8_2);
    \end{scope}
  \end{tikzpicture}
  \caption{Applying algorithms~\ref{alg:R1} and \ref{alg:R2} on the cluster graph from Fig.~3(d), for a univariate BM model with mean~0 and variance rate~1,
  edge lengths of 1 in the original network and inheritance
  probabilities of 0.5.
  Cluster/sepset precision matrices have rows labelled by variables to show
  the nodes in scope. Precision matrices show entries before regularization (black)
  and after one pass through the outermost loop of the algorithm
  (coloured adjustments).
  Left: regularization \ref{alg:R1} starting with variable $x_8$.
  Right: regularization \ref{alg:R2} starting with cluster $\{8, 10\}$,
  assuming that it is the first cluster scheduled to be processed.
  For \ref{alg:R2}, we differentiate the effects of lines 3-8 (blue) and lines 9-12 (red).
  For example, the resulting precision matrix for sepset $\{x_{10}\}$ is
  $[3\tilde{\epsilon}]$ after summing these effects, where
  $\tilde{\epsilon}=\epsilon+o(\epsilon)$.}\label{fig:algoR1R2}
\end{figure}

\section{Handling deterministic factors}

This section illustrates two approaches to running BP in the
presence of a deterministic Gaussian factor that arises because the state at
a hybrid node is a linear combination of its parents' states.

Let $X$ be a univariate continuous trait evolving on the 3-taxon network in
Fig.~2(a) (reproduced in Fig.~\ref{fig:net1zippeddown}(a)) under a BM model with ancestral state 0 at the root and
variance rate $\sigma^2$.
For simplicity, we assume that tree
edges have length 1, hybrid edges have length 0, and inheritance probabilities
are 1/2. The conditional distribution for each node given its parents is
non-deterministic and can be expressed in a canonical form, except for
$X_5$ at the hybrid node. Because of 0-length hybrid edges, we have the
deterministic relationship: $X_5 = (X_4+X_6)/2$.
After absorbing evidence $\mathrm{x}_1,\mathrm{x}_2,\mathrm{x}_3$ at
the tips and fixing $x_\rho=0$, the factors are:
\begin{align*}
\phi_1 &= \text{C}(x_4;\sigma^{-2},\sigma^{-2}\mathrm{x}_1,g_1) &
\phi_2 &= \text{C}(x_5;\sigma^{-2},\sigma^{-2}\mathrm{x}_2,g_2) &
\phi_3 &= \text{C}(x_6;\sigma^{-2},\sigma^{-2}\mathrm{x}_3,g_3) \\
\phi_4 &= \text{C}(x_4;\sigma^{-2},0,g_4) &
\phi_6 &= \text{C}(x_6;\sigma^{-2},0,g_6) &
\phi_5 &=\delta(x_5-(x_4+x_6)/2)
\end{align*}
where $g_i$ normalizes $\phi_i$ to a valid probability density
and $\delta(\cdot)$ denotes a Dirac distribution at 0.
$\mathcal{U}$ in Fig.~2(b) (reproduced in Fig.~\ref{fig:net1zippeddown}(b)) remains a valid clique tree for this model.
We index the cliques in
$\mathcal{U}$ as $\mathcal{C}_1=\{x_1,x_4\}$,
$\mathcal{C}_2=\{x_2,x_5\}$, $\mathcal{C}_3=\{x_3,x_6\}$,
$\mathcal{C}_4=\{x_5,x_4,x_6\}$, $\mathcal{C}_5=\{x_4,x_6,x_\rho\}$. We set
$\mathcal{C}_5$ as the root clique and assume the following factor assignment for
$\mathcal{U}$:
$\phi_1\mapsto\mathcal{C}_1$, $\phi_2\mapsto\mathcal{C}_2$,
$\phi_3\mapsto\mathcal{C}_3$, $\phi_5\mapsto\mathcal{C}_4$,
$\{\phi_4, \phi_6\} \mapsto\mathcal{C}_5$.

\subsection{Substitution}
\label{sec:degenerate-substitution}
The substitution approach removes the Dirac factor $\phi_5$ by
removing $x_5$ from the model,
substituting it by $(x_4+x_6)/2$ where needed.
Since $\phi_2$ has scope $\{x_5\}$,
it is reparametrized as $\phi_2'$ on scope $\{x_4,x_6\}$:
\begin{equation*}
  \phi_2'=\text{C}\left(\begin{bmatrix}x_4 \\ x_6\end{bmatrix};
  \frac{1}{4\sigma^{2}}\begin{bmatrix}1 & 1 \\ 1 & 1\end{bmatrix},
  \frac{\mathrm{x}_2}{2\sigma^{2}}\begin{bmatrix}1 \\ 1\end{bmatrix},g_2\right).
\end{equation*}
For the simple univariate BM, it is well known in the
admixture graph literature \citep{2012PickrellPritchard} that
this substitution corresponds to using a modified network $N'$
in which hybrid edges do not all have length~0
(Fig.~\ref{fig:net1zippeddown}(c)).
$N'$ is built from the original network $N$ by removing
the hybrid node~5 and connecting its parents (nodes 4 and 6) to its child (node 2)
with edges of lengths $\ell_4=\ell_6=2$ for example
(to ensure that $\gamma_4^2\ell_4 + \gamma_2^2\ell_2$
equals the length of the original child edge to node 2).
A clique tree  $\mathcal{U}'$ for $N'$ can be obtained from $\mathcal{U}$
by replacing $\mathcal{C}_2$ and $\mathcal{C}_4$ with
$\mathcal{C}_4'=\{x_2,x_4,x_6\}$ (Fig.~\ref{fig:net1zippeddown}d).
Factor $\phi_2'$ is
assigned to $\mathcal{C}_4'$ while the other factor assignments stay the same.

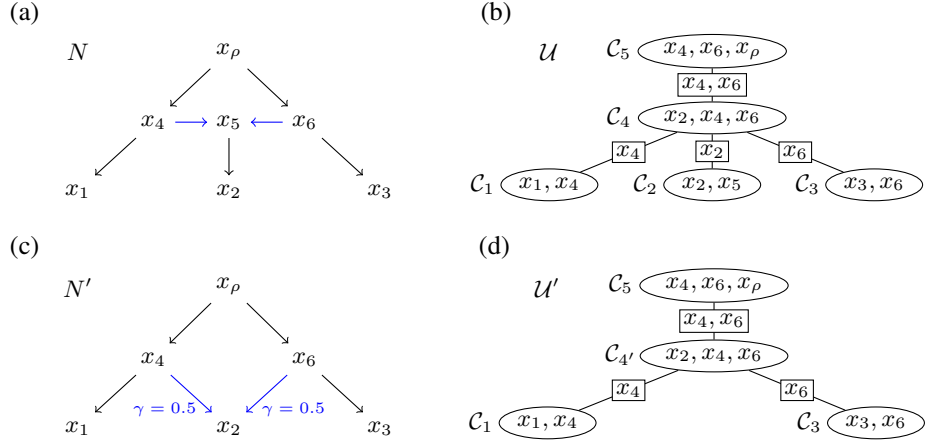
\begin{figure}[h]
  \begin{center}
      \begin{tikzpicture}[-]
      \matrix [matrix,anchor=north,
                  column sep={1cm,between origins},
                  row sep={0.5cm,between borders},
                  nodes={font=\small}] at (0,3.1)
      {\node[label={[xshift=-2em]:\normalsize (a)}](){$N$}; & &
      \node(x0_0){$x_\rho$}; & & \\
      & \node(x4_0){$x_4$}; & \node(x5_0){$x_5$}; & \node(x6_0){$x_6$}; & \\
      \node(x1_0){$x_1$}; & & \node(x2_0){$x_2$}; & & \node(x3_0){$x_3$}; \\
      };
      \matrix [matrix,anchor=north,
                  column sep={0cm,between borders},
                  row sep={.4cm,between borders},
                  nodes={font=\small}] at (6.5,3.1)
      {\node[label={[xshift=-2em]:\normalsize (b)}](){$\mathcal{U}$}; &
      \node[ellipse,draw,inner sep=.5mm,
        label={[xshift=-3.5em,yshift=-1.3em]:$\mathcal{C}_5$}]
        (c0_0){$x_4,x_6,x_\rho$}; & \\
        & \node[ellipse,draw,inner sep=.5mm,
          label={[xshift=-3.5em,yshift=-1.3em]:$\mathcal{C}_{4}$}]
          (c5_0){$x_2,x_4,x_6$}; & \\
      \node[ellipse,draw,inner sep=.5mm,
        label={[xshift=-2.5em,yshift=-1.3em]:$\mathcal{C}_1$}]
        (c1_0){$x_1,x_4$}; &
      \node[ellipse,draw,inner sep=.5mm,
        label={[xshift=-2.5em,yshift=-1.3em]:$\mathcal{C}_2$}]
        (c2_0){$x_2,x_5$}; &
      \node[ellipse,draw,inner sep=.5mm,
        label={[xshift=-2.5em,yshift=-1.3em]:$\mathcal{C}_3$}]
        (c3_0){$x_3,x_6$}; \\};
      \matrix [matrix,anchor=north,
                  column sep={1cm,between origins},
                  row sep={0.5cm,between borders},
                  nodes={font=\small}] at (0,0)
      {\node[label={[xshift=-2em]:\normalsize (c)}](){$N'$}; & &
      \node(x0){$x_\rho$}; & & \\
      & \node(x4){$x_4$}; & & \node(x6){$x_6$}; & \\
      \node(x1){$x_1$}; & & \node(x2){$x_2$}; & & \node(x3){$x_3$}; \\
      };
      \matrix [matrix,anchor=north,
                  column sep={0cm,between borders},
                  row sep={.4cm,between borders},
                  nodes={font=\small}] at (6.5,0)
      {\node[label={[xshift=-2em]:\normalsize (d)}](){$\mathcal{U}'$}; &
      \node[ellipse,draw,inner sep=.5mm,
        label={[xshift=-3.5em,yshift=-1.3em]:$\mathcal{C}_5$}]
        (c0){$x_4,x_6,x_\rho$}; & \\
      & \node[ellipse,draw,inner sep=.5mm,
          label={[xshift=-3.5em,yshift=-1.3em]:$\mathcal{C}_{4'}$}]
          (c5){$x_2,x_4,x_6$}; & \\
      \node[ellipse,draw,inner sep=.5mm,
        label={[xshift=-2.5em,yshift=-1.3em]:$\mathcal{C}_1$}]
        (c1){$x_1,x_4$}; & &
      \node[ellipse,draw,inner sep=.5mm,
        label={[xshift=-2.5em,yshift=-1.3em]:$\mathcal{C}_3$}]
        (c3){$x_3,x_6$}; \\};
      \begin{scope}[every node/.style={font=\small\itshape}]
          \draw[->] (x0_0) -- (x4_0); \draw[->] (x0_0) -- (x6_0);
          \draw[->] (x4_0) -- (x1_0);
          \draw[blue,->] (x4_0) -- node[below,xshift=-.2em,blue]{}(x5_0);
          \draw[black,->] (x5_0) -- (x2_0);
          \draw[blue,->] (x6_0) -- node[below,xshift=.2em]{}(x5_0);
          \draw[->] (x6_0) -- node[right]{}(x3_0);
          \draw[-] (c1_0) -- (c5_0) node[midway,rectangle,draw,fill=white,
              inner sep=.5mm]{$x_4$};
          \draw[-] (c3_0) -- (c5_0) node[midway,rectangle,draw,fill=white,
              inner sep=.5mm]{$x_6$};
          \draw[-] (c2_0) -- (c5_0) node[midway,rectangle,draw,fill=white,
              inner sep=.5mm]{$x_2$};
          \draw[-] (c5_0) -- (c0_0) node[midway,rectangle,draw,fill=white,
              inner sep=.5mm]{$x_4,x_6$};
          \draw[->] (x0) -- (x4); \draw[->] (x0) -- (x6);
          \draw[->] (x4) -- (x1);
          \draw[blue,->] (x4) -- node[below, xshift=-1em,blue]{\tiny $\gamma=0.5$}(x2); \draw[blue,->] (x6) -- node[below,
              xshift=1em]{\tiny $\gamma=0.5$}(x2); \draw[->] (x6) -- node[right]{}(x3);
\draw[-] (c1) -- (c5) node[midway,rectangle,draw,fill=white,
              inner sep=.5mm]{$x_4$};
          \draw[-] (c3) -- (c5) node[midway,rectangle,draw,fill=white,
              inner sep=.5mm]{$x_6$};
          \draw[-] (c5) -- (c0) node[midway,rectangle,draw,fill=white,
              inner sep=.5mm]{$x_4,x_6$};
      \end{scope} 
      \end{tikzpicture}
  \end{center}
  \caption{(a) Network $N$ from Fig.~2(a). (b) Clique tree $\mathcal{U}$ from
  Fig.~2(b). (c) Network $N'$ obtained by removing the hybrid node 5 from
  $N$ in (a).
  The BM model on $N$ leads to the same probability model for the nodes in $N'$
  as the BM model on $N'$, given a valid assignment of hybrid edge lengths in $N'$
  (see text). (d) Clique tree $\mathcal{U}'$ for $N'$ after moralization.}
  \label{fig:net1zippeddown}
\end{figure}

Standard BP can be used for the BM model on $N'$ because all factors
are non-degenerate.
After one postorder traversal of $\mathcal{U}'$, the
message $\tilde{\mu}_{4'\rightarrow 5}$, final belief $\beta_5$ and
log-likelihood $\loglik(\sigma^2)=\log p_{\sigma^2}(\mathrm{x}_1,\mathrm{x}_2,
\mathrm{x}_3)$ are:
\begin{equation*}
  \begin{split}
    \tilde{\mu}_{4'\rightarrow 5} &= \psi_{4'}\tilde{\mu}_{1\rightarrow 4'}
    \tilde{\mu}_{3\rightarrow 4'} = \phi_2'\phi_1\phi_3 = \text{C}\left(
      \begin{bmatrix}x_4 \\ x_6\end{bmatrix};\frac{1}{4\sigma^{2}}
      \begin{bmatrix}5 & 1 \\ 1 & 5\end{bmatrix},\frac{1}{2\sigma^{2}}
      \begin{bmatrix}2\mathrm{x}_1+\mathrm{x}_2 \\ 2\mathrm{x}_3+\mathrm{x}_2
      \end{bmatrix},\sum_{i=1}^3 g_i\right) \\
      \beta_5 &= \psi_5\tilde{\mu}_{4'\rightarrow 5} = \phi_4\phi_6
      \tilde{\mu}_{4'\rightarrow 5} = \text{C}\left(\begin{bmatrix}x_4 \\ x_6
      \end{bmatrix};
      \bm{K} = \frac{1}{4\sigma^{2}}\begin{bmatrix}9 & 1 \\ 1 & 9
      \end{bmatrix}, \,
      h = \frac{1}{2\sigma^{2}}\begin{bmatrix}2\mathrm{x}_1+
      \mathrm{x}_2 \\ 2\mathrm{x}_3+\mathrm{x}_2\end{bmatrix}, \,
      g = \sum_{i=1,i\neq 5}^6 g_i\right) \\
      \loglik(\sigma^2) &= \int\beta_5 dx_4dx_6 = \sum_{i=1,i\neq 5}^6 g_i
      +\left(\log\left|2\pi\left(\frac{1}{4\sigma^{2}}
      \begin{bmatrix}9 & 1 \\ 1 & 9\end{bmatrix}\right)^{-1}\right|
      +\left\|\frac{1}{2\sigma^{2}}\begin{bmatrix}2\mathrm{x}_1+\mathrm{x}_2 \\
      2\mathrm{x}_3+\mathrm{x}_2\end{bmatrix}\right\|_{\left(\frac{1}{4\sigma^{2}}
      \begin{bmatrix}9 & 1 \\ 1 & 9\end{bmatrix}\right)^{-1}}\right)/2
  \end{split}
\end{equation*}
We can still recover the conditional distribution of $X_5$ from
$\beta_5$, because $\beta_5$ has scope $\{x_4,x_6\}$.
Let $(\bm{K},h,g)$ be the parameters of the canonical form of $\beta_5$,
given above.
After the postorder traversal of $\mathcal{U}'$, $\beta_5$ contains information
from all the tips such that the distribution of $(X_4,X_6)$ conditional on
the data $(\mathrm{x}_1,\mathrm{x}_2,\mathrm{x}_3)$ is
$\mathcal{N}\left(\bm{K}^{-1}h,\bm{K}^{-1}\right)$.
Since
$\displaystyle X_5 = \gamma^{\top}\begin{bmatrix}X_4 \\ X_6\end{bmatrix}$
with $\gamma^{\top}=[1/2,1/2]$,
we get
\[
  X_5\mid(\mathrm{x}_1,\mathrm{x}_2,\mathrm{x}_3) \sim\mathcal{N}(\gamma^{\top}
  \bm{K}^{-1}h,\gamma^{\top}\bm{K}^{-1}\gamma) =
  \mathcal{N}\left(\sum_{i=1}^3\mathrm{x}_i/5,\sigma^2/5\right)\,.
\]

\subsection{Generalized canonical form}

A more general approach generalizes canonical form operations to include Dirac
distributions without modifying the original set of factors and clique tree,
as demonstrated in \citet{schoeman2022degenerate}.
Crucially, they derived message passing operations (evidence absorption,
marginalization, factor product, etc.) for a generalized canonical form:
\begin{equation*}
  \begin{split}
    \mathcal{D}(x;\bm{Q},\bm{R},\bm{\Lambda},h,c,g) &\coloneqq
    \text{C}(\bm{Q}^{\top}x;\bm{\Lambda},h,g)\cdot\delta(\bm{R}^{\top}x-c)
  \end{split}
\end{equation*}
where $x$ is an $n$-dimensional vector,
$\bm{\Lambda}\succeq 0$ is a $(n-k)\times(n-k)$ diagonal matrix, and
$\bm{Q}$ and $\bm{R}$ are matrices
of dimension $n\times (n-k)$ and $n\times k$ respectively,
that are orthonormal and orthogonal to each other, that is:
$\bm{Q}^{\top}\bm{Q} = I_{n-k}$,
$\bm{R}^{\top}\bm{R} = I_{k}$, and
$\bm{Q}^{\top} \bm{R} = \bm{0}$.
If $\bm{Q}$ is square (thus invertible) then $\bm{R}$ is empty
and the Dirac $\delta(\cdot)$ term is dropped or defined as 1.
The same applies to the $\text{C}(\cdot)$ term if
$\bm{R}$ is square (and $\bm{Q}$ is empty).
Non-deterministic linear Gaussian factors are represented in
generalized canonical form with $\bm{Q}$ square
from the eigendecomposition of $\bm{K} = \bm{Q}\bm{\Lambda}\bm{Q}^{\top}$.
Thus, we can run BP on
$\mathcal{U}$, converting beliefs or messages to generalized canonical form as
needed.

Running BP according to a postorder traversal of $\mathcal{U}$, the message
$\tilde{\mu}_{4\rightarrow 5}$
involves a degenerate component:
\begin{equation*}
  \begin{split}
    \tilde{\mu}_{4\rightarrow 5} &= \int\psi_4\prod_{i=1}^3
    \tilde{\mu}_{i\rightarrow 4}dx_5 = \int\phi_5\prod_{i=1}^3\phi_i dx_5 \;.
  \end{split}
\end{equation*}

To compute $\tilde{\mu}_{4\rightarrow 5}$, we first convert $\phi_5$ and
$\prod_{i=1}^3\phi_i$ to generalized canonical forms:
\begin{equation*}
  \begin{split}
    \phi_5 &= \mathcal{D}\left(\begin{bmatrix}x_4 \\ x_6 \\ x_5\end{bmatrix};
    \begin{bmatrix}1/w_1 & 1/w_2 \\ 1/w_1 & -1/w_2 \\ 1/w_1 & 0\end{bmatrix},
    \begin{bmatrix}1/2w_0 \\ 1/2w_0 \\ -1/w_0\end{bmatrix},
    \bm{\Lambda}_5=\begin{bmatrix}0&0\\0&0\end{bmatrix},
    \begin{bmatrix}0\\0\end{bmatrix},0,0\right) \\
    \prod_{i=1}^3\phi_i &=
\mathcal{D}\left(\begin{bmatrix}x_4 \\ x_6 \\ x_5\end{bmatrix};\bm{I}_3,-,
    \sigma^{-2}\bm{I}_3, \sigma^{-2}\begin{bmatrix}\mathrm{x}_1 \\ \mathrm{x}_3 \\
    \mathrm{x}_2\end{bmatrix},-,\sum_{i=1}^3g_i\right)
  \end{split}
\end{equation*}
where $(w_0,w_1,w_2)=(\sqrt{6/4},\sqrt{3},\sqrt{2})$ are normalization constants
for the respective columns, and the dashes indicate that the $\delta(\cdot)$
part is dropped. By \citet[Algorithm 3]{schoeman2022degenerate}, their product
evaluates to:
\begin{equation*}
  \phi_5\prod_{i=1}^3\phi_i = \mathcal{D}\left(\begin{bmatrix}x_4 \\ x_6 \\ x_5\end{bmatrix};
  \bm{Q}=\begin{bmatrix}1/w_1 & 1/w_2 \\ 1/w_1 & -1/w_2 \\ 1/w_1 & 0\end{bmatrix},
  \bm{R}=\begin{bmatrix}1/2w_0 \\ 1/2w_0 \\ -1/w_0\end{bmatrix},\sigma^{-2}\bm{I}_2,
  h=\sigma^{-2}\begin{bmatrix}(\mathrm{x}_1+\mathrm{x}_2+\mathrm{x}_3)/w_1 \\
  (\mathrm{x}_1-\mathrm{x}_3)/w_2\end{bmatrix},0,\sum_{i=1}^3g_i\right) \\
\end{equation*}
We partition
$\displaystyle \bm{Q}=\begin{bmatrix}\bm{Q}_{4,6}\\\bm{Q}_5\end{bmatrix}$
and
$\displaystyle \bm{R}=\begin{bmatrix}\bm{R}_{4,6}\\\bm{R}_5\end{bmatrix}$
by separating the first
two rows from the last row.\\
Then by \citet[Algorithm 2]{schoeman2022degenerate},
integrating out $x_5$ yields:
\begin{equation*}
  \int\mathcal{D}\left(\begin{bmatrix}x_4 \\ x_6 \\ x_5\end{bmatrix};
    \begin{bmatrix}\bm{Q}_{4,6} \\ \bm{Q}_5\end{bmatrix},
    \begin{bmatrix}\bm{R}_{4,6} \\ \bm{R}_5\end{bmatrix},\sigma^{-2}\bm{I}_2,
    h,0,\sum_{i=1}^3g_i\right)dx_5 =
    \mathcal{D}\left(\begin{bmatrix}x_4 \\ x_6\end{bmatrix};
    \bm{Q}_{4\rightarrow 5},-, \bm{\Lambda}_{4\rightarrow 5},
    h_{4\rightarrow 5},-, \sum_{i=1}^3g_i\right)
\end{equation*}
where, following notations from \citet{schoeman2022degenerate}
for intermediate quantities in their Algorithm 3:
\begin{equation*}
  \begin{split}
    \bm{U} &= \bm{I}_2, \bm{W}=[1] \\
\bm{F} &= (\bm{W}(\bm{R}_5\bm{W})^{+}\bm{Q}_5)^{\top}=
    \begin{bmatrix}-w_0/w_1 \\ 0\end{bmatrix} \\
    \bm{G} &= (\bm{Q}_{4,6}^{\top}-\bm{F}\bm{R}_{4,6}^{\top})\bm{U}=
    \begin{bmatrix}3/(2w_1) & 3/(2w_1) \\ 1/w_2 & -1/w_2\end{bmatrix} \\
    \bm{Z}\bm{\Lambda}_{4\rightarrow 5}\bm{Z}^{\top} &=
    \text{SVD}(\bm{G}^{\top}\,(\sigma^{-2}\bm{I}_2)\,\bm{G})=
    \text{SVD}\left(\frac{1}{4\sigma^{2}}\begin{bmatrix}5 & 1 \\ 1 & 5\end{bmatrix}\right) \\
    \text{therefore } \bm{Z} &=\begin{bmatrix}1/w_2 & 1/w_2 \\ 1/w_2 & -1/w_2\end{bmatrix}
    \text{ and }
    \bm{\Lambda}_{4\rightarrow 5}=\frac{1}{2\sigma^{2}}\begin{bmatrix}3 & 0 \\ 0 & 2\end{bmatrix} \\
    \bm{Q}_{4\rightarrow 5} &= \bm{U}\bm{Z} = \bm{Z} \\
    h_{4\rightarrow 5} &= \bm{Z}^{\top}\bm{G}^{\top}h =
    \sigma^{-2}\begin{bmatrix}w_2(\mathrm{x}_1+\mathrm{x}_2+\mathrm{x}_3) \\
    w_2(\mathrm{x}_1-\mathrm{x}_3)\end{bmatrix}
  \end{split}
\end{equation*}
This generalized canonical form can be rewritten as a standard canonical form:
\begin{eqnarray*}
  \tilde{\mu}_{4\rightarrow 5} &=&
  \text{C}\left(\bm{Q}_{4\rightarrow 5}^{\top}
  \begin{bmatrix}x_4 \\ x_6\end{bmatrix};\bm{\Lambda}_{4\rightarrow 5},
  h_{4\rightarrow 5}, \sum_{i=1}^3g_i\right) = \ \text{C}\left(\begin{bmatrix}x_4 \\ x_6\end{bmatrix};
  \bm{Q}_{4\rightarrow 5}\bm{\Lambda}_{4\rightarrow 5}
  \bm{Q}_{4\rightarrow 5}^{\top},\bm{Q}_{4\rightarrow 5}h_{4\rightarrow 5},
  \sum_{i=1}^3g_i\right) \\
  &=& \ \text{C}\left(\begin{bmatrix}x_4 \\ x_6\end{bmatrix};
  \frac{1}{4\sigma^{2}}\begin{bmatrix}5 & 1 \\ 1 & 5\end{bmatrix},
  \frac{1}{2\sigma^{2}}\begin{bmatrix}2\mathrm{x}_1+\mathrm{x}_2 \\
  2\mathrm{x}_3+\mathrm{x}_2\end{bmatrix},\sum_{i=1}^3g_i\right)
\end{eqnarray*}
which agrees with $\tilde{\mu}_{4'\rightarrow 5}$ from the substitution
approach in~\ref{sec:degenerate-substitution}.
Since the remaining operations to compute
$\beta_5= \psi_5\tilde{\mu}_{4\rightarrow 5}$ and
$\loglik(\sigma^2)=\int\beta_5 dx_4dx_6$
involve non-deterministic canonical forms,
it is clear that they both evaluate to the same quantity
as when using the substitution approach above.

\clearpage{}

\end{document}